\documentclass[11pt,a4paper]{scrartcl}
\usepackage[english]{babel}
\usepackage[left=2.5cm,right=2.5cm,top=2.5cm,bottom=2.5cm]{geometry}
\usepackage[OT1]{fontenc}
\usepackage{graphicx,amsthm,amssymb}
\usepackage[cmex10]{amsmath}
\usepackage{bm,dsfont,mathtools}
\usepackage{array,booktabs,adjustbox,threeparttable}
\usepackage[normalem]{ulem}
\usepackage[round,authoryear]{natbib}
\usepackage{xcolor}
\usepackage{auxhook}
\usepackage{caption}
\usepackage{authblk}
\usepackage[multiple]{footmisc}
\usepackage[colorlinks,citecolor=blue,urlcolor=blue]{hyperref}
\graphicspath{{figures/}}
\theoremstyle{plain}                 
\theoremstyle{plain}

\newtheorem{theorem}{Theorem}

\newtheorem{lemma}[theorem]{Lemma}

\theoremstyle{plain}
\newtheorem{definition}{Definition}

\newtheorem{remark}{Remark}

\newtheorem{corollary}{Corollary}
\newtheorem{proposition}{Proposition}
\newtheorem{assumption}{Assumption}

\makeatletter
\newcommand*\rel@kern[1]{\kern#1\dimexpr\macc@kerna}
\newcommand*\widebar[1]{%
  \begingroup
  \def\mathaccent##1##2{%
    \rel@kern{0.8}%
    \overline{\rel@kern{-0.8}\macc@nucleus\rel@kern{0.2}}%
    \rel@kern{-0.2}%
  }%
  \macc@depth\@ne
  \let\math@bgroup\@empty \let\math@egroup\macc@set@skewchar
  \mathsurround\z@ \frozen@everymath{\mathgroup\macc@group\relax}%
  \macc@set@skewchar\relax
  \let\mathaccentV\macc@nested@a
  \macc@nested@a\relax111{#1}%
  \endgroup
}
\makeatother

\newtheorem{ancillary}{Result}

\newcommand{\defeq}{\equiv}

\newcommand{\trans}{\prime}              

\newcommand{\indicator}{\mathds{1}}      
\newcommand{\dd}{\mathrm{d}}             
\newcommand{\idm}{\boldsymbol{I}}        
\newcommand{\mfzero}{\boldsymbol{0}}         
\newcommand{\rR}{\mathbb{R}}             
\newcommand{\rN}{\mathbb{N}}             
\newcommand{\var}{\operatorname{Var}}
\newcommand{\cov}{\operatorname{Cov}}

\newcommand{\wto}{\Rightarrow}
\newcommand{\op}{o_{\mathrm{P}}}
\newcommand{\Op}{O_{\mathrm{P}}}
\newcommand{\filtr}{\mathcal{F}}
\newcommand{\filtrn}{\mathcal{F}^{(\n)}}
\newcommand{\law}{\mathbb{P}}
\newcommand{\lawn}{\mathrm{P}^{(\n)}}
\newcommand{\Exp}{\mathbb{E}}
\newcommand{\prob}{\mathrm{Pr}}

\newcommand{\experiment}{\mathcal{E}}
\newcommand{\experimentn}{\mathcal{E}^{(\n)}}
\newcommand{\stat}{\mathfrak{t}}
\newcommand{\contrast}{\tau}
\newcommand{\Contrast}{\Gamma}
\newcommand{\vContrast}{\boldsymbol\Contrast}
\newcommand{\vphi}{\varphi}
\newcommand{\mPhi}{\bm{\Phi}}
\newcommand{\cv}{\kappa}

\newcommand{\power}{\Upsilon}
\newcommand{\cond}{\,|\,}

\newcommand{\pint}{\theta}
\newcommand{\pintb}{\boldsymbol{\pint}}
\newcommand{\Pintb}{\boldsymbol{\Theta}}
\newcommand{\lpint}{h}
\newcommand{\lpintb}{\boldsymbol{\lpint}}

\newcommand{\vmu}{\boldsymbol{\mu}}
\newcommand{\m}{m}
\newcommand{\vm}{\boldsymbol{\m}}
\newcommand{\mref}{\bar{\m}}
\newcommand{\dlpint}{\delta}

\newcommand{\vbeta}{\boldsymbol{\beta}}
\newcommand{\vzeta}{\boldsymbol{\zeta}}

\newcommand{\vb}{\boldsymbol{b}}

\newcommand{\vG}{\mathbf{G}}
\newcommand{\vOmega}{\boldsymbol{\Omega}}
\newcommand{\LLR}{\Lambda}
\newcommand{\LLRn}{\LLR^{(\n)}}
\newcommand{\LLRlim}{\Lambda}
\newcommand{\ii}{t}
\newcommand{\uu}{s}
\newcommand{\n}{T}
\newcommand{\psista}{\varphi}

\newcommand{\measure}{\nu}
\newcommand{\reward}{R}
\newcommand{\rewardb}{\boldsymbol{\reward}}
\newcommand{\rewardd}{\reward^{\dagger}}
\newcommand{\freq}{D}
\newcommand{\freqb}{\boldsymbol{\freq}}
\newcommand{\freqd}{\freq^{\dagger}}
\newcommand{\vr}{\boldsymbol{r}}
\newcommand{\vd}{\boldsymbol{d}}
\newcommand{\vS}{\boldsymbol{S}}

\newcommand{\vC}{\boldsymbol{C}}
\newcommand{\vs}{\boldsymbol{s}}

\newcommand{\vc}{\boldsymbol{c}}
\newcommand{\policy}{\psi}
\newcommand{\Epolicy}{\varpi}
\newcommand{\Epolicyb}{\boldsymbol{\Epolicy}}

\newcommand{\Y}{Y}
\newcommand{\Z}{Z}
\newcommand{\A}{A}
\newcommand{\X}{X}

\newcommand{\vA}{\boldsymbol{\A}}
\newcommand{\vX}{\boldsymbol{\X}}
\newcommand{\vx}{\boldsymbol{x}}

\newcommand{\Zkt}{\Z_{k,\ii}}

\newcommand{\ZAt}{\Z_{\A_{\ii},\ii}}

\newcommand{\Xks}{\X_{s}}

\newcommand{\SX}{\mathcal{X}}
\newcommand{\SK}{[K]}
\newcommand{\Sn}{[\n]}

\newcommand{\funqb}{\boldsymbol{v}^{\star}}
\newcommand{\funfb}{\boldsymbol{v}^{\circ}}
\newcommand{\parsumq}{V^{\star}}
\newcommand{\parsumf}{V^{\circ}}
\newcommand{\parsumqb}{\boldsymbol{V}^{\star}}
\newcommand{\parsumfb}{\boldsymbol{V}^{\circ}}
\newcommand{\Ubparsum}{\boldsymbol{U}}
\newcommand{\UbparsumE}{\Ubparsum^{\ast}}
\newcommand{\Ublim}{\tilde\Ubparsum}
\newcommand{\UblimE}{\tilde\Ubparsum^{\ast}}
\newcommand{\vu}{\boldsymbol{u}}
\newcommand{\vv}{\boldsymbol{v}}
\newcommand{\e}{\varepsilon}
\newcommand{\ekt}{\e_{k,\ii}}
\newcommand{\f}{f}
\newcommand{\fzk}{\f_{k}}

\newcommand{\W}{W}
\newcommand{\vW}{\boldsymbol{\W}}
\newcommand{\Wek}{\W_{\e_k}}
\newcommand{\B}{B}
\newcommand{\vB}{\boldsymbol{\B}}
\newcommand{\CS}{\Delta}
\newcommand{\CSb}{\boldsymbol{\CS}}
\newcommand{\QV}{\mathcal{Q}}
\newcommand{\QVb}{\boldsymbol{\QV}}
\newcommand{\FJ}{J}

\newcommand{\FJbpintk}{\boldsymbol{\FJ}_{\pintb,k}}

\newcommand{\score}{\dot\ell}
\newcommand{\scorek}{\dot\ell_{k}}

\newcommand{\scorebk}{\dot{\boldsymbol{\ell}}_{\pintb,k}}
\newcommand{\vP}{\boldsymbol{P}}

\newcommand{\ve}{\boldsymbol{e}}
\newcommand{\vSigma}{\boldsymbol{\Sigma}}
\newcommand{\s}{s}

\title{Valid Post-Contextual Bandit Inference\thanks{We thank Keisuke Hirano and participants at various seminars and conferences for their helpful feedback. The editor and three referees provided comments that significantly extended the results in the paper.}}
\date{October 26, 2026}
\author[1]{Ramon van den Akker}
\author[2]{Bas J.M. Werker}
\author[3]{Bo Zhou}
\affil[1]{Econometrics Group, Tilburg University}
\affil[2]{Econometrics and Finance Group, Tilburg University}
\affil[3]{Department of Economics, Virginia Tech}
\begin{document}
\maketitle
\begin{abstract}
Contextual multi-armed bandits (CMAB) are increasingly used as adaptively sampled experiments, after which researchers test hypotheses about treatment effects. We derive the H\'ajek-Le Cam limit experiment for multi-armed contextual bandits under equal-arms (weak-signal) asymptotics. It is Locally Asymptotically Quadratic, but not a (mixed) Gaussian shift experiment. Girsanov's theorem gives it a structural representation in stochastic differential equations. Together with a general weak convergence result, this yields the asymptotic null distribution and local asymptotic power of often-used test statistics. Two kinds of nuisance parameters threaten size control: the common reward level of the arms of interest and the rewards of the arms not of interest. The first threat is removed by translation-invariant sampling. We propose Regularized Translation-Invariant (RTI) versions of Thompson, tempered-greedy, tempered-(Lin)UCB, and exploration sampling. The second threat calls for appropriate pivots. The structural representation of the limit experiment is key to overcome the inferential challenges arising from adaptive sampling. Simulations corroborate the theory.
\end{abstract}



\section{Introduction} \label{sec:introduction}
Stochastic contextual multi-armed bandits (CMABs) are a cornerstone of adaptive randomized experimentation and sequential decision-making under uncertainty. In each round, an agent observes a context vector $\vX_\ii$, selects one of $K$ treatment arms, and observes the reward $\Y_\ii$ drawn from that arm’s context-dependent distribution only. CMABs underpin real-world systems in dynamic pricing, news and ad recommendation, online education, and mobile health,\footnote{\cite{bouneffouf2020survey} provides a review of applications.}
and are increasingly used to run field experiments in economics and political science (\citet{kasy2021adaptive}, \citet{OfferWestortEtAl2021}, \citet{CariaEtAl2024}).

The standard objective is to minimize (expected) cumulative regret, which yields the classic exploration-exploitation dilemma. Often used sampling schemes, or policies, include Thompson sampling (\cite{thompson1933likelihood}, \cite{agrawal2013thompson}), upper-confidence-bound (UCB) algorithms (\cite{lai1985asymptotically}, \cite{auer2002finite}, \cite{li2010contextual}), $\epsilon$-greedy/tempered-greedy heuristics, and the recent exploration sampling of \cite{kasy2021adaptive}. An extensive literature -- initiated by \cite{lai1985asymptotically} and surveyed by, amongst others, \cite{slivkins2019introduction} and \cite{lattimore2020bandit} -- focuses on studying the (asymptotic) regret properties of policies.\footnote{While these policies aim to optimize expected cumulative regret, \cite{fan2024fragility} and \cite{simchi2024simple} show that this can lead to heavy-tailed regret, motivating alternative designs that limit tail risk.}

Regret, however, is not the only objective. When a CMAB is used as an experiment its data are also the evidence on which decisions have been based: whether a new treatment should be adopted, by how much it outperforms the control, or whether any candidate intervention works at all (\citet{ImbensEtAl2026}). A policymaker scaling up a program, a firm deciding on product launch, or a health agency evaluating an intervention needs a valid test or confidence interval, not only a well-performing allocation during the experiment (see also \citet{hadad2021confidence} and \citet{bibaut2025demystifying}). Yet, the adaptivity that makes CMABs efficient -- in the sense of maximizing within-experiment cumulative rewards or minimizing expected regret -- complicates statistical inference. Bandit data are not independently and identically distributed (i.i.d.), which often invalidates standard properties of estimators and test statistics. For example, the sample mean of the observations collected on a treatment arm is, in general, a biased estimator of the mean of the arm and asymptotically has a non-Gaussian distribution; see \cite{deshpande2018accurate}, \cite{nie2018adaptively} and \cite{shin2019sample} for detailed discussions and illustrations. As a result, naive statistical inference ignoring the adaptive sampling generally is invalid.

A growing literature restores validity of post-experiment inference by reweighting the observations. \citet{hadad2021confidence} introduced adaptively weighted (AW) estimators whose studentized versions are asymptotically normal. \citet{zhang2020inference} and \citet{zhang2021statistical} developed related approaches for batched bandits and M-estimators.
\citet{bibaut2021post} extend this to contextual bandits: their contextual adaptive doubly robust (CADR) estimator of a policy value or treatment effect weights each doubly robust score by the inverse of a predictable estimate of its conditional standard deviation, a stabilization introduced by \citet{LuedtkeVanderlaan2016} for inference when the optimal treatment rule is not unique. CADR is asymptotically normal even when the outcome model is misspecified and exploration decays (see \citet{bibaut2025demystifying} for a survey).
These papers establish normality of particular statistics under conditions on the policy. They do not, however, provide a common framework in which, for example, the power of such statistics can be compared with each other. It should be noted that confidence sequences (\citet{howard2021time}, \citet{waudby2024time}, \citet{ImbensEtAl2026}) offer a different paradigm: optional stopping in combination with time-uniform guarantees on validity. This paper focuses on finite-horizon experiments.

To allow for power comparisons between tests and policies we adopt local asymptotics, i.e., arm gaps shrink with the horizon $T$.
\citet{FanGlynn2021} and \cite{kuang2024weak} show that, for gaps of the order $O(T^{-1/2})$, the arm-pull frequencies and (centered and scaled) cumulative rewards of Thompson sampling, and of a  broader class of sequentially randomized Markov experiments, converge to solutions of stochastic differential equations (SDEs).
Both papers develop their main results for a zero common mean of the arms. \citet{kuang2024weak} extend them, for bivariate bandits, to a nonzero common mean by replacing Thompson sampling with a translation-invariant version, and \citet{fan2025diffusion} does the same. That this replacement is needed shows that a zero common mean is not an innocent normalization.

We embed these diffusion limits via a limit experiment approach following the H\'{a}jek-Le Cam theory (see, for example, \cite{le1986asymptotic}, \cite{hajek1970characterization}, or \cite{van2000asymptotic}). Limit experiments have been used before to study adaptive experiments. \citet{adusumilli2021risk} derives diffusion asymptotics for optimal policies and \citet{adusumilli2023optimal} shows that the asymptotic power of any test after a sequential experiment can be matched in a limit experiment based on observing Gaussian processes, and derives optimal tests. \citet{hirano2025asymptotic} develop asymptotic representations for batched bandits and \citet{chen2023optimal} studies optimal conditional inference in batched experiments. We share this perspective, but consider fully sequential \emph{contextual} bandits, allow for widely used weighted statistics (that are not functions of the sufficient statistics of the limit experiment), and consider validity over composite null hypotheses.

As mentioned before, we use \emph{equal-arms} asymptotics: arm means, or arm-specific regression coefficients in the contextual case, are localized at rate $\n^{-1/2}$ around a common, unknown value. This is the natural regime for inference on treatment effects. The null hypothesis of no effect places the arms at equal means, so the size of a test is determined in this regime. The $\n^{-1/2}$-local alternatives are exactly those under which the signal accumulates at the same rate as the sampling noise, so
that power is nontrivial (i.e., we are considering statistically contiguous alternatives). It is also the regime in which the policy keeps sampling all arms with random, data-dependent frequencies throughout the experiment. Under fixed arm gaps (strong signal), by contrast, the policy focuses eventually on the best arm, the limiting allocation is deterministic, and inference on the arms is asymptotically standard, albeit at slower rates for the inferior arms. The nonstandard behavior is thus caused not (necessarily) by undersampling of inferior arms, but by the randomness of the limiting allocation and its dependence on the realized rewards. A policy that deliberately oversamples inferior arms, such as exploration sampling, does not restore standard inference. Only a \emph{deterministic} limiting allocation, as under non-adaptive designs or fixed exploration schedules, does (see Remark~\ref{rem:whenLAN}).

Our analysis delivers a simple organizing principle. Testing whether a treatment arm outperforms the control is a composite-null problem, and, unlike in a classical experiment, the sampling policy itself reacts to the (nuisance) parameters. The limit experiment shows that the null distribution of a test statistic only depends on these nuisance parameters through the limiting sampling policy. Two kinds of nuisance parameters matter. The first is the \emph{common reward level} of the arms of interest. A policy need not react to it, and \emph{translation invariance} -- the requirement that adding a constant to all rewards leaves the sampling probabilities unchanged -- removes it by design. Standard Thompson sampling is not translation invariant, because its prior shrinks the arms' estimates by arm-specific amounts towards a fixed point. The
second kind, which arises only with $K>2$ arms, consists of the
\emph{effects of the arms not involved in the comparison}. Any policy that aims to find the best arm must react to them, so they cannot be designed away and must be handled by the choice of test statistic. We show that, for example, the two-sample $t$-statistic is no longer pivotal, whereas a suitably weighted statistic is.

Our contributions are fourfold. First, we derive the limit experiment for $K$-armed contextual bandits with parametric reward models that are differentiable in quadratic mean -- covering, for instance, Gaussian, Bernoulli, and Poisson arms, but also heavy-tailed location models -- and a general class of policies, which
extends the sequentially randomized Markov experiments of
\citet{kuang2024weak} to contextual bandits and to policies and statistics based on general partial-sum processes. The log-likelihood ratios converge to those of a Locally Asymptotically
Quadratic (LAQ) experiment \citep{Jeganathan1995} that is neither LAN, LAMN, nor LABF. However, Girsanov's theorem yields a structural representation as a system of SDEs driven by the limiting sampling policy. We establish convergence jointly with that of a general class of statistics, including statistics weighted by powers of the inverse sampling probabilities, such as AW and inverse propensity weighted (IPW) statistics. Le Cam's third lemma then delivers local asymptotic power without triangular-array arguments. For the Gaussian non-contextual bandit, our SDEs for arm-pull frequencies and cumulative rewards coincide with those of \citet{kuang2024weak} and \citet{FanGlynn2021,fan2025diffusion}.
With non-Gaussian rewards a new feature emerges: the cumulative rewards used by standard policies are no longer sufficient in the limit experiment, whose sufficient statistics are score-based.

Second, we formalize translation invariant sampling for contextual bandits, show that it implies feasibility of the policy and that it carries over to the limiting policy (Corollary~\ref{cor_TI_I}). We construct regularized translation-invariant (RTI) versions of tempered-greedy, tempered-UCB/LinUCB, Thompson sampling, and the exploration sampling of \citet{kasy2021adaptive} for $K$ arms and linear contexts. \citet{kuang2024weak} introduced translation-invariant Thompson sampling for the two-armed non-contextual bandit and stressed its importance for comparing arms. We extend this to $K>2$ arms, allow for contexts, and other sampling policies. Once translation invariance is recognized as the relevant
requirement, some of these modifications are natural. The difficulty
lies in combining it with the regularity needed for a well-posed limit, as the arms' design matrices are (near-)singular early in the
experiment. We propose a determinant-based regularization that achieves both (see Theorem~\ref{thm:well-posedness}).\footnote{An alternative approach is to use a 'warm start'; see \citet{fan2025diffusion}.} Regularization and translation invariance come at little cost in regret (see Appendix~\ref{S-subsec:MAB_regret}).

Third, we use the limit experiment to analyze tests of the null
hypothesis that a treatment arm and the control have equal (conditional) mean rewards, treating the common level and the other arms' effects as nuisance parameters. Under RTI sampling, a family of two-sample statistics -- including the two-sample $t$-statistic, the AW statistic of \citet{hadad2021confidence}, and the IPW statistic -- is distribution-free with respect to the common level, both in finite samples and in the limit. For $K=2$ these tests are therefore valid over the entire composite null,
with critical values obtained by simulating the limit experiment. The $t$-test, whose limit is non-normal, is the most powerful among them for the parameter settings we investigated.
However, for $K>2$ these tests lose their validity: their null distributions depend on the other arms' rewards. For instance, the two-sample AW statistic is standard normal exactly if the limiting ratio of the arms' adaptive weights is deterministic -- the analogue of the variance-ratio condition of \citet[Theorem 4]{hadad2021confidence} -- which holds under fixed gaps but fails under the null and contiguous alternatives. Again, the structural representation of the limit experiment leads the way and suggests an adaptively weighted pivot, which applies the predictable variance stabilization of \citet{LuedtkeVanderlaan2016} and \citet{bibaut2021post} to the direct treatment-effect score of \citet[Eq.~(19)]{hadad2021confidence}.
This statistic weights each reward increment by its instantaneous inverse sampling probability and applies a variance-stabilizing weight. Its limit is exactly standard normal under the null for any policy with sampling probabilities bounded away from zero, so it is valid regardless of all nuisance parameters and even when the policy is not translation invariant. We obtain analogous results for the Wald test
and a contextual AW pivot in linear contextual bandits.

Fourth, Monte Carlo simulations corroborate the asymptotic theory. Null distributions at $T=200$ are close to their limits, also for
heavy-tailed rewards. The finite-sample power of the $t$-test, for $K=2$, is close to its limit already at $T=100$, whereas that of the AW-type statistics converges more slowly. Across policies, RTI Thompson sampling, which explores most, yields the highest power: confirming the basic intuition that exploration pays off at the
inference stage. Appendix~\ref{S-appsec:additionalsimulations} shows that using RTI sampling policies has no negative consequences for the regret of the sampling schemes for the considered parameters.

The paper is organized as follows. Section~\ref{sec:setup} sets up the stochastic contextual bandit problem. Section~\ref{sec:limitexperiment} derives the limit experiment, its structural SDE representation, and the joint convergence of a general class of statistics. Section~\ref{sec:sampling} introduces translation invariance and the RTI policies. Sections~\ref{sec:MAB} and~\ref{sec:CMAB} apply the theory to tests for non-contextual and contextual bandits, respectively, each with a Monte Carlo study. Section~\ref{sec:conclusion} concludes the paper. The appendices contain proofs and additional results.

\section{Setup} \label{sec:setup}
Consider the following multi-armed contextual bandit problem. At the beginning of each time step $\ii$, where $\ii\in\Sn\defeq\{1,\dots,\n\}$, an agent observes an exogenous variable $\vX_{\ii}$ collecting contextual information.\footnote{Non-contextual bandits can be embedded in this framework by setting $\vX_{\ii} = 1$ ($q=1$).} We assume the $\vX_{\ii}\in\SX\subset\rR^q$ to be independently and identically distributed and denote by $\measure_X$ their probability measure.
Subsequently, on basis of previously collected observations (in a way that will be made precise below) and the new context $\vX_{\ii}$, the agent chooses one of $K \geq 2$ possible arms, or treatments, $\A_{\ii}\in\SK$. Each arm is associated with an unknown probability distribution of outcomes. All potential outcomes are assumed to be mutually independent, both over arms and over time. Let $\Zkt$ denote the $\mathcal{Z}$-valued potential outcome of arm $k$ at time $\ii$ with $\mathcal{Z}\subset\rR$. The agent only observes $\Y_{\ii} = \ZAt$.
For $k\in\SK$, $\Zkt$ given $\vX_{\ii}$ has law $\mathcal{L}_{\pintb}(\Z_k\cond\vX)$, where $\pintb$ is a common parameter in an open parameter space $\Pintb\subset\rR^{p}$, possibly all of $\rR^{p}$. We assume that, with respect to some $\sigma$-finite dominating measure $\measure$, densities $z\mapsto \fzk(z \cond \vx,\pintb)$ of $\mathcal{L}_{\pintb}(\Z_k\cond\vX=\vx)$ exist, $\vx\in\SX$. We impose the following \emph{Differentiable in Quadratic Mean (DQM)} condition on these densities, together with moment conditions on the outcomes, contexts, and scores.

\begin{assumption} \label{assm:DQM}
For all $k\in\SK$, $\pintb\in\Pintb$, and $\measure_X$-almost every $\vx\in\SX$, the densities $\fzk(\cdot|x, \theta)$ are strictly positive on $\mathcal{Z}$ and differentiable in quadratic mean (DQM) at $\pintb$, i.e.,
\begin{align*}
\frac{\sqrt{\fzk(\Z\cond\vx,\pintb + \boldsymbol{\omega})}}{\sqrt{\fzk(\Z\cond\vx,\pintb)}} = 1 + \frac{1}{2}\left(\scorebk(\Z\cond\vx)^\trans\boldsymbol{\omega} + r_k(\Z\cond\vx,\boldsymbol{\omega})\right),
\end{align*}
for all $\boldsymbol{\omega}\in\rR^p$ in a neighborhood of $\mfzero$, and some $p$-vector $\scorebk(\cdot\cond\vx)$, called score functions. The remainder term $r_k$ satisfies
$\Exp\big[r_k^2(\Z_{k}\cond\vX,\boldsymbol{\omega})\big] = o(\|\boldsymbol{\omega}\|^2)$.
Moreover, we assume $\Exp[\vX\vX^\trans]>0$ and, for some $\delta>0$, $4+\delta$-th moments exist of $\Z_k$, $\vX$, and $\scorebk(\Z_k\cond\vX)$.
\end{assumption}
The DQM condition (Assumption~\ref{assm:DQM}) implies  $\Exp\big[\scorebk(\Z_{k}\cond\vX)\cond\vX\big] = \mfzero$ and existence of the $p\times p$ Fisher information matrix $\FJbpintk(\vX) \defeq \Exp\big[\scorebk(\Z_{k}\cond\vX)\scorebk(\Z_{k}\cond\vX)^\trans\cond\vX\big]$ (see \citet[Theorem 7.2]{van2000asymptotic}), almost surely. Moreover, the moment condition on the score function in Assumption~\ref{assm:DQM} ensures that $\Exp[\FJbpintk(\vX)]$ is finite. Assumption~\ref{assm:DQM} contains, for ease of presentation, conditions that are not always needed in all settings. For example, for the non-contextual case, the existence of $2+\delta$ moments is generally sufficient.

\begin{remark}
At this stage we do not assume that $\FJbpintk(\vX)$ is positive definite as some elements of $\pintb$ may be specific to a single arm only. This happens for instance in the simple $K$-arm non-contextual bandit problem where the mean rewards for the arms yield $K$ unrestricted parameters.
\end{remark}

\begin{remark}
The DQM condition is a standard regularity condition in the asymptotic analysis of statistical experiments. It is satisfied by wide classes of parametric models, also those including, e.g., binary or fat-tailed observations. Appendix~\ref{S-appsec:DQM_examples} verifies Assumption~\ref{assm:DQM} for Gaussian, Bernoulli, and Poisson rewards (with a $\mathrm{Uniform}(0,\theta)$ non-example). Using classical arguments it is probably possible to relax the assumption that $\fzk(z\cond\vX,\pintb) > 0$ for $z\in\mathcal{Z}$. And the
 existence of $4+\delta$ moments can likely be relaxed to Lindeberg type conditions.  But we maintain them for presentational convenience.
\end{remark}

The agent is allowed to use all information available at each moment $\ii$ in the sampling strategy. Thus, we define a filtration $(\filtr_{\ii})_{\ii \geq 1}$ through
$\filtr_{\ii} \defeq \sigma\left((\vX_{s},\A_{s},\Y_{s}): s = 1,\dots,{\ii}\right)$,
which collects the historical information of contexts, actions, and rewards up to and including time $\ii$. The agent chooses the $(\ii+1)$-th action $A_{\ii+1}$ via a draw from a multinomial distribution conditional on $\filtr_{\ii}$ and the newly observed context $\vX_{\ii+1}$. We call a strategy \emph{feasible} if it satisfies the following assumption.

\begin{assumption} \label{assm:policy_feasible}
For all $\ii = 0,\ldots,\n-1$, the conditional sampling probability
\begin{align*}
\pi_{\ii+1}(k \cond \vX_{\ii+1},\filtr_{\ii}) \defeq \prob(A_{\ii+1} = k \cond \vX_{\ii+1},\filtr_{\ii}),
\end{align*}
for $k\in\SK$, does not depend on $\pintb$.
\end{assumption}

Assumption~\ref{assm:policy_feasible} states that the agent has no knowledge of the true value of $\pintb$. It is worth noting that the unconditional distribution of $A_{\ii}$ typically \emph{does} depend on $\pintb$.

In this paper, we focus on the inferential problem of $\pintb$, particularly the testing problems in Sections~\ref{sec:MAB} and~\ref{sec:CMAB} for non-contextual and contextual bandits, respectively. In Section~\ref{sec:limitexperiment}, we first derive the limit experiment. This limiting framework allows us to readily analyze the asymptotic behavior of the original sequence of bandit experiments -- specifically, for our purpose, the asymptotic properties of statistics (see Section~\ref{subsec:limitexperiment_structural} for details). While our framework is primarily used for inference in this paper, it can also be seamlessly applied to other tasks, such as designing sampling schemes, which we leave for future work.

\section{Limit Experiment} \label{sec:limitexperiment}
We derive the limit experiment for the contextual bandit problem described in Section~\ref{sec:setup} at some fixed point $\pintb\in\Pintb$, where the arms have equal expected rewards -- we refer to this as \emph{equal-arms asymptotics}. It turns out that, given our parametric setup, we need to use $\sqrt{\n}$ as localizing rate.

\subsection{Quadratic expansion of likelihood ratios}
Localizing the parameter of interest at $\pintb$ as
\begin{equation} \label{eqn:localparameter_general}
\pintb_{\n} = \pintb + \frac{\lpintb}{\sqrt{\n}},
\end{equation}
we denote by $\lawn_{\pintb,\lpintb}$ the law of $(\vX_{1},\A_{1},\Y_{1},\dots,\vX_{\n},\A_{\n},\Y_{\n})$
generated by the aforementioned stochastic contextual $K$-armed bandit problem. Formally, we define the sequence of localized experiments as
\begin{align*}
\experimentn_{\pintb} \defeq \left(\Omega^{(\n)}, \mathcal{F}^{(\n)}, \left(\lawn_{\pintb,\lpintb} : \lpintb\in\rR^{p}\right)\right), ~~~ \n\in\rN,
\end{align*}
where $\Omega^{(\n)} = {\SX}^{\n}\otimes\SK^{\n}\otimes\rR^{\n}$ and $\filtrn = \mathcal{B}\left(\SX^{\n}\otimes\SK^{\n}\otimes\rR^{\n}\right)$, the Borel $\sigma$-field.

Using Assumption~\ref{assm:policy_feasible} and that the distribution of the contexts $\vX_{\ii}$ does not depend on $\pintb$, the log-likelihood ratio equals
\begin{equation*}
\begin{aligned}
    \log\frac{\dd\lawn_{\pintb,\lpintb}}{\dd\lawn_{\pintb,\mfzero}}
&= \log\frac{\textstyle\prod_{\ii=1}^{\n} \pi_{\ii}(\A_{\ii}\cond\vX_{\ii},\filtr_{\ii-1})\f_{A_\ii}(\Y_{\ii}\cond\vX_{\ii},\pintb_{\n})}{\textstyle\prod_{\ii=1}^{\n} \pi_{\ii}(\A_{\ii}\cond\vX_{\ii},\filtr_{\ii-1})\f_{A_{\ii}}(\Y_{\ii}\cond\vX_{\ii},\pintb)}  \\
&= \sum_{\ii=1}^{\n} \log\frac{\f_{A_{\ii}}(\Y_{\ii}\cond\vX_{\ii},\pintb_{\n})}{\f_{A_{\ii}}(\Y_{\ii}\cond\vX_{\ii},\pintb)}
= \sum_{k=1}^{K} \sum_{\ii=1}^{\n}\indicator_{\{A_{\ii} = k\}} \log\frac{\fzk(\Y_{\ii}\cond\vX_{\ii},\pintb_{\n})}{\fzk(\Y_{\ii}\cond\vX_{\ii},\pintb)}
\defeq \sum_{k=1}^{K}\LLRn_{\pintb,k}(\lpintb).
\end{aligned}
\end{equation*}
In the following proposition, we provide a quadratic expansion of this log-likelihood ratio which shows that the model exhibits the LAQ property (see \cite{Jeganathan1995}, \citet[Section~9.6]{van2000asymptotic}). We provide its proof in Appendix~\ref{S-appsec:proofs}.

\begin{proposition}\label{prop:LAQ}
Let Assumptions \ref{assm:DQM}--\ref{assm:policy_feasible} hold. Under $\lawn_{\pintb,\mfzero}$, we have, for $k\in\SK$ and $\lpintb\in\rR^{p}$, the decomposition
\begin{align} \label{eqn:loglikelihoodratio_sequence}
\LLRn_{\pintb,k}(\lpintb)
= \lpintb^\trans\CSb_{k,\n} - \frac{1}{2}\lpintb^\trans\QVb_{k,\n}\lpintb + \op(1),
\end{align}
where $\CSb_{k,\n}\in\rR^p$ and $\QVb_{k,\n}\in\rR^{p \times p}$ are defined by
\begin{align*}
\CSb_{k,\n} \defeq \frac{1}{\sqrt{\n}}\sum_{\ii=1}^{\n}\indicator_{\{A_{\ii} = k\}}\scorebk(\Y_{\ii}\cond\vX_{\ii}) \text{~~and~~}
\QVb_{k,\n} \defeq \frac{1}{\n}\sum_{\ii=1}^{\n}\indicator_{\{A_{\ii} = k\}}\FJbpintk(\vX_{\ii}).
\end{align*}
\end{proposition}

\subsection{Weak convergence of relevant processes}
To derive the limit experiment for the contextual bandit problem, we next establish weak convergence of the central sequence and the Fisher information in Proposition~\ref{prop:LAQ}, jointly with other statistics of interest. This will allow us to both describe the limit experiment in Section~\ref{subsec:limitexperiment_likelihood} and study commonly used inference methods in subsequent sections.

To this end, we need the joint limiting behavior of the following partial-sum processes
\begin{equation*}
\begin{aligned}
\parsumfb_{k,\ii} \defeq &~ \frac{1}{\sqrt{\n}}\sum_{s=1}^{\ii}\indicator_{\{A_{s} = k\}}\funfb_{k,\n}(\Y_{s}, \vX_{s} \cond \Ubparsum_{s-1}),\quad  
\parsumqb_{k,\ii} \defeq  \frac{1}{\n}\sum_{s=1}^{\ii}\indicator_{\{A_{s} = k\}}\funqb_{k,\n}(\Y_{s}, \vX_{s}, \Ubparsum_{s-1}),
\end{aligned}
\end{equation*}
for $k\in\SK$ and vector-valued functions $\funfb_{k,\n}$ and $\funqb_{k,\n}$, where $\Ubparsum_{\ii} \defeq ({\parsumfb_{\ii}}^\trans,{\parsumqb_{\ii}}^\trans)^\trans$ with $\parsumfb_{\ii} \defeq \big({\parsumfb_{1,\ii}}^\trans,\dots,{\parsumfb_{K,\ii}}^\trans\big)^\trans$ and $\parsumqb_{\ii} \defeq \big({\parsumqb_{1,\ii}}^\trans,\dots,{\parsumqb_{K,\ii}}^\trans\big)^\trans$. We denote by $l_1$ and $l_2$ the dimensions of $\funfb_{k,\n}$ and $\funqb_{k,\n}$, respectively.\footnote{The effective dimensions can be $k$-specific, by stacking vectors with zeros.} $\Ubparsum_{\ii}$ is thus of dimension $lK$ with $l=l_1+l_2$. The functions $\funfb_{k,\n}$ and $\funqb_{k,\n}$ generally depend on the horizon $\n$; the asymptotic analysis below, based on triangular arrays, accommodates this dependence. We retain their $\n$-index but, for readability, suppress the induced $\n$-dependence of $\parsumfb_{k,\ii}$, $\parsumqb_{k,\ii}$, and $\Ubparsum_{\ii}$.

We assume that $\funfb_{k,\n}$ and $\funqb_{k,\n}$ satisfy, for all $k\in\SK$ and $\ii\in\Sn$, and under $\lawn_{\pintb,\mfzero}$,
\begin{equation} \label{eqn:functioncondition_moment_I}
\Exp_{\pintb}\big[\funfb_{k,\n}(\Z_k, \vX_{\ii} \cond \Ubparsum_{\ii-1})\cond\vX_{\ii}, \Ubparsum_{\ii-1}\big] = 0
\end{equation}
and there exists $\delta > 0$ such that, for all $c>0$,
\begin{equation} \label{eqn:functioncondition_moment_II}
\begin{aligned}
&~ \sup_{\n}\sup_{\{\vu : \|\vu\| \leq c\}}\Exp_{\pintb}\big[\|\funfb_{k,\n}(\Z_{k}, \vX \cond \vu) \|^{2+\delta} + \|\funqb_{k,\n}(\Z_{k}, \vX, \vu)\|^{2+\delta}\big] < \infty,
\end{aligned}
\end{equation}
The supremum over $\n$ is often easily verified as it involves vanishing tuning parameters.

Moreover, we impose the following assumption on the sampling policy, which states that $\Ubparsum_{\ii-1}$ forms a sufficient statistic for the selection strategy at round $\ii$.

\begin{assumption} \label{assm:policy_limit}
\begin{itemize}
\item[(a)] We assume that the sampling probabilities are described via strictly positive deterministic functions $\policy_k^{(\n)}$ such that, for $k \in \SK$ and $\ii\in\Sn$,
\begin{equation} \label{eqn:policy_k}
\pi_{\ii}(k\cond\vX_{\ii},\filtr_{\ii-1}) = \policy_k^{(\n)}\left(\Ubparsum_{\ii-1},\vX_{\ii}\right).
\end{equation}
Moreover, for all $k\in\SK$, $\policy_k^{(\n)}$ converges (pointwise) to a continuous function $\policy_k(\vu,\vx)$.
Finally, uniformly over bounded sets of $\vu$, we have, as $T\to\infty$\footnote{The expectation $\Exp$ refers to integrating with respect to the distribution of $\vX_\ii$ and $\Z_{k}$.},
\begin{align*}
\Exp\left[\policy_k^{(\n)}(\vu,\vX)\funfb_{k,\n}(\Z_{k}, \vX \cond \vu)\funfb_{k,\n}(\Z_{k},\vX\cond\vu)^\trans\right] &\to \Epolicyb_k^{\circ}(\vu), \\
\Exp\left[\policy_k^{(\n)}(\vu,\vX)\funqb_{k,\n}(\Z_{k},\vX,\vu)\right] &\to \Epolicyb_k^{\star}(\vu),
\end{align*}
for some continuous functions $\Epolicyb_k^{\circ}$ and $\Epolicyb_k^{\star}$.

\item[(b)] There exists $\policy_{\min}>0$ such that
$\policy_k^{(\n)}(\vu,\vx)\geq\policy_{\min}$ for all $\n$, $k\in\SK$, $\vu$,
and $\measure_X$-a.e.\ $\vx\in\SX$. By pointwise convergence, $\policy_k\geq\policy_{\min}$ as well.
\end{itemize}
\end{assumption}

In applications one usually has
\begin{eqnarray*}
\Epolicyb_k^{\circ}(\vu)
    &\defeq& \lim_{\n\to\infty}\Exp\left[\policy_k(\vu,\vX)\funfb_{k,\n}(\Z_{k},\vX\cond\vu)\funfb_{k,\n}(\Z_{k},\vX\cond\vu)^\trans\right],\\
\Epolicyb_k^{\star}(\vu)
    &\defeq& \lim_{\n\to\infty}\Exp\big[\policy_k(\vu,\vX)\funqb_{k,\n}(\Z_{k},\vX,\vu)\big].
\end{eqnarray*}

\begin{remark} \label{remark:clipping}
Part~(b) is needed only where $\policy_k$ enters a denominator, namely for
$\vC_{r_1;k}$ and $\vS_{r_2;k}$ in \eqref{eqn:statistic_general_CS} with
$r_1>0$ or $r_2>0$, that is, for the AW ($r=\tfrac12$) and IPW ($r=1$)
statistics below. The results on the limit experiment and all results at $r_1,r_2\leq0$,
(including the $t$~statistic), use Part~(a) alone. Part~(b) in particular holds under \emph{clipping} of a scheme $\tilde\policy_k^{(\n)}$
at level $\policy_{\min}$, defined as
\begin{equation} \label{eqn:uniform_mixing}
\policy_k^{(\n)} \defeq (1-K\policy_{\min})\,\tilde\policy_k^{(\n)} + \policy_{\min},
\qquad \policy_{\min}\in(0,1/K],
\end{equation}
which preserves the regularity in Part~(a), and is the form used in our Monte Carlo studies.
\end{remark}

\noindent
\begin{remark}
Note that
$\pi_{\ii}(k\cond\vX_{\ii},\filtr_{\ii-1}) = \policy_k^{(\n)}\left(\Ubparsum_{\ii-1},\vX_{\ii}\right)
$
is allowed to depend deterministically on time $\ii$ as we may include in $\parsumqb_{k,\ii}$ a component $\funqb_{k,\n}=1$ for all $k$ and sum these components over $k\in\SK$ to get $\ii/T$.
\end{remark}

It is easily seen that, for each $\n$, the process $\Ubparsum_{\ii} = \Ubparsum^{(\n)}_{\ii}$, $\ii\in\Sn$, is a time-homogeneous Markov chain, as $\vX_{\ii}$ is independent of $\filtr_{\ii-1}$ and $\Y_{\ii}$ is independent of $\filtr_{\ii-1}$ conditionally on $\vX_{\ii}$ and $\A_{\ii}$. In Proposition~\ref{prop:weakconvergence_general} we demonstrate that an appropriately time-changed version of $\Ubparsum_{\ii}$ converges weakly to the solution of a system of stochastic differential equations. This system is required to admit a unique solution.

\begin{assumption} \label{assm:sde_solution}
The system of stochastic differential equations
\begin{equation} \label{eqn:limit_QVVV}
\begin{aligned}
\dd\parsumfb_{k}(\uu) \defeq&~ \sqrt{\Epolicyb_k^{\circ}(\Ublim(\uu))} \dd\vW_{k}(\uu), \quad
\dd\parsumqb_{k}(\uu) \defeq \Epolicyb_k^{\star}(\Ublim(\uu)) \dd\uu,
\quad \Ublim(0)=0,
\end{aligned}
\end{equation}
with $\uu\in[0,1]$, $k\in\SK$, and $\Ublim \defeq \left( ( \parsumfb_{k} )_{k=1}^{K}, ( \parsumqb_{k} )_{k=1}^{K}\right)^\trans$,
admits a unique (weak) solution for any multivariate Brownian motion $(\vW_{k})_{k=1}^{K}$ with covariance matrix per-unit-of-time given by $\cov\left[\vW_{k}(1),\vW_{k'}(1)\right] = \indicator_{\{k = k'\}}\idm_{l_1}$, $k,k'\in\SK$, where $\idm_{l_1}$ denotes the $l_1$-dimensional identity matrix and $\sqrt{\vA}$ the symmetric positive semi-definite root of $\vA$.
\end{assumption}

A sufficient condition for the existence of a unique strong solution to~(\ref{eqn:limit_QVVV}) is that both $\vu \mapsto \Epolicyb_k^{\star}(\vu)$ and $\vu \mapsto \sqrt{\Epolicyb_k^{\circ}(\vu)}$ are locally Lipschitz continuous on bounded sets and of linear growth, see \citet[Theorems~IV.2.4 and~IV.3.1]{ikeda2014stochastic}. In Section~\ref{sec:sampling} we propose a regularization of sampling schemes guaranteeing this.

We now state a convergence result for the discrete time process $\Ubparsum_{\ii}$. This result builds on existing convergence results. We use it later to prove the convergence of likelihood ratios and, thereby, of the statistical experiments of interest. Its proof is provided in Appendix~\ref{S-appsec:proofs}.

\begin{proposition} \label{prop:weakconvergence_general}
Assume that Assumptions~\ref{assm:DQM}--\ref{assm:sde_solution} hold. Embed $\Ubparsum_{\ii}$ in $D^{mK}[0,1]$ by defining $\bar\Ubparsum(\uu) \defeq \Ubparsum_{\lfloor\uu\n\rfloor}$, for $\uu\in[0,1]$, where $\lfloor\cdot\rfloor$ takes the integer part of the argument.
Then, under $\lawn_{\pintb,\mfzero}$ and as $\n\to\infty$,
\begin{equation}
\bar\Ubparsum \wto  \Ublim,
\end{equation}
where $\Ublim$ is defined in Assumption~\ref{assm:sde_solution}.
\end{proposition}

\subsection{Likelihood representation of limit experiment} \label{subsec:limitexperiment_likelihood}
To fully leverage the limit experiment framework (see the detailed discussion at the end of this section), we establish joint convergence of the CMAB log-likelihood ratio, $\LLRn_{\pintb,k}(\lpintb)$, along with the following two statistics
\begin{equation} \label{eqn:statistic_general_CS}
\begin{aligned}
\vC_{r_1;k,\ii} \defeq&~ \frac{1}{\sqrt{\n}}\sum_{s=1}^{\ii}\frac{\indicator_{\{A_{s} = k\}}}{\big(\policy_k^{(\n)}(\UbparsumE_{s-1},\vX_s)\big)^{r_1}}\vX_s(\Y_{s} - \Exp_{\pintb}[\Z_{k}\cond \vX_s]), \\
\vS_{r_2;k,\ii} \defeq&~ \frac{1}{\n}\sum_{s=1}^{\ii}\frac{\indicator_{\{A_{s} = k\}}}{\big(\policy_k^{(\n)}(\UbparsumE_{s-1},\vX_s)\big)^{r_2}}\vX_s\vX_s^\trans,
\end{aligned}
\end{equation}
for $r_1, r_2\in\rR$.
Here, $\UbparsumE_{\ii} \defeq \left(({\CSb_{k,\ii}}^\trans)_{k=1}^{K}, ({\QVb_{k,\ii}}^\trans)_{k=1}^{K}, ({\vC_{r_1;k,\ii}}^\trans)_{k=1}^{K}, ({\vS_{r_2;k,\ii}}^\trans)_{k=1}^{K}\right)^\trans$ collects all relevant statistics, i.e. both the quantities used by the sampling policy and those employed as test statistics in subsequent sections. For instance, in the non-contextual bandits setting ($\vX_{\ii} = 1$, for all $\ii$), $\vC_{r_1=0;k,\ii}$ and $\vS_{r_2=0;k,\ii}$ reduce to the centered and re-scaled accumulated rewards and sampling frequencies $\reward_{k, \ii}$ and $\freq_{k,\ii}$, introduced in~(\ref{eqn:reward_frequency}) below.

We define the sampling policy $\policy^{(\n)}_k$ as a function of $\left({\vC_{r_1;k,\ii}}^\trans)_{k=1}^{K}, ({\vS_{r_2;k,\ii}}^\trans)_{k=1}^{K}\right)$ and the new incoming contextual observation $\vX_{\ii+1}$. Applying Proposition~\ref{prop:weakconvergence_general} then yields the following joint weak convergence result.

\begin{proposition} \label{prop:likelihoodweakconvergence}
Let $r_1,r_2\in\rR$ and Assumption~\ref{assm:DQM}--\ref{assm:policy_limit}(a) and \ref{assm:sde_solution} hold. If $r_1>0$ or $r_2>0$, assume also Assumption~\ref{assm:policy_limit}(b).
\begin{itemize}
\item[(a)] Under $\lawn_{\pintb,\mfzero}$, we have the joint convergence, for $k\in\SK$ and $r_1,r_2\in\rR$,
\begin{equation}
\begin{aligned}
\LLRn_{\pintb,k}(\lpintb) &\wto \LLRlim_{\pintb,k}(\lpintb) = \lpintb^\trans\CSb_{k}(1) - \frac{1}{2}\lpintb^\trans\QVb_{k}(1)\lpintb, \\
\vC_{r_1;k,\n} &\wto \vC_{r_1;k}(1)  \text{~~and~~}
\vS_{r_2;k,\n} \wto \vS_{r_2;k}(1),
\end{aligned}
\end{equation}
with, for $\uu\in[0,1]$,
\begin{equation*}
\begin{aligned}
\dd\left[\begin{array}{c}
    \CSb_{k}(\uu)\\
    \vC_{r_1;k}(\uu)
\end{array}\right]
    \defeq&~
\left[\begin{array}{cc}
\Epolicyb^{\vartriangle}_{k}(\UblimE(\s)) &
    \Epolicyb^{\vx}_{r_1;k}(\UblimE(\s))^\trans\\
\Epolicyb^{\vx}_{r_1;k}(\UblimE(\s)) &
    \Epolicyb^{\vc}_{r_1;k}(\UblimE(\s))
\end{array}\right]^{\frac{1}{2}}\dd\vW_{k}(\s),\\
\dd\QVb_{k}(\s) \defeq&~ \Epolicyb^{\vartriangle}_{k}(\UblimE(\s))\dd s,\\
\dd\vS_{r_2;k}(\s) \defeq&~ \Epolicyb^{\vs}_{r_2;k}(\UblimE(\s))\dd\s,
\end{aligned}
\end{equation*}
where $\UblimE(\uu) \defeq \left(({\CSb_{k}(\uu)}^\trans)_{k=1}^{K}, ({\QVb_{k}(\uu)}^\trans)_{k=1}^{K}, (\vC_{r_1;k}(\uu)^\trans)_{k=1}^{K}, (\vS_{r_2;k}(\uu)^\trans)_{k=1}^{K}\right)^\trans$ collects all processes, with
\begin{eqnarray}
\Epolicyb^{\vartriangle}_{k}(\vu)
 &\defeq& \Exp\big[\policy_k(\vu,\vX)\FJbpintk(\vX)\big],\\
\Epolicyb^{\vc}_{r_1;k}(\vu)
 &\defeq& \Exp\big[\policy_k^{1-2r_1}(\vu,\vX)\var(\Z_k\cond\vX)\vX\vX^\trans\big],\\
\Epolicyb^{\vx}_{r_1;k}(\vu)
 &\defeq& \Exp\big[\policy_k^{1-r_1}(\vu,\vX)\vX
  \left(\Z_k-\Exp\left[\Z_k\cond\vX\right]\right)
  \scorebk(\Z_k\cond\vX)^\trans
  \big],\mbox{ and }\\
\Epolicyb^{\vs}_{r_2;k}(\vu)
 &\defeq& \Exp\big[\policy_k^{1-r_2}(\vu,\vX)\vX\vX^\trans\big].
\end{eqnarray}
Here, $\vW_k$ is a $p+q$-dimensional standard Brownian motion.
Brownian motions associated with different arms are mutually independent.
\item[(b)] We have, still under $\lawn_{\pintb,\mfzero}$,
$
\log\frac{\dd\lawn_{\pintb,\lpintb}}{\dd\lawn_{\pintb,\mfzero}} \wto \LLRlim_{\pintb}(\lpintb) \defeq \sum_{k=1}^{K}\LLRlim_{\pintb,k}(\lpintb).
$
\item[(c)] With $\CSb_{k}$ and $\QVb_{k}$ defined in~(a), $\exp\left(\lpintb^\trans\CSb_{k}(\uu)-\frac12\lpintb^\trans
\QVb_{k}(\uu)\lpintb\right)$ is a martingale for all $k\in\SK$ and $\lpintb\in\rR^{p}$ and so is $\exp\left(\lpintb^\trans\sum_k\CSb_{k}(\uu)-\frac12\lpintb^\trans\sum_k
\QVb_{k}(\uu)\lpintb\right)$. In particular, $\Exp_{\pintb,\mfzero}\left[\exp\LLRlim_{\pintb}(\lpintb)\right] = 1$.
\end{itemize}
\end{proposition}
Part~(a) demonstrates that the limiting log-likelihood ratio is quadratic in $\lpintb$, thus placing the limit experiment in the \emph{Locally Asymptotically Quadratic (LAQ)} class of \citet{Jeganathan1995}. However, due to the possible dependence between the integrands containing $\Epolicyb^{\vartriangle}_{k}(\UblimE)$ and the integrator $\vW_{k}$, $\CSb_{k}$ is generally not normally or mixed-normally distributed. Therefore, the limit experiment does not adhere to the traditional \emph{Locally Asymptotically Normal (LAN)} or \emph{Locally Asymptotically Mixed Normal (LAMN)} forms. The limiting likelihood ratio generally does not satisfy the conditions for being \emph{Locally Asymptotically Brownian Functional (LABF)} either, due to the presence of the $\QVb_k$-processes. We remark that Section~\ref{sec:MAB} and~\ref{sec:CMAB} below use a slightly generalized version of Proposition~\ref{prop:likelihoodweakconvergence} where the policy may be based on different values for $r_1$ and $r_2$ than those in the test statistics.

\begin{remark}\label{rem:whenLAN}
The equal-arms localization (arm-mean gaps of order $\n^{-1/2}$) is the informative regime in which signal and sampling noise accumulate at the same rate and local power is non-trivial. Two boundary cases yield Gaussian-shift (LAN) limits. (i)~\emph{Strong signal}: with fixed gaps the policy concentrates on the best arm, whose mean is estimated at the regular $\sqrt{\n}$ rate. The experiment is LAN, and the nonstandard phenomena studied here do not arise; though the rates on non-optimal arms are nonstandard (see \cite{akker2025local}) (ii)~\emph{Deterministic information}: if the policy is non-adaptive (constant $\policy_k$) or follows a fixed exploration schedule, the integrand $\sqrt{\Epolicyb^{\circ}_k}$ is deterministic, $\CSb_k$ is Gaussian, and the limit experiment is LAN.
This also holds true when the terminal information $Q_k(1)$ is deterministic.
Non-LAN behavior arises when continuously-updated, path-dependent sampling correlates $\sqrt{\Epolicyb^{\circ}_k}$ with its driving Brownian motion. In particular, oversampling inferior arms does not, in general, restore validity of classical statistical inference.
\end{remark}

\subsection{Structural representation of limit experiment} \label{subsec:limitexperiment_structural}
Part (c) of Proposition~\ref{prop:likelihoodweakconvergence} allows us to introduce a new collection of probability measures, denoted by $\law_{\pintb,\lpintb}$, with $\lpintb\in\rR^{p}$, via
\begin{equation}
\frac{\dd\law_{\pintb,\lpintb}}{\dd\law_{\pintb,\mfzero}} = \exp(\LLRlim_{\pintb}(\lpintb)),
\end{equation}
i.e., the Radon-Nikodym derivative with respect to $\law_{\pintb,\mfzero}$. Now we formally define the limit experiment as
$
\experiment_{\pintb} \defeq \left(\Omega, \mathcal{F}, \left(\law_{\pintb,\lpintb} : \lpintb\in\rR^{p}\right)\right),
$
where $(\Omega,\mathcal F)$ is the canonical measurable space of
$\left(\vW_k\right)_{k=1}^K$ in Proposition~\ref{prop:likelihoodweakconvergence}(a).
Applying Girsanov's theorem, we have the following theorem. The  proof is organized in Appendix~\ref{S-appsec:proofs}.

\begin{theorem} \label{thm:structurallimitexperiment_f}
Let $\lpintb\in\rR^{p}$ and fix the sampling scheme $\policy$. The limit experiment $\experiment_{\pintb}$ associated with the log-likelihood ratio $\LLRlim_{\pintb}(\lpintb)$ can be described as follows.
We observe $\CSb_{k}$, $\QVb_{k}$, $\vC_{r_1;k}$, and $\vS_{r_2;k}$ generated by the stochastic differential equations, $k\in\SK$ and $\s\in[0,1]$,
\begin{eqnarray}
\nonumber
\dd\left[\begin{array}{c}
    \CSb_{k}(\s)\\
    \vC_{r_1;k}(\s)\end{array}\right]
    &=&
\left[\begin{array}{c}
    \Epolicyb^{\vartriangle}_{k}(\UblimE(\s))\\
    \Epolicyb^{\vx}_{r_1;k}(\UblimE(\s))\end{array}\right]
    \lpintb\dd\s+
\left[\begin{array}{cc}
    \Epolicyb^{\vartriangle}_{k}(\UblimE(\s)) &
        \Epolicyb^{\vx}_{r_1;k}(\UblimE(\s))^\trans\\
    \Epolicyb^{\vx}_{r_1;k}(\UblimE(\s)) &
        \Epolicyb^{\vc}_{r_1;k}(\UblimE(\s))
    \end{array}\right]^{\frac12}
    \dd\vB_{k}(\s), \\
\label{eqn:structurallimitexperiment}
\dd\QVb_{k}(\s) &=& \Epolicyb^{\vartriangle}_{k}(\UblimE(\s))\dd\s, \\
\nonumber
\dd\vS_{r_2;k}(\s) &=& \Epolicyb^{\vs}_{r_2;k}(\UblimE(\s))\dd\s,
\end{eqnarray}
where $\vB_{k}$ are $p+q$-dimensional standard Brownian motions, independent across arms.
\end{theorem}

Up to this point, we have developed the limit experiment for the bandit problem and its structural version formulated in SDEs. These results demonstrate that the bandit limit experiment $\experiment_{\pintb}$ corresponds to observing certain continuous-time processes $\CSb_{k}$, $\QVb_{k}$, $\vC_{r_1;k}$, and $\vS_{r_2;k}$, $k\in\SK$, from a model $(\law_{\pintb,\lpintb} \cond \lpintb\in\rR^{p})$. We now illustrate how this limiting framework can be used to analyze the local asymptotic behavior of statistics and, consequently, the local asymptotic performance of inferential procedures. Such analyses typically require cumbersome `triangular array' limit theory. Fortunately, we can exploit the limit experiment to avoid this by invoking Le Cam's third lemma. This works as follows.

Suppose we are interested in studying the bandit problem around a specific value, $\pintb_0$, typically the value under the null hypothesis, though in some cases this choice may be inconsequential (see, e.g., the case of comparing arms in Sections~\ref{subsec:MAB_comparetwoarms} and~\ref{subsec:CMAB_comparetwoarms}). To this end, we adopt the local reparameterization $\pintb_{\n} = \pintb_0 + \lpintb/\sqrt{\n}$, which translates the hypotheses of the original bandit problem into those in the limit experiment derived above, $(\law_{\pintb_0,\lpintb} \cond \pintb_0\in\Pintb, \lpintb\in\rR^{p})$, now characterized by $\lpintb$. Consider a sequence of test statistics $\stat_\n$ that satisfies, under $\lawn_{\pintb_0,\mfzero}$ and for all $\lpintb\in\rR^p$,
\begin{equation}\label{eqn:statistic_conv_joint_lr}
\left(\stat_\n, \dd\lawn_{\pintb_0,\lpintb}\big/\dd\lawn_{\pintb_0,\mfzero} \right)
\Rightarrow  \mathcal{L}\left( g(\UblimE), \dd\law_{\pintb_0,\lpintb}/\dd\law_{\pintb_0,\mfzero} \cond \law_{\pintb_0,\mfzero}\right).
\end{equation}
for some functional $g$. Le Cam's third lemma (see, for example, \citet[Theorem~6.6]{van2000asymptotic}) implies that, under $\lawn_{\pintb_0,\lpintb}$,
\begin{equation}\label{eqn:statistic_conv_alt}
\stat_\n \Rightarrow   \mathcal{L}( g(\UblimE) \cond \law_{\pintb_0,\lpintb}).
\end{equation}

In short, establishing (\ref{eqn:statistic_conv_joint_lr}) for a given sequence of test statistics provides the full distributional behavior under local alternatives $\lpintb$ via (\ref{eqn:statistic_conv_alt}). Notably, their asymptotic behaviors are explicitly described by the structural representation in (\ref{eqn:structurallimitexperiment}), expressed in SDEs, which can be readily simulated using an Euler scheme. In Sections~\ref{sec:MAB} and \ref{sec:CMAB} we will repeatedly leverage this property to study asymptotic validity and (local) asymptotic powers of commonly used and newly proposed test statistics.

Beyond facilitating the derivation of the (local) asymptotic behavior of test statistics, the limit experiment can also be used to obtain upper bounds on the (local) asymptotic power of asymptotically valid tests. The Asymptotic Representation Theorem (see, for example, \citet[Chapter~9]{van2000asymptotic}) states that for any sequence of statistics that converges in distribution to a distribution $\mathcal{L}_{\lpintb}$ under $\lawn_{\lpintb}$ for all $\lpintb$, there exists a (randomized) statistic $\stat$ in the limit experiment such that the distribution of $\stat$ under $\law_{\lpintb}$ is given by $\mathcal{L}_{\lpintb}$.
Consequently, the asymptotically optimal procedures in $\experimentn_{\pintb}$ are determined by the optimal procedures in $\experiment_{\pintb}$ for the same purpose. Building on this insight, we derive upper power bounds in the subsequent sections.

\section{Sampling schemes}\label{sec:sampling}

In this section we introduce \emph{regularized} and \emph{translation invariant} (RTI) versions of the following, often used, sampling schemes: tempered greedy, (Lin)UCB, and Thompson sampling. The regularization makes the sampling schemes satisfy Assumption~\ref{assm:policy_limit}(a) and Assumption~\ref{assm:sde_solution}; see Theorem~\ref{thm:well-posedness} below. The translation-invariance
implies Assumption~\ref{assm:policy_feasible} and allows for the construction of tests for the equality of the arms' reward distributions whose size can be controlled over the entire (composite) null hypothesis.


We present the sampling schemes for the linear contextual $K$-armed bandit. The potential reward of Arm-$k$ at round $\ii$ follows the linear model
\begin{equation}
\Zkt = \vX_{\ii}^\trans\vbeta_{k} + \ekt, ~~~ k\in\SK,
\end{equation}
where $\vX_{\ii} = (\X_{1,\ii},\dots,\X_{q,\ii})^\trans$ represents i.i.d.\ contextual random variables, with finite and non-singular $\Exp\left[\vX_\ii\vX_\ii^\trans\right]$. The dimension of the parameter space is $p=Kq$. The errors $\ekt$ are independent across arms and rounds, with mean zero, unit variance, and density $\f_k$. We impose Assumption~\ref{assm:DQM} with $\pintb = (\vbeta_1^\trans,\dots,\vbeta_K^\trans)^\trans$; a sufficient condition is that $\f_k$ is absolutely continuous with score $\scorek \defeq -\dot{\f}_k/\f_k$ and finite Fisher information $\FJ_k \defeq \int\scorek(\e)^2\f_k(\e)\dd\e$. The contexts and the errors are mutually independent. The non-contextual $K$-armed bandit (MAB) is the special case $q=1$ and $\vX_\ii\equiv1$ (Section~\ref{subsec:MABLimit}).

As we are primarily interested in testing equality of arm's reward distributions, we consider \emph{equal-arms asymptotics} where the parameters $\vbeta_k$ are localized around a common parameter vector $\vbeta\in\rR^q$ as, with $\vb_k\in\rR^q$,
\begin{equation} \label{eqn:equalarms_asymptotics_CMAB}
\vbeta_{k,\n} = \vbeta + \frac{\vb_k}{\sqrt{\n}}.
\end{equation}
At the observed context $\vx$, we write $\vm(\vx)=(m_1(\vx),\dots,m_K(\vx))^\trans$, with $m_k(\vx)=\vx^\trans\vb_k$, for the $\sqrt{\n}$-localized reward vector.

We will consider sampling policies for linear CMAB that are functions of
\begin{align*}
\vC_{k,\ii} =&~ \vC_{r_1=0;k,\ii} = \frac{1}{\sqrt{\n}}\sum_{s=1}^{\ii}\indicator_{\{A_{s} = k\}}\vX_s(\Y_{s} - \vX_s^\trans\vbeta), \\
\vS_{k,\ii} =&~ \vS_{r_2=0;k,\ii} = \frac{1}{\n}\sum_{s=1}^{\ii}\indicator_{\{A_{s} = k\}}\vX_s\vX_s^\trans,
\end{align*}
where $\vC_{r_1;k,\ii}$ and $\vS_{r_2;k,\ii}$ are defined in (\ref{eqn:statistic_general_CS}). In particular, with the incoming context $\vx\in\rR^q$, applications often use the OLS predictor for arm $k$. Its localized version is $\hat\m_k(\vx) = \vx^\trans\hat\vb_k \defeq \vx^\trans\vS_k^{-1}\vC_k$. We let the probability of selecting Arm-$k$ at time $\ii+1$ be
\begin{align*}
\pi_{\ii+1}(k\cond\vX_{\ii+1},\filtr_{\ii}) = \policy_k^{(\n)}\left(\vS_{\ii},\vC_{\ii},\vX_{\ii+1}\right),
\end{align*}
where $\vC_{\ii} = \left(\vC_{1,\ii}^\trans,\dots,\vC_{K,\ii}^\trans\right)^\trans\in\rR^{Kq}$ and $\vS_{\ii} = \left(\vS_{1,\ii}^\trans,\dots,\vS_{K,\ii}^\trans\right)^\trans\in\rR^{Kq\times q}$, with each $\vS_{k,\ii}$ positive semi-definite. In the MAB case, $\vS_k=\freq_k$ (arm-pull frequency), $\vC_k=\reward_k$ (accumulated reward); see~(\ref{eqn:reward_frequency}) below.

\subsection{Translation Invariance}\label{subsec:TI-schemes}
Sampling based on the predictors $\hat\m_k(\vx)$ is infeasible,
i.e. violating Assumption~\ref{assm:policy_feasible}, as it depends on the unknown common parameter $\vbeta$ (through $\vC_k$). We therefore restrict attention to \emph{translation-invariant} schemes.

\begin{definition} \label{def:translationinvariant}
A sequence of sampling schemes $\policy_k^{(\n)}$ is \emph{translation invariant} if
\begin{equation}
\policy_k^{(\n)}\big(\vs, \vc + \vs\boldsymbol{e}, \vx\big) = \policy_k^{(\n)}\big(\vs, \vc, \vx\big)
\end{equation}
for all $\n\in\rN$, $\boldsymbol{e}\in\rR^q$, $\vs\in\rR^{Kq\times q}$, $\vc\in\rR^{Kq}$, and $\vx\in\rR^q$.
\end{definition}
\noindent In the MAB case, this is $\policy_k^{(\n)}(\vd, \vr + c\vd) = \policy_k^{(\n)}(\vd, \vr)$ for all $c\in\rR$, $\vr\in\rR^K$, $\vd\in[0,1]^K$.

\begin{corollary}\label{cor_TI_I}
If $\policy_k^{(\n)}$ is translation invariant, then (a) it satisfies Assumption~\ref{assm:policy_feasible} and (b) the associated limiting scheme $\policy_k$ of Assumption~\ref{assm:policy_limit}(a) satisfies $$\policy_k(\vs,\vc+\vs\boldsymbol{e},\vx)=\policy_k(\vs,\vc,\vx)$$ for all $\ve\in\rR^q$, $\vs\in\rR^{Kq\times q}$, $\vc\in\rR^{Kq}$, and $\vx\in\rR^q$.
\end{corollary}

Part~(a) makes the finite-sample sampling path independent of the global common parameter $\vbeta$, while part~(b) makes the limiting sampling path independent of a common shift of $\vb_1,\dots,\vb_K$. Translation invariance thus removes the common level while preserving the contrasts $\vb_k-\vb_1$. Accordingly, a sufficient condition for translation invariance is that $\policy_k^{(\n)}$ depends on the data only through $\vx$, $\vS_\ii$, and the predicted-reward \emph{contrasts} $\vx^\trans(\vS_{k,\ii}^{-1}\vC_{k,\ii}-\vS_{j,\ii}^{-1}\vC_{j,\ii})$, the latter being estimates of $\vx^\trans(\vb_k - \vb_j)$ based on the information available up to round $\ii$. Indeed, for each $a\in\SK$ and any $\ve\in\rR^q$,
\[
\vS_{a,\ii}^{-1}\big(\vC_{a,\ii}+\vS_{a,\ii}\ve\big)
=
\vS_{a,\ii}^{-1}\vC_{a,\ii}+\ve ,
\]
so the shift $\vc\mapsto\vc+\vs\ve$ of Definition~\ref{def:translationinvariant} moves every predictor $\hat m_a(\vx)$ by the same amount $\vx^\trans\ve$ and leaves the contrasts unchanged. In terms of the model, this shift corresponds to replacing the common parameter $\vbeta$ by $\vbeta-\ve/\sqrt{\n}$. In particular, this implies that $(\vS_{a,\ii},\vS_{k,\ii}^{-1}\vC_{k,\ii}-\vS_{j,\ii}^{-1}\vC_{j,\ii})$ is free of $\vb$ under $\vb=\vb_1=\dots=\vb_K$. For MAB, this corresponds to $\policy_k^{(\n)}$ being a function only of $\freqb$ and the contrasts $\reward_k/\freq_k - \reward_j/\freq_j$, $j\in\SK$.


\subsection{A Regularized Translation Invariant (RTI) contrast}\label{subsec:TI-regularized-schemes}
To accommodate asymptotically valid inference, the sampling scheme must not only be translation invariant, as discussed in Section~\ref{subsec:TI-schemes}, but also satisfy Assumptions~\ref{assm:policy_limit}--\ref{assm:sde_solution}. For the latter, we regularize the translation-invariant contrast $\vx^\trans(\vS_{k,\ii}^{-1}\vC_{k,\ii}-\vS_{j,\ii}^{-1}\vC_{j,\ii})$ to deal with (near-)singularity of $\vS_{k,\ii}$ or $\vS_{j,\ii}$, for instance when the experiment is initialized. Recall the adjugate identity $\vS^{-1}=\operatorname{adj}(\vS)/\det(\vS)$: the singularity enters the contrast only through the determinant in the denominator, and it is this determinant that we regularize.

Fix $\kappa>0$ and define, for $k,j\in\SK$, the \emph{regularized translation-invariant (RTI) contrast}
\begin{equation}\label{eqn:regularized_contrast}
\contrast_{k,j}(\vS,\vC,\vx)\defeq
\frac{\det(\vS_j)\,\vx^\trans\!\operatorname{adj}(\vS_k)\vC_k-\det(\vS_k)\,\vx^\trans\!\operatorname{adj}(\vS_j)\vC_j}
{\max\!\big(\det(\vS_k)\det(\vS_j),\,\kappa\big)}.
\end{equation}
It equals $\vx^\trans(\vS_k^{-1}\vC_k-\vS_j^{-1}\vC_j)$ whenever $\det(\vS_j)\det(\vS_k)\geq\kappa$. Since $\operatorname{adj}(\vS_j)\vS_j=\det(\vS_j)\idm_q$, the numerator is invariant under $\vc\mapsto\vc+\vS\boldsymbol{e}$, while the denominator depends on $\vS$ only; thus $\contrast_{k,j}$ is translation invariant. As a polynomial over a denominator bounded below by $\kappa$, it is continuous and locally Lipschitz on the entire state space, including $\{\det(\vS_j)=0\}$. In the MAB case ($q=1$ and $\vX_\ii\equiv1$), \eqref{eqn:regularized_contrast} reduces to the contrast
\begin{equation}\label{eqn:regularized_contrast_MAB}
\contrast_{k,j}(\vd,\vr)=\frac{d_jr_k-d_kr_j}{\max(d_kd_j,\,\kappa)},
\end{equation}
which equals $r_k/d_k-r_j/d_j$ on $\{d_kd_j\geq\kappa\}$.

The pairwise RTI contrasts introduced above compare each arm with the remaining arms. To obtain a single prediction for each arm without having to choose a benchmark arm, we average its pairwise contrasts. For $k\in\SK$, define the \emph{RTI reward prediction}
\begin{equation}\label{eq:ti-average-tau}
\Contrast_k(\vS,\vC,\vx)
\defeq
\frac{1}{K}\sum_{j=1}^K
\contrast_{k,j}(\vS,\vC,\vx),
\qquad
\vContrast(\vS,\vC,\vx)
\defeq
\big(\Contrast_1,\ldots,\Contrast_K\big)^\trans.
\end{equation}
Let $\hat{\vm}(\vx)
\defeq
\big(
\vx^\trans\vS_1^{-1}\vC_1,\ldots,
\vx^\trans\vS_K^{-1}\vC_K
\big)^\trans$
collect the OLS reward predictions for the $K$ arms, and let
$\boldsymbol\iota=(1,\ldots,1)^\trans$ and $\vP=\idm_K-\boldsymbol\iota\boldsymbol\iota^\trans/K$ denote the
centering projection. Wherever all pairwise regularizations are
inactive, the RTI reward-prediction vector coincides with the
centered OLS prediction vector, $\vContrast(\vS,\vC,\vx)=\vP\hat{\vm}(\vx).$
Equivalently,
$\Contrast_k(\vS,\vC,\vx)
=
\vx^\trans\vS_k^{-1}\vC_k
-\frac{1}{K}\sum_{j=1}^K
\vx^\trans\vS_j^{-1}\vC_j.$

In the same spirit, we regularize the prediction variability and define the \emph{RTI prediction variance} of Arm-$k$ as
\begin{equation}\label{eq:ti-prediction-variance}
\vphi_{k,\kappa}(\vx)
\defeq
\vx^\trans(\vS_k\vee\kappa\idm_q)^{-1}\vx.
\end{equation}
Here, $\vS_k\vee\kappa\idm_q$ is obtained by replacing each eigenvalue $\lambda$ of $\vS_k$ with $\max\{\lambda,\kappa\}$, leaving its eigenvectors unchanged. This quantity is automatically translation invariant because it depends on $\vS_k$ but not on $\vC_k$. In the MAB case, $\vphi_{k,\kappa}=1/(d_k\vee\kappa)$.

The policies below use these two quantities in different ways. RTI tempered-greedy applies a softmax directly to the RTI reward predictions; RTI tempered-LinUCB augments them with the exploration bonuses $\sqrt{\vphi_{k,\kappa}}$; and RTI Thompson sampling combines \(\vContrast\) with $\vphi_{k,\kappa}$'s to define a Gaussian distribution over centered reward predictions, with each arm’s sampling probability given by the probability that its corresponding component is largest.

\subsection{Regularized translation-invariant (RTI) policies} \label{subsec:translation_invariant_policies}
Starting from some well-known sampling policies, we discuss adapted versions based on the regularized translation-invariant contrast~(\ref{eqn:regularized_contrast}). We state each policy in its limiting form $\policy_k$, with hyperparameters $\alpha,\lambda\in(0,\infty)$ and $b\in[0,\infty)$. The finite-sample scheme $\policy_k^{(\n)}$ is the same function of $(\vS_\ii,\vC_\ii,\vX_{\ii+1})$ with $(\alpha,\lambda,b)$ replaced by $(\alpha_\n/\sqrt{\n},\lambda_\n,b_\n/\sqrt{\n})$, where $\alpha_\n/\sqrt{\n}\to\alpha$, $\lambda_\n\to\lambda$, and $b_\n/\sqrt{\n}\to b$.


\subsubsection*{RTI tempered-greedy}
We begin with the tempered-greedy algorithm (see, e.g., \cite{kuang2024weak} for its MAB version), which applies the softmax transformation to the estimated rewards. Replacing the raw reward estimates with the RTI reward predictions, the limiting RTI tempered-greedy policy selects arm $k$ with probability
\begin{equation}\label{eqn:TG}
\policy_k(\vS,\vC,\vx)
    =
\frac{\exp\left(\alpha\Contrast_k(\vS,\vC,\vx)\right)}{\textstyle\sum_{j=1}^K\exp\left(\alpha\Contrast_j(\vS,\vC,\vx)\right)}.
\end{equation}
In the MAB case, $\Contrast_k(\vd,\vr) = \frac{1}{K}\sum_{j=1}^K\contrast_{k,j}(\vd,\vr)$ in~\eqref{eqn:TG}; as $\kappa\downarrow0$, this converges to $1\big/\sum_{j}\exp(\alpha(r_j/d_j-r_k/d_k))$, pointwise on $\bigcap_{j=1}^K\{d_jd_k>0\}$.

\subsubsection*{RTI tempered-LinUCB / tempered-UCB}
Upper-confidence-bound (UCB) algorithms augment estimated rewards with
uncertainty bonuses to encourage the selection of underexplored arms; see, for example,
\cite{lai1985asymptotically,auer2002finite}. LinUCB (\cite{li2010contextual}) extends this principle
to linear contextual bandits by using context-dependent prediction
uncertainty. Following this idea, we construct the
RTI tempered-LinUCB/UCB policy by adding the regularized prediction standard
deviation $\sqrt{\vphi_{k,\kappa}(\vx)}$ to each RTI reward prediction and
applying the softmax transformation to the resulting indices.

Specifically, the limiting RTI LinUCB policy selects arm $k$ with probability
\begin{equation}\label{eqn:LinUCB}
\policy_k(\vS,\vC,\vx)
=
\frac{\textstyle
\exp\!\left\{
\alpha\left(
\Contrast_k(\vS,\vC,\vx)
+
\lambda\sqrt{\vphi_{k,\kappa}(\vx)}
\right)
\right\}
}{
\displaystyle
\textstyle\sum_{j=1}^K
\exp\!\left\{
\alpha\left(
\Contrast_j(\vS,\vC,\vx)
+
\lambda\sqrt{\vphi_{j,\kappa}(\vx)}
\right)
\right\}
}.
\end{equation}

In the MAB special case ($q=1$ and $\vx\equiv1$), we have
$\vphi_{k,\kappa}=1/(d_k\vee\kappa)$, so the RTI tempered-LinUCB policy reduces to an RTI tempered-UCB policy obtained by applying the softmax transformation to the indices $\Contrast_k(\vd,\vr) + \lambda/\sqrt{d_k\vee\kappa}$ for $k\in\SK$.

\subsubsection*{RTI Thompson sampling}


In the linear contextual bandit, classical Thompson sampling \citep[see][]{agrawal2013thompson} combines the OLS predictions in $\hat\vm(\vx)$ and the prediction variances $\vx^\trans\vS_k^{-1}\vx$ with a Gaussian prior into a Gaussian posterior for $\vm(\vx)$, and selects the arm corresponding to the largest coordinate of a draw from this posterior. The Gaussian prior and likelihood are a \emph{working} model: they generate the sampling rule but are not assumed to hold for the rewards. As the prior shrinks each arm's prediction towards zero by an arm-specific amount, this rule is not translation invariant.

Since the arm maximizing a reward vector is unchanged when the same constant is added to all its coordinates, only reward contrasts matter for arm selection. We therefore place the posterior on the centered reward vector $\vP\vm(\vx)$ rather than on $\vm(\vx)$, which loses nothing for arm selection and makes the rule invariant to common translations of the rewards. As the plug-in estimator $\vP\hat\vm(\vx)$ is undefined or unstable when some $\vS_k$ is (near-)singular, we build the posterior from the RTI reward prediction $\vContrast$ of~\eqref{eq:ti-average-tau}, which coincides with $\vP\hat\vm(\vx)$ wherever the regularization is inactive, and from the RTI prediction variances $\vphi_{k,\kappa}(\vx)$ of~\eqref{eq:ti-prediction-variance}. RTI Thompson sampling then selects each arm with the probability that its coordinate of a draw from this posterior is the largest.

Specifically, let the prior on $\vm(\vx)$ be Gaussian with mean $\mfzero$ and precision $b^2\idm_K$, $b\in[0,\infty)$, so that $\vP\vm(\vx)$ has prior precision $b^2$ and $b=0$ is the diffuse (flat) prior. Let $\mPhi(\vx)\defeq\mathrm{diag}\big(\vphi_{1,\kappa}(\vx),\dots,\vphi_{K,\kappa}(\vx)\big)$ and treat $\vContrast$ as an observation of $\vP\vm(\vx)$ with covariance $\vP\mPhi(\vx)\vP$, the covariance of $\vP\hat\vm(\vx)$ when the OLS predictions have covariance $\mPhi(\vx)$. Gaussian conditioning yields the posterior $\vP\vm(\vx)\mid\vContrast\sim\mathcal{N}(\tilde{\vm},\tilde{\bm{\Sigma}})$ with\footnote{We assume here $\nu_x\left(\{0\}\right)=0$, automatic when $\vx$ contains an intercept, so that $\mPhi(\vx)\neq0$ a.s.}
\begin{equation*}
\tilde{\bm{\Sigma}} = \big(b^2\vP + (\vP\,\mPhi(\vx)\,\vP)^{+}\big)^{+},
\qquad
\tilde{\vm} = \tilde{\bm{\Sigma}}\,(\vP\,\mPhi(\vx)\,\vP)^{+}\,\vContrast,
\end{equation*}
where $(\cdot)^{+}$ denotes the Moore--Penrose pseudoinverse, required because $\vP$ has rank $K-1$. RTI Thompson sampling selects Arm-$k$ with probability
\begin{equation}\label{eqn:TS}
\policy_k(\vS,\vC,\vx)
=\mathbb{P}\big\{\mathcal{N}(\vA_k\tilde\vm,\vA_k\tilde{\bm\Sigma}\vA_k^\trans)\le\mfzero\big\},
\qquad \vA_k=\idm_K-\boldsymbol\iota\boldsymbol e_k^\trans,
\end{equation}
the probability that the $k$-th coordinate of a posterior draw is the largest.

\subsubsection*{RTI exploration sampling}
Exploration sampling, introduced by \citet{kasy2021adaptive} as a modification of Thompson sampling for adaptive experiments aimed at best-arm identification, reweights the posterior probability that each arm is optimal. Writing $\policy_k^{\mathrm{RTI\text{-}TS}}$ for the RTI Thompson-sampling probability defined in \eqref{eqn:TS}, the corresponding RTI exploration-sampling rule selects arm $k$ with probability
\[
\policy_k^{\mathrm{RTI\text{-}Explore}}(\vS,\vC,\vx)
=
\frac{
\policy_k^{\mathrm{RTI\text{-}TS}}(\vS,\vC,\vx)
\left\{1-\policy_k^{\mathrm{RTI\text{-}TS}}(\vS,\vC,\vx)\right\}
}{
\textstyle\sum_{j=1}^K
\policy_j^{\mathrm{RTI\text{-}TS}}(\vS,\vC,\vx)
\left\{1-\policy_j^{\mathrm{RTI\text{-}TS}}(\vS,\vC,\vx)\right\}
}.
\]
Because all RTI-TS choice probabilities lie strictly between zero and one, the RTI-Explore
transformation is well defined and preserves translation invariance.

The following theorem shows that all regularized translation-invariant sampling schemes introduced above also satisfy Assumption~\ref{assm:policy_limit} and Assumption~\ref{assm:sde_solution}.
\begin{theorem}\label{thm:well-posedness}
Fix $\kappa>0$. Let $\policy_k^{(\n)}$ be the RTI tempered-greedy, tempered-LinUCB, Thompson, or exploration sampling scheme. Then, for every $k\in\SK$: $\policy_k^{(\n)}$ and $\policy_k$ are translation invariant and Assumptions~\ref{assm:policy_feasible}, \ref{assm:policy_limit}(a), and \ref{assm:sde_solution} hold. The same conclusions hold for the clipped schemes \eqref{eqn:uniform_mixing}, which additionally satisfy Assumption~\ref{assm:policy_limit}(b).
\end{theorem}
The proof is in Appendix~\ref{S-appsec:TI_algorithms}: it composes the boundedness and local-Lipschitz properties of $\contrast_{k,j}$ and $\vphi_{k,\kappa}$ with those of the softmax and Gaussian-orthant maps; existence and uniqueness of the SDE solution follow from the local-Lipschitz/linear-growth existence theorem \citep[Thm~IV.3.1]{ikeda2014stochastic} (the regularization coefficients being uniformly bounded; see also \citealp[Ch.~5.2]{karatzas2012brownian}).

\section{Statistical applications: non-contextual bandits}\label{sec:MAB}
It is well known that most traditional sampling schemes (like Thompson sampling) lead to non-standard statistical inference, even for the simplest reward distributions (like normal location models). In this section, we study often-used statistics in the non-contextual multi-armed bandit (MAB) problem with $K$ arms. The potential outcomes are generated by
\begin{align*}
\Zkt = \mu_k + \ekt,  \qquad k\in\SK,
\end{align*}
where the innovations $\ekt$ are as in Section~\ref{sec:sampling} (mean zero, unit variance, density $\f_k$ with Fisher information for location $\FJ_k$) and Assumption~\ref{assm:DQM} holds with $\pintb = (\mu_1,\dots,\mu_K)^\trans$. Note that the Fisher information matrix $\FJbpintk$ is now singular: its only non-zero entry is the scalar $\FJ_k$ in position $(k,k)$, and it does not depend on $\pintb$.

Recall the equal-arms localization of Section~\ref{sec:sampling}. In the non-contextual case ($q = 1$, $X_t \equiv 1$), it specializes to
\begin{equation} \label{eqn:equalarms_asymptotics}
\mu_{k,\n} = \mu + \frac{\m_k}{\sqrt{\n}}, \qquad \m_k = \m + \dlpint_k,
\end{equation}
around a common value $\mu$, Taking Arm~1 as the control, we set $\dlpint_1 = 0$, so
that $\m$ is the common local reward level and $\dlpint_k$ the local effect of Arm~$k$
relative to the control.

Section~\ref{subsec:MABLimit} specializes the general results of Section~\ref{sec:limitexperiment} to this non-contextual $K$-armed setting, while Section~\ref{subsec:MAB_comparetwoarms} discusses various statistics testing equality of arms' expected rewards. Section~\ref{subsec:montecarlo_MAB_twoarms} corroborates our theory with simulation results.

\subsection{MAB limit experiment}\label{subsec:MABLimit}

Specializing \eqref{eqn:statistic_general_CS} to the non-contextual model, define
\begin{equation} \label{eqn:reward_frequency_r1r2}
\rewardd_{r_1;k,\ii}
\defeq \frac{1}{\sqrt{\n}}\sum_{s=1}^{\ii}
\policy_{k,s}^{-r_1}\indicator_{\{A_s=k\}}(\Y_s-\mu),
\qquad
\freqd_{r_2;k,\ii}
\defeq \frac{1}{\n}\sum_{s=1}^{\ii}
\policy_{k,s}^{-r_2}\indicator_{\{A_s=k\}},
\end{equation}
where $\policy_{k,s}$ is the conditional probability of selecting
Arm~$k$ at round $s$, given $\filtr_{s-1}$. The unweighted accumulated
rewards and arm-pull frequencies are
\begin{equation} \label{eqn:reward_frequency}
\reward_{k,\ii}\defeq\rewardd_{0;k,\ii},
\qquad
\freq_{k,\ii}\defeq\freqd_{0;k,\ii}.
\end{equation}
Write $\rewardb_{\ii}=(\reward_{1,\ii},\dots,\reward_{K,\ii})^\trans$
and $\freqb_{\ii}=(\freq_{1,\ii},\dots,\freq_{K,\ii})^\trans$.
We use the MAB versions of the RTI policies in
Section~\ref{subsec:translation_invariant_policies}, with
$\policy_{k,\ii+1}=\policy_k^{(\n)}(\freqb_{\ii},\rewardb_{\ii})$.
Changing the centering value $\mu$ shifts $\rewardb_{\ii}$ by a
scalar multiple of $\freqb_{\ii}$. Translation invariance removes
this dependence, so these policies can be implemented without
knowing $\mu$.

Under Assumptions~\ref{assm:policy_feasible}--\ref{assm:sde_solution},
Proposition~\ref{prop:likelihoodweakconvergence} yields joint
convergence of the likelihood ratios and the weighted statistics
in~\eqref{eqn:reward_frequency_r1r2}, which we use for testing in
Section~\ref{subsec:MAB_comparetwoarms}. The limit experiment has
laws $\law_{\m_1,\dots,\m_K}$ and the structural representation
below. It follows from Theorem~\ref{thm:structurallimitexperiment_f}
with $\vu=(\freqb,\rewardb)$ and
$\Epolicyb^{\vartriangle}_k=\FJ_k\policy_k$,
$\Epolicyb^{\vc}_{r_1;k}=\policy_k^{1-2r_1}$,
$\Epolicyb^{\vx}_{r_1;k}=\policy_k^{1-r_1}$, and
$\Epolicyb^{\vs}_{r_2;k}=\policy_k^{1-r_2}$,
all evaluated at $(\freqb,\rewardb)$.

\begin{ancillary} \label{ancillary:structurallimitexperiment_MAB}
The limit experiment $\experiment$ is described as follows. We observe $\CS_{k}$, $\QV_{k}$, $\rewardd_{r_1;k}$, and $\freqd_{r_2;k}$ $k\in\SK$, generated by
\begin{equation} \label{eqn:SDE_MAB}
\begin{aligned}
\dd\W_{\score_k}(\uu) &= \m_k\FJ_{k}\sqrt{\policy_k(\freqb(\uu),\rewardb(\uu))}\dd\uu + \dd\B_{\score_k}(\uu),  \\
\dd\CS_{k}(\uu) &= \sqrt{\policy_k(\freqb(\uu),\rewardb(\uu))}\dd\W_{\score_k}(\uu),  \\
\dd\QV_{k}(\uu) &= \FJ_{k}\policy_k(\freqb(\uu),\rewardb(\uu))\dd\uu, \\
\dd\Wek(\uu) &= \m_k\sqrt{\policy_k(\freqb(\uu),\rewardb(\uu))}\dd\uu + \dd\B_{\e_k}(\uu),  \\
\dd\rewardd_{r_1;k}(\uu) &= \sqrt{\policy_k^{1-2r_1}(\freqb(\uu),\rewardb(\uu))}\dd\Wek(\uu),  \\
\dd\freqd_{r_2;k}(\uu) &= \policy_k^{1-r_2}(\freqb(\uu),\rewardb(\uu))\dd\uu,
\end{aligned}
\end{equation}
for $\uu\in[0,1]$, with $\m_k=\m+\dlpint_k$ as in~(\ref{eqn:equalarms_asymptotics}). Here, $\vB_k=(\B_{\e_k},\B_{\score_k})^\trans$ is a zero-drift Brownian motion with covariance $\left(\begin{smallmatrix}1 & 1 \\ 1 & \FJ_{k}\end{smallmatrix}\right)$, and Brownian motions associated to different arms are independent.\footnote{This is the covariance matrix of $(\e_k,\scorek(\e_k))^\trans$: $\Exp[\e_k^2]=1$, $\Exp[\e_k\scorek(\e_k)]=1$, and $\Exp[\scorek(\e_k)^2]=\FJ_k$.} The limiting accumulated-reward processes $\reward_k\defeq\rewardd_{0;k}$ and arm-pulling frequency processes $\freq_k\defeq\freqd_{0;k}$ drive the limiting sampling scheme through $\rewardb(\uu)=(\reward_1(\uu),\dots,\reward_K(\uu))^\trans$ and $\freqb(\uu)=(\freq_1(\uu),\dots,\freq_K(\uu))^\trans$, respectively.
\end{ancillary}

\begin{remark} \label{remark:KWFG}
Let $r_1 = r_2 = 0$. The accumulated rewards and arm-pulling frequencies in Result~\ref{ancillary:structurallimitexperiment_MAB} then reduce to
\begin{equation} \label{eqn:SDE_MAB_reward_frequency}
\begin{aligned}
\dd\Wek(\uu) &= \m_k\sqrt{\policy_k(\freqb(\uu),\rewardb(\uu))}\dd\uu + \dd\B_{\e_k}(\uu),  \\
\dd\reward_{k}(\uu) &= \sqrt{\policy_k(\freqb(\uu),\rewardb(\uu))}\dd\Wek(\uu),  \\
\dd\freq_{k}(\uu) &= \policy_k(\freqb(\uu),\rewardb(\uu))\dd\uu,
\end{aligned}
\end{equation}
for $k\in\SK$. Plugging $\dd\Wek(\uu)$ into $\dd\reward_{k}(\uu)$ yields
\begin{equation} \label{eqn:xuwager_eq3.1}
\begin{aligned}
\dd\reward_{k}(\uu) = \m_k\policy_k(\freqb(\uu),\rewardb(\uu))\dd\uu + \sqrt{\policy_k(\freqb(\uu),\rewardb(\uu))}\dd\B_{\e_k}(\uu),
\end{aligned}
\end{equation}
which corresponds to the diffusion approximation for these two statistics developed by \cite{fan2025diffusion} and \cite{kuang2024weak}.
\end{remark}

\begin{remark} \label{remark:MAB_Gaussian}
Under standard normal distributions (i.e., $\f_1 = \dots = \f_K = \phi$), for all $k\in\SK$ we have $\score_{k}(\e) = \e$, $\FJ_{k} = 1$, $\CS_{k,\ii} = \reward_{k,\ii}$ and $\QV_{k,\ii} = \freq_{k,\ii}$, for $t = 1,\dots,\n$. Correspondingly, in the limit experiment, Gaussianity implies $\W_{\score_k} = \Wek$, $\CS_{k}(\uu) = \reward_{k}(\uu)$, and $\QV_{k}(\uu) = \freq_{k}(\uu)$, for $k\in\SK$ and $\uu\in[0,1]$.
Note that, even in this Gaussian setting, the limit experiment is generally not a Gaussian shift experiment. An exception is when sampling is non-adaptive, i.e., $\policy_k$ is constant over both arguments.
\end{remark}



\subsection{Tests for comparing two arms} \label{subsec:MAB_comparetwoarms}
In the main text we consider the problem of comparing two arms (e.g., treatment versus control), i.e., testing
\begin{equation} \label{eqn:hyp_2}
H_0\colon \dlpint_2 = 0
    \quad\text{versus}\quad
H_1\colon \dlpint_2 > 0,
\end{equation}
with $\m$ and, for $K>2$, the remaining effects $\dlpint_3,\dots,\dlpint_K$ as nuisance parameters. Evaluating a single arm and testing equality of all arms are in Appendices~\ref{S-subsec:MAB_evaluateonearm} and~\ref{S-subsec:MAB_omnibus}, respectively.

For all cases the sequence of sampling schemes is assumed to be RTI as in Theorem~\ref{thm:well-posedness}. For brevity, we write $\policy_{k,\ii}$ for $\policy_k^{(\n)}(\freqb_{\ii-1}, \rewardb_{\ii-1})$ and $\policy_k(\uu)$ for $\policy_k(\freqb(\uu), \rewardb(\uu))$. In the limit experiment, (\ref{eqn:SDE_MAB}) yields, for $k\in\SK$,
\begin{equation} \label{eqn:SDERepar}
\dd\rewardd_{r_1;k}(\uu)
    =
\m_k\policy_k^{1-r_1}(\uu)\dd\uu +
    \policy_k^{1/2-r_1}(\uu)\dd\B_{\e_k}(\uu),
\end{equation}
which is the representation we work with below.

The nuisance parameters are dealt with in turn. Focusing first on the common level $\m$, we introduce a family of two-sample (TS) statistics, whose distributions do not depend on the nuisance parameters, that includes Student's~$t$, Adaptively Weighted (AW), and Inverse Probability Weighted (IPW) statistics, and a variant that weights the reward increments by their instantaneous inverse sampling probabilities. By construction, the finite-sample counterparts of these statistics are likewise distribution-free with respect to the global common reward parameter $\mu$, a desirable property in practice. Turning to the remaining nuisance parameters $\dlpint_3,\dots,\dlpint_K$, \emph{which arise only for $K>2$}, we find that the distributions of the former do depend on them, whereas that of the latter does not. The representation~(\ref{eqn:SDERepar}) makes the limit distributions of all these statistics explicit.

Since all statistics but the TS-$t$ involve inverse sampling probabilities, we impose Assumption~\ref{assm:policy_limit}(b) throughout. In the sequence of experiments, we pull each arm once during the first $K$ rounds so that all statistics are well initialized. Section~\ref{subsec:montecarlo_MAB_twoarms} shows that the asymptotic results are accurate in finite samples.

\subsubsection*{A family of distribution-free two-sample statistics}
Let $r_1 = r_2 = r$ in (\ref{eqn:reward_frequency_r1r2}). Consider the family of statistics
\begin{equation} \label{eqn:distributionfree_stat}
\stat^{\text{TS-DF}}_{r;\n} \defeq \frac{\rewardd_{r;2,\n}}{\freqd_{r;2,\n}} - \frac{\rewardd_{r;1,\n}}{\freqd_{r;1,\n}},
\end{equation}
\begin{sloppypar}
which weakly converges to $\stat^{\text{TS-DF}}_r \defeq \frac{\int_0^1 \policy_2^{-r}(\uu)\dd\reward_{2}(\uu)}{\int_0^1 \policy_2^{-r}(\uu)\dd\freq_{2}(\uu)} - \frac{\int_0^1 \policy_1^{-r}(\uu)\dd\reward_{1}(\uu)}{\int_0^1 \policy_1^{-r}(\uu)\dd\freq_{1}(\uu)}.$ That $\stat^{\text{TS-DF}}_r$ is distribution-free w.r.t $\m$, follows as, using~(\ref{eqn:SDERepar}),
\begin{equation}
\frac{\textstyle \int_0^1 \policy_k^{-r}(\uu)\dd\reward_k(\uu)}{\textstyle \int_0^1 \policy_k^{-r}(\uu)\dd\freq_k(\uu)}
= \m_k + \frac{\textstyle \int_0^1\policy_k^{1/2-r}(\uu)\dd\B_{\e_k}(\uu)}{\textstyle \int_0^1 \policy_k^{1-r}(\uu)\dd\uu},
\end{equation}
for $k\in\SK$: $\m$ cancels in the drift, $\m_2-\m_1=\dlpint_2$, and stays out of the noise terms because the RTI policy $\policy$ is free of $\m$ (see the discussion below Corollary~\ref{cor_TI_I}). Translation invariance is sufficient here, not necessary. As
$\frac{\rewardd_{r;k,\n}}{\freqd_{r;k,\n}}
=
-\sqrt{\n}\mu + \sqrt{\n}\frac{\sum_{\ii=1}^\n \policy_{k,\ii}^{-r}\indicator_{\{A_{\ii} = k\}}\Y_{\ii}}{\sum_{\ii=1}^\n \policy_{k,\ii}^{-r}\indicator_{\{A_{\ii} = k\}}},$
$\stat^{\text{TS-DF}}_{r;\n}$ is also distribution free w.r.t the global common reward parameter $\mu$.
\end{sloppypar}

The above family of statistics encompassed the two-sample t, AW, and IPW tests (using $r=0$, $1/2$, and $1$, respectively). We discuss these now.

\subsubsection*{Two-sample Student's t-test ($r=0$)}
The classical two-sample Student's t-test corresponds to
\begin{align*}
\stat_\n^{\text{TS-t}} \defeq \left(\frac{\reward_{2,\n}}{\freq_{2,\n}} - \frac{\reward_{1,\n}}{\freq_{1,\n}}\right) \bigg/ \sqrt{\frac{1}{\freq_{1,\n}} + \frac{1}{\freq_{2,\n}}},
\end{align*}
\begin{sloppypar}
which weakly converges to
$\stat^{\text{TS-t}} \defeq \left(\frac{\reward_{2}(1)}{\freq_{2}(1)} - \frac{\reward_{1}(1)}{\freq_{1}(1)}\right) \Big/ \sqrt{\frac{1}{\freq_{1}(1)} + \frac{1}{\freq_{2}(1)}}.$
The distribution of the numerator of $\stat^{\text{TS-t}}$, under $\law_{\m_1,\dots,\m_K}$, converges to
\[
\mathcal{L}\left( \dlpint_2 + \frac{\textstyle \int_0^1\sqrt{\textstyle \policy_2(\uu)}\dd\B_{\e_2}(\uu)}{\textstyle \int_0^1\policy_2(\uu)\dd\uu} - \frac{\textstyle \int_0^1\sqrt{\policy_1(\uu)}\dd\B_{\e_1}(\uu)}{\textstyle \int_0^1\policy_1(\uu)\dd\uu}\right).
\]
The $\stat^{\text{TS-t}}$ statistic is generally not asymptotically $\mathcal{N}(0,1)$ distributed, but it is in  case $\policy$ is non-adaptive. The statistic is distribution free under the null ($\dlpint_2=0$) for $K=2$. However, for $K>2$ its distribution will depend on $\dlpint_k=0$, $k>2$. We come back to this in Section~\ref{subsubK4}. Note that the above expression easily allows for power calculations.
\end{sloppypar}

\subsubsection*{Two-sample Adaptively-Weighted (AW) test ($r=1/2$)}
The two-sample adaptively-weighted statistic uses $r = 1/2$, i.e.,
\begin{equation}
\stat_\n^{\text{TS-AW}}
\defeq \bigg(\frac{\rewardd_{1/2;2,\n}}{\freqd_{1/2;2,\n}} - \frac{\rewardd_{1/2;1,\n}}{\freqd_{1/2;1,\n}}\bigg) \bigg/ \sqrt{\frac{1}{(\freqd_{1/2;2,\n})^2} + \frac{1}{(\freqd_{1/2;1,\n})^2}}
\end{equation}
\begin{sloppypar}
which weakly converges to $\stat^{\text{TS-AW}} \defeq \stat^{\text{TS-DF}}_{1/2}
\big/ \sqrt{\freqd_{1/2;1}(1)^{-2} + \freqd_{1/2;2}(1)^{-2}}$, with $\freqd_{1/2;k}(1)=\int_0^1\policy_k^{1/2}(\uu)\dd\uu$. This is the AW treatment-effect statistic of \citet[Eq.~(23)]{hadad2021confidence}, written in our notation. By~(\ref{eqn:SDERepar}), the numerator of $\stat^{\text{TS-AW}}$ under $\law_{\m_1,\dots,\m_K}$ equals
\[
\dlpint_2 + \frac{\B_{\e_2}(1)}{\freqd_{1/2;2}(1)} - \frac{\B_{\e_1}(1)}{\freqd_{1/2;1}(1)}.
\]
Under $H_0$, the standardized statistic is thus $w_2\B_{\e_2}(1)-w_1\B_{\e_1}(1)$ with $w_1^2+w_2^2=1$, where the weights $w_k$ depend on the path only through the ratio $\freqd_{1/2;2}(1)/\freqd_{1/2;1}(1)$. It is exactly standard normal when this ratio is deterministic. This is the analogue of the variance-ratio condition in \citet[Theorem~4, Eq.~(21)]{hadad2021confidence}, which requires the two arms' variance estimates to have a deterministic limiting ratio in $[0,\infty]$. The obstacle is thus not adaptivity itself, but randomness in the limiting allocation. The condition holds under a fixed alternative with a gap between the arms, where the allocation converges to a deterministic limit. It fails under the null and contiguous local alternatives --the regimes of interest here -- where the limiting allocation is nondegenerate and random. Normality is then no longer guaranteed, though the simulations in Section~\ref{subsec:montecarlo_MAB_twoarms} show only a small deviation.
\end{sloppypar}

\subsubsection*{Two-sample Inverse Propensity Weighting (IPW) test ($r=1$)}
The two-sample IPW statistic is defined as
\begin{align*}
\stat_\n^{\text{TS-IPW}}
\defeq \frac{\rewardd_{1;2,\n}}{\freqd_{1;2,\n}} - \frac{\rewardd_{1;1,\n}}{\freqd_{1;1,\n}},
\end{align*}
\begin{sloppypar}
which weakly converges to
$\stat^{\text{TS-IPW}}
\defeq \int_0^1\policy_2^{-1}(\uu)\dd\reward_{2}(\uu) - \int_0^1\policy_1^{-1}(\uu)\dd\reward_{1}(\uu),$
noting that the denominators equal one as $\int_0^1\policy_k^{-1}(\uu)\dd\freq_{k}(\uu) = 1$ for $k\in\SK$. The distribution of $\stat^{\text{TS-IPW}}$, under $\law_{\m_1,\dots,\m_K}$, is
\[
\mathcal{L}\left( \dlpint_2 + \int_0^1\policy_2^{-1/2}(\uu)\dd\B_{\e_2}(\uu) - \int_0^1\policy_1^{-1/2}(\uu)\dd\B_{\e_1}(\uu)\right),
\]
which again is generally not normal as the processes $\policy_k^{-1/2}$ and $\B_{\e_k}$ are dependent.
\end{sloppypar}

\subsubsection*{A distribution-free two-sample statistic that is exactly normal}
The statistics above take ratios only after aggregating over the horizon and are generally not normally distributed.
Instead, weight each reward increment by its instantaneous inverse sampling probability. Since $\dd\freq_k(\uu)=\policy_k(\uu)\dd\uu$, (\ref{eqn:SDERepar}) gives
\begin{equation}\label{eqn:instantaneous_contrast}
\frac{\dd\reward_{2}(\uu)}{\policy_{2}(\uu)}-\frac{\dd\reward_{1}(\uu)}{\policy_{1}(\uu)}
=\dlpint_2\dd\uu+\frac{\dd\B_{\e_2}(\uu)}{\sqrt{\policy_2(\uu)}}-\frac{\dd\B_{\e_1}(\uu)}{\sqrt{\policy_1(\uu)}}.
\end{equation}
The common local mean $m$ cancels from the drift at every instant; integrating~(\ref{eqn:instantaneous_contrast}) with weight $1$ returns $\stat^{\text{TS-IPW}}$. The conditional variance of~(\ref{eqn:instantaneous_contrast}) is $\big(\policy_1^{-1}(\uu)+\policy_2^{-1}(\uu)\big)\dd\uu$, which suggests weighting it by $h(\uu)\defeq\big(\policy_1^{-1}(\uu)+\policy_2^{-1}(\uu)\big)^{-1/2}$ and integrating, giving
\begin{align}
\stat^{\text{TS-AW-P}}
&\defeq \int_0^1 h(\uu)\left\{
    \frac{\dd\reward_2(\uu)}{\policy_2(\uu)}
    -\frac{\dd\reward_1(\uu)}{\policy_1(\uu)}
\right\}  \label{eqn:TSAWpivot} \\
&= \dlpint_2\int_0^1 h(\uu)\dd\uu
 +\int_0^1\!\sqrt{\frac{\policy_1(\uu)}{\policy_1(\uu)+\policy_2(\uu)}}\dd\B_{\e_2}(\uu)
-\int_0^1\!\sqrt{\frac{\policy_2(\uu)}{\policy_1(\uu)+\policy_2(\uu)}}\dd\B_{\e_1}(\uu). \nonumber
\end{align}
The two integrands are predictable and their squares sum to one at every instant, so the martingale part in~(\ref{eqn:TSAWpivot}) has quadratic variation $\uu$; by L\'evy's characterization it is a standard Brownian motion, whence $\stat^{\text{TS-AW-P}}\sim\mathcal{N}(0,1)$ exactly under $H_0$ for any predictable sampling scheme with $\policy_1,\policy_2>0$ almost everywhere. Its finite-sample version is given in~(\ref{S-eqn:finite_MAB_pivot}) of Appendix~\ref{S-appsec:finite_sample_pivots}. The statistic is thus distribution-free w.r.t.\ \emph{all} nuisance parameters $\m$ and $\dlpint_3,\dots,\dlpint_K$ (i.e., \emph{pivotal}).
This is the direct treatment-effect score of \citet[Eq.~(19)]{hadad2021confidence} with the predictable variance-stabilizing weight of \citet{LuedtkeVanderlaan2016}.
\citet{hadad2021confidence} establish asymptotic normality of this direct approach. Stabilization was introduced for ties, i.e., a non-unique optimal treatment rule.
The price is power: see Section~\ref{subsec:montecarlo_MAB_twoarms}.

\subsubsection*{A power envelope ($K=2$)}
Assume Gaussian rewards (Remark~\ref{remark:MAB_Gaussian}). Under an RTI scheme, a shift $c$ of the common level $\m$ maps $(\reward_k,\freq_k)$ to $(\reward_k+c\freq_k,\freq_k)$; integrating the limit likelihood ratio over $c$ gives the most powerful (compare~\citet{andrewsploberger1994}) invariant test against $\dlpint_2=\bar\dlpint_2>0$, with statistic
\begin{equation}\label{eqn:TSenvelope}
\stat^{\text{NP}}(\bar\dlpint_2) \defeq \bar\dlpint_2\,\mathcal{I}^{1/2}\,\stat^{\text{TS-t}} - \tfrac12\bar\dlpint_2^2\,\mathcal{I},
\qquad \mathcal{I}\defeq\big(\freq_1(1)^{-1}+\freq_2(1)^{-1}\big)^{-1}.
\end{equation}
Its null distribution is free of $\m$, so critical values can be simulated at $\m=0$, and its power bounds that of every invariant test, including the four TS tests above. Note that the NP test reduces to the TS-$t$ test only if $\mathcal{I}$ is deterministic (non-adaptive sampling); otherwise~(\ref{eqn:TSenvelope}) measures how far the TS-$t$ test falls short.

\subsection{Simulation results}\label{subsec:montecarlo_MAB_twoarms}
We use Monte Carlo simulations to corroborate our theoretical results on size and assess the power of tests for comparing two arms. We consider bandit problems first with $K = 2$ and then with $K = 4$ arms. In particular we show that, for $K=4$, only the TS-AW pivot-based test is valid over the entire (composite) null. The limiting distributions of traditionally used statistics depend on the reward distribution for the $k=3,4$ arms.

To approximate distributions of statistics in the limit experiment (which reflect asymptotic distributions of statistics in the sequence),  we simulate the SDEs~(\ref{eqn:SDE_MAB}) using an Euler scheme with $200$ grid points. At the significant level $\alpha=5\%$, we use the standard-normal quantile for the TS-AW pivot. For the other two-sample statistics below when $K=2$, their construction and RTI sampling make their null distributions free of $\m$, so we simulate critical values in the limit experiment at $\m_1=\m_2=0$. All results use $20{,}000$ replications.

We use the RTI sampling schemes introduced in Section~\ref{subsec:translation_invariant_policies}: Thompson sampling with $b = 1/20$, tempered-greedy with $\alpha = 1$, and tempered-UCB with $\alpha = 1$ and $\lambda = 1$. In the sequence of experiments, we set $b_\n = \sqrt{\n}/20$, $\alpha_\n = \sqrt{\n}$, and $\lambda_\n = 1$, consistent with the asymptotic scaling specified in Section~\ref{subsec:translation_invariant_policies}. We also include classical Thompson sampling with $b = 1/20$ to illustrate how a lack of translation invariance can lead to invalid tests. Following Remark~\ref{remark:clipping}, we apply clipping at level $\policy_{\min}=0.01$ to all sampling schemes.

\subsubsection{\texorpdfstring{$K = 2$}{K = 2}}

\begin{figure}[ht]
\centering
\centering

\includegraphics[width=\linewidth]{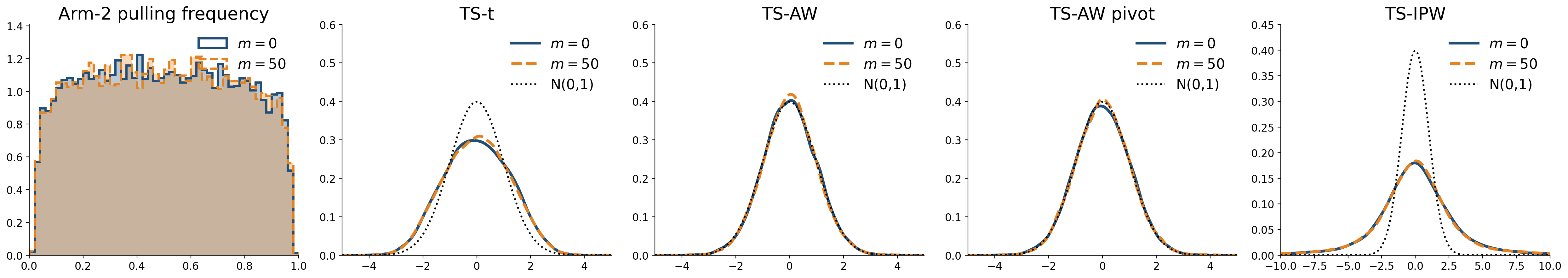}

\includegraphics[width=\linewidth]{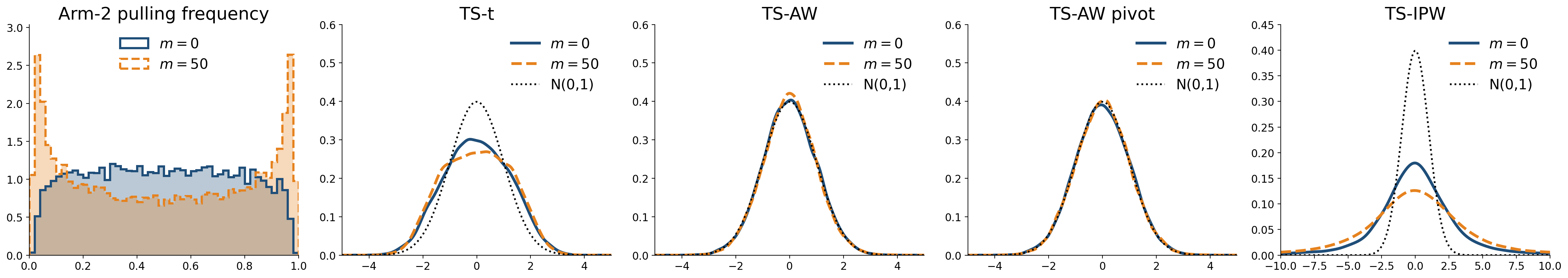}
\caption{{\small Histograms of Arm-2 pulling frequency and density estimates for TS-t, TS-AW, TS-AW pivot, and TS-IPW (from left to right) in the MAB limit experiment with $K = 2$ under RTI (top) and classical (bottom) Thompson sampling. Parameter settings: $\m_1 = \m_2 = 0$ (solid blue) and $\m_1 = \m_2 = 50$ (dashed orange). Dotted black curves show the standard-normal density.}}
\label{fig:MABlim_statD_Thompson_ThompsonInv}
\end{figure}

Figure~\ref{fig:MABlim_statD_Thompson_ThompsonInv} compares the null distributions of Arm-2 pulling frequency and the TS-t, TS-AW, TS-AW pivot, and TS-IPW statistics under RTI (top) and classical (bottom) Thompson sampling, with $\m_1 = \m_2 = \m$ and $\m = 0,50$. Under RTI Thompson sampling, the allocation histogram and all four statistic densities remain essentially unchanged as $\m$ varies. This corroborates our theoretical result that these statistics are distribution-free w.r.t $\m$ under translation-invariant sampling. Both TS-AW and the TS-AW pivot appear standard normal, although exact normality is guaranteed only for the pivot by~(\ref{eqn:TSAWpivot}). TS-t and TS-IPW show clear departures from standard normality, so invariance to $\m$ alone does not justify standard-normal critical values. Under classical, non-translation-invariant Thompson sampling, changing $\m$ substantially alters the arm-pulling frequency histogram and the TS-t and TS-IPW distributions, illustrating their dependence on the nuisance parameter when the sampling scheme is not translation invariant. TS-AW and the TS-AW pivot nevertheless appear standard normal as $\m$ varies under both sampling schemes. Further size results for TS-t, TS-AW, TS-AW pivot, and TS-IPW are reported in Table~\ref{S-table:sizes} of Appendix~\ref{S-appsec:additionalsimulations}.

\begin{figure}[!htb]
\centering

\includegraphics[width=\linewidth]{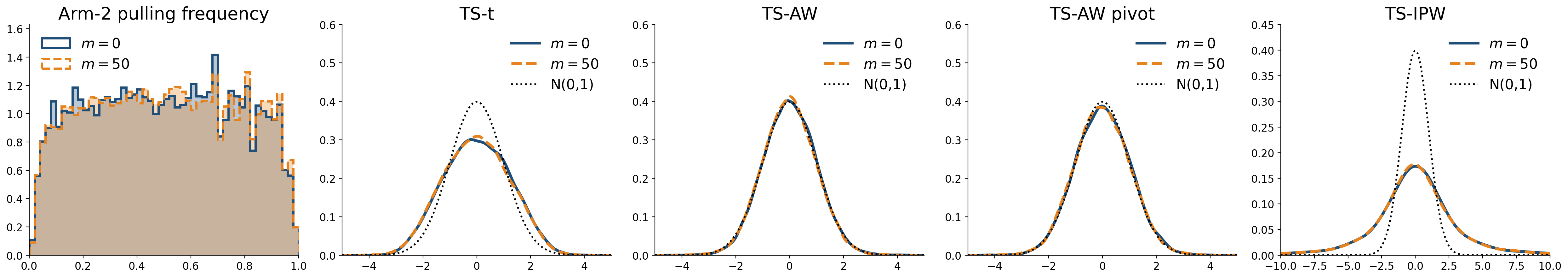}

\includegraphics[width=\linewidth]{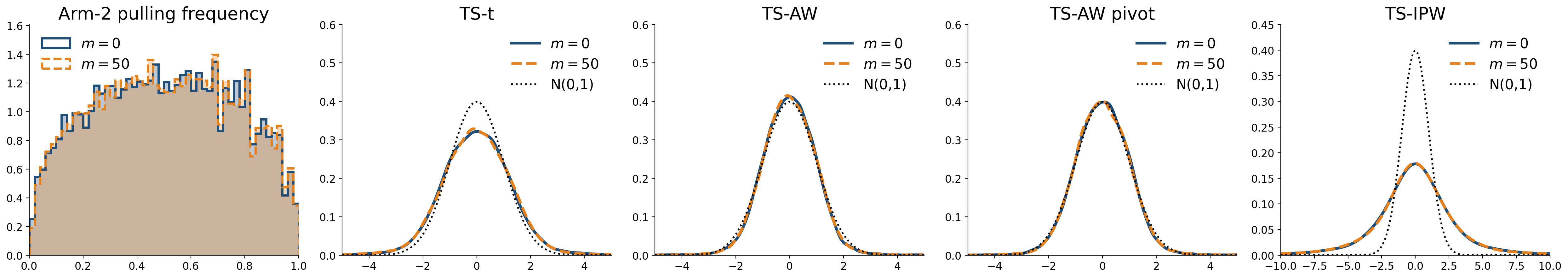}
\caption{{\small Histograms of Arm-2 pulling frequency and density estimates for TS-t, TS-AW, TS-AW pivot, and TS-IPW (from left to right) in finite-sample MAB experiments with $K = 2$ and $\n = 200$ under RTI Thompson sampling, with normal (top) and Student $t_3$ (bottom) reward errors. Parameter settings: $\m_1 = \m_2 = 0$ (solid blue) and $\m_1 = \m_2 = 50$ (dashed orange).}}
\label{fig:MABseq_statD_ThompsonInv}
\end{figure}

Figure~\ref{fig:MABseq_statD_ThompsonInv} presents the finite-sample null distributions under RTI Thompson sampling with $\n = 200$, for normal (top) and Student $t_3$ (bottom) reward errors.\footnote{Note that the $t_3$ errors have finite variance but violate our $4+\delta$ moment assumption.} For both error distributions, the results remain essentially unchanged between $\m = 0$ and $\m = 50$ and closely resemble their limiting counterparts in the top panel of Figure~\ref{fig:MABlim_statD_Thompson_ThompsonInv}. In particular, TS-AW and the TS-AW pivot are already close to standard normal at this sample size.

\begin{figure}[ht]
\centering

\includegraphics[width=\linewidth]{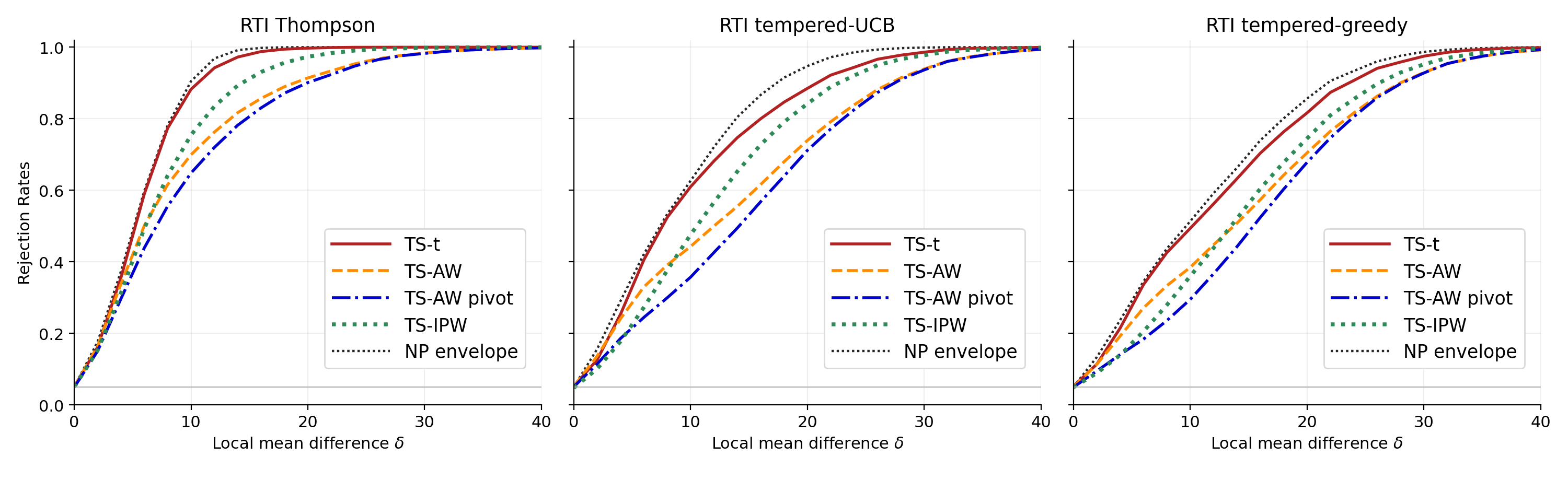}
\caption{{\small Asymptotic powers of TS-t (solid red), TS-AW (dashed orange), TS-AW pivot (dash-dotted blue), and TS-IPW (dotted green) in the MAB limit experiment under RTI Thompson sampling (left), RTI tempered-UCB (middle), and RTI tempered-greedy (right). The thin dotted black curves show the power envelope~(\ref{eqn:TSenvelope}).}}
\label{fig:MABlim_powerDalgorithms}
\end{figure}

Figure~\ref{fig:MABlim_powerDalgorithms} compares the asymptotic powers of TS-t, TS-AW, TS-AW pivot, and TS-IPW under the three RTI sampling schemes. TS-t achieves the highest power under each scheme.
TS-AW and its pivot have similar power, with a small additional power loss for the pivot in exchange for its exactly standard-normal null distribution. TS-IPW generally outperforms both AW-based tests at moderate to large effects, but falls below TS-AW for small effects. Across sampling schemes, power is generally highest under RTI Thompson sampling, followed by RTI tempered-UCB and RTI tempered-greedy. This pattern suggests that sampling schemes that favor exploration, such as RTI Thompson sampling, tend to yield higher testing power in this setting

\begin{figure}[ht]
\centering

\includegraphics[width=\linewidth]{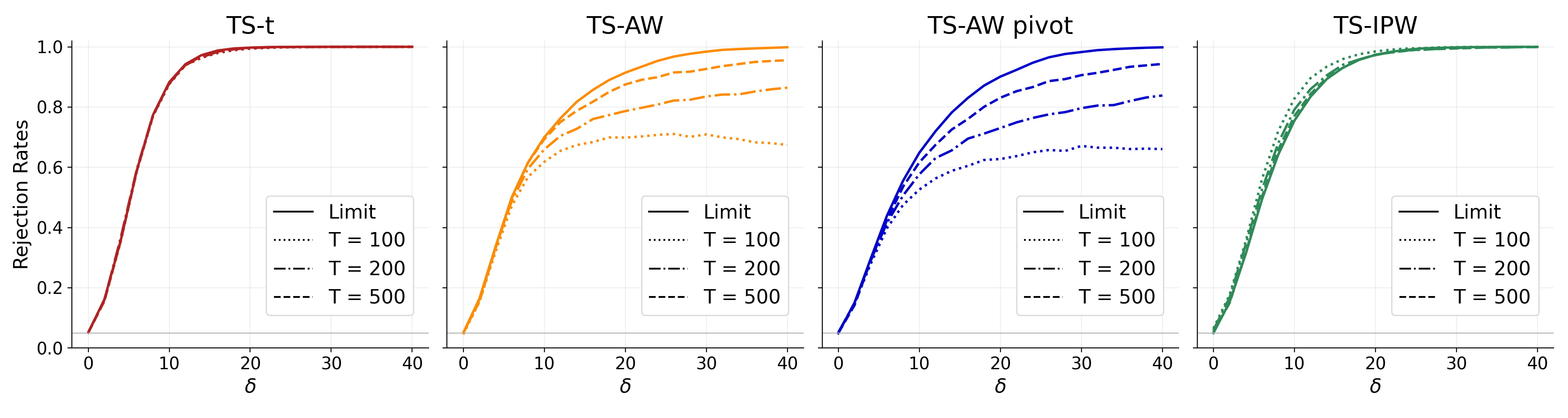}
\caption{{\small MAB asymptotic (solid) and finite-sample powers for $\n = 100$ (dotted), $200$ (dash-dotted), and $500$ (dashed) of TS-t (red), TS-AW (orange), TS-AW pivot (blue), and TS-IPW (green) (from left to right) under RTI Thompson sampling with normal reward errors.}}
\label{fig:MABlim_powerDconvergence_ThompsonInv}
\end{figure}

Figure~\ref{fig:MABlim_powerDconvergence_ThompsonInv} compares finite-sample and asymptotic powers under RTI Thompson sampling. For TS-t and TS-IPW, finite-sample power is already close to its limit at $\n = 100$. Convergence is slower for TS-AW and the TS-AW pivot: the gaps narrow as $\n$ increases but remain noticeable at larger effects even for $\n = 500$.

\subsubsection{\texorpdfstring{$K = 4$}{K = 4}}\label{subsubK4}

\begin{figure}[!htb]
\centering

\includegraphics[width=\linewidth]{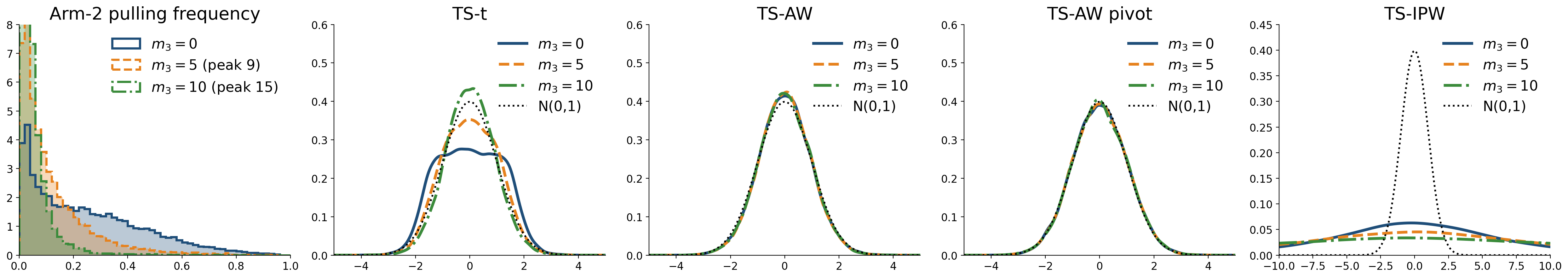}
\caption{{\small Histograms of Arm-2 pulling frequency and density estimates for TS-t, TS-AW, TS-AW pivot, and TS-IPW (from left to right) in the MAB limit experiment with $K = 4$ under RTI tempered-UCB. Parameter settings: $\m_1 = \m_2 = \m_4 = 0$, and $\m_3 = 0$ (solid blue), $5$ (dashed orange), and $10$ (dash-dotted green). Dotted black curves show the standard-normal density.}}
\label{fig:MABlim_statD_K4_temperedUCB}
\end{figure}

With $K = 4$, the effects of the other arms introduce additional nuisance parameters when comparing Arms 1 and 2. Figure~\ref{fig:MABlim_statD_K4_temperedUCB} varies $\m_3$ while holding $\m_1 = \m_2 = \m_4 = 0$ under RTI tempered-UCB. As $\m_3$ increases, the other arms are pulled less frequently and the TS-t and TS-IPW null distributions change substantially. Thus, translation invariance removes dependence on the common reward level but does not eliminate dependence on the other arms\textquotesingle{} effects. TS-AW and the TS-AW pivot remain close to $\mathcal{N}(0,1)$ across the three values of $\m_3$, with exact $\mathcal{N}(0,1)$ guaranteed for the pivot regardless of these nuisance parameters.

\begin{figure}[ht]
\centering
\centering
\includegraphics[width=0.415\linewidth]{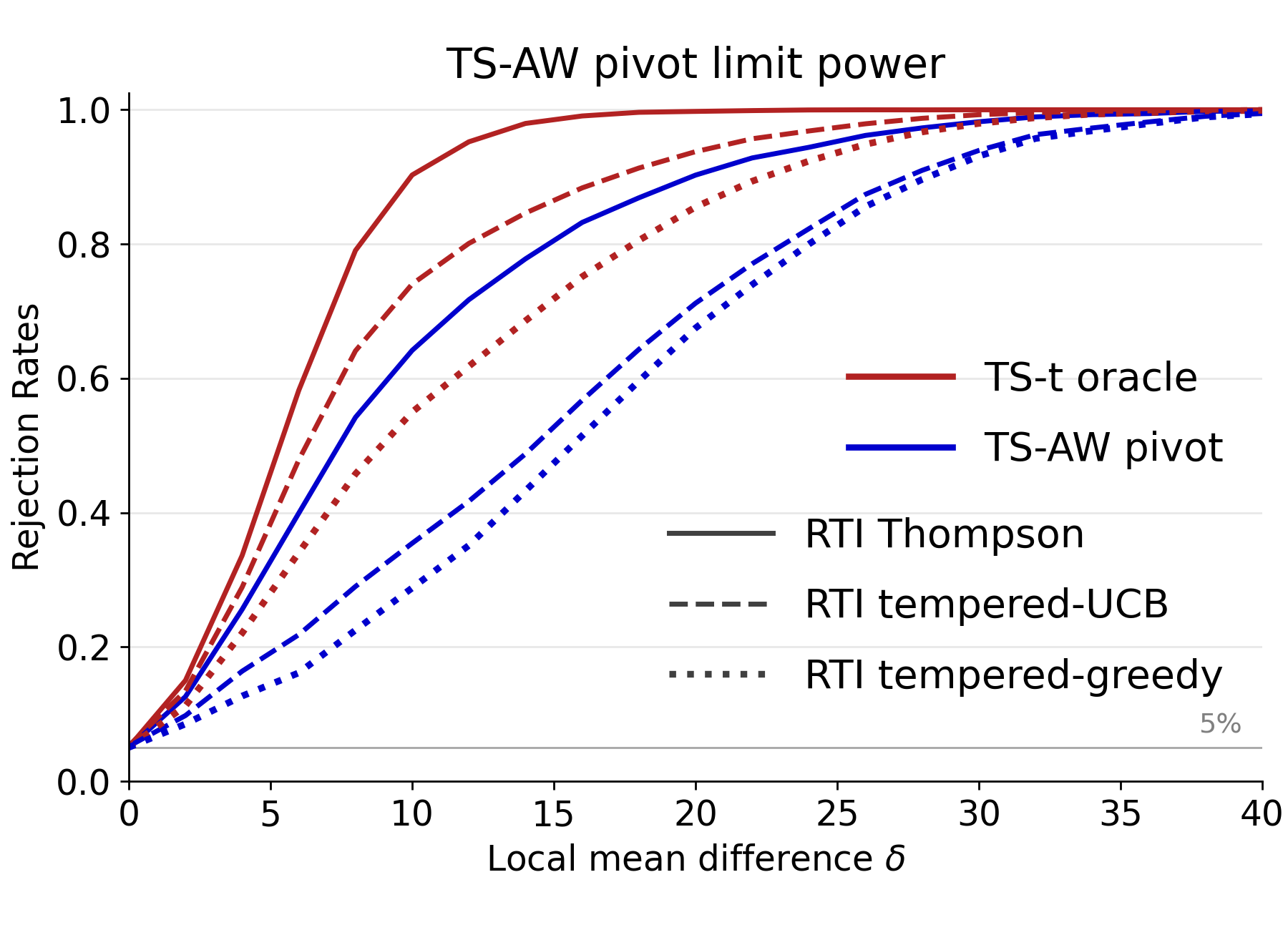}
\includegraphics[width=0.415\linewidth]{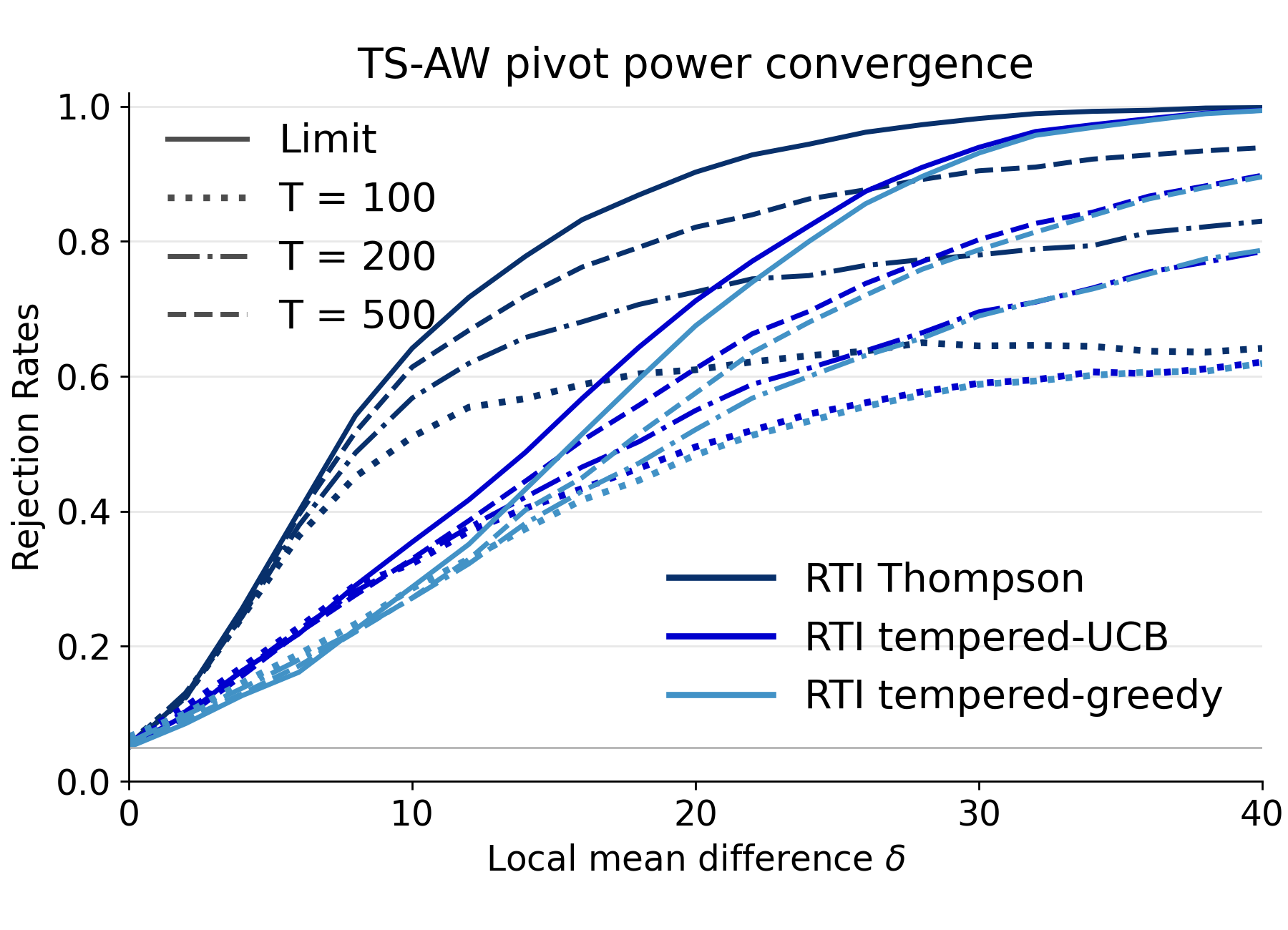}
\caption{{\small MAB asymptotic powers of TS-t oracle (red) and TS-AW pivot (blue) (left), and pivot power convergence with normal reward errors (right). RTI Thompson, tempered-UCB, and tempered-greedy use solid, dashed, and dotted lines (left) and dark blue, blue, and light blue (right), respectively. Right-panel line styles: asymptotic (solid); $\n = 100$ (dotted), $200$ (dash-dotted), $500$ (dashed).}}
\label{fig:MAB_TSAWpivot_power_K4}
\end{figure}

Figure~\ref{fig:MAB_TSAWpivot_power_K4} compares the TS-AW pivot with a TS-t oracle whose critical values are calibrated at the specified nuisance parameters $\m_3 = \m_4 = 0$.
The oracle has higher asymptotic power under each sampling scheme, illustrating the power cost of using the pivot in this comparison. RTI Thompson sampling yields the highest power, followed by RTI tempered-UCB and RTI tempered-greedy, with the latter two giving similar powers for the pivot. The right panel shows that the pivot\textquotesingle{}s finite-sample power approaches its limit as $\n$ increases, although appreciable gaps remain at moderate to large effects even for $\n = 500$.

\section{Statistical applications: contextual bandits}\label{sec:CMAB}

In this section, we consider a contextual $K$-armed bandit problem,
where the potential reward of Arm-$k$ at round $\ii$ follows the linear model
\begin{equation}
\Zkt = \vX_{\ii}^\trans\vbeta_k + \ekt, \qquad k\in\SK,
\end{equation}
where the contexts and independence conditions are as in Section~\ref{sec:sampling}, but we now specialize to standard Gaussian innovations, $\ekt \sim \mathcal{N}(0,1)$.
Thus, $\FJ_k=1$, and Assumption~\ref{assm:DQM} holds with $\pintb = (\vbeta_1^\trans,\dots,\vbeta_K^\trans)^\trans$. The Fisher information matrix $\FJbpintk(\vx)$ is again singular: its only nonzero block is $\vx\vx^\trans$ in block $(k,k)$, and it does not depend on $\pintb$.

Recall the equal-arms localization of Section~\ref{sec:sampling}. For this contextual bandit problem, the parameters $\vbeta_{k}$ are localized around a common parameter vector $\vbeta \in \rR^q$ as
\begin{equation} \label{eqn:equalarms_asymptotics_CMAB_decomposition}
\vbeta_{k,\n} = \vbeta + \frac{\vb_k}{\sqrt{\n}}, \qquad \vb_k = \vb + \vzeta_k,
\end{equation}
where $\vb$ determines the common reward level and $\vzeta_k$ the arm-specific deviation. Taking Arm~1 as the control (i.e., $\vzeta_1 = \mfzero$), we focus on testing the linear hypothesis
\begin{equation} \label{eqn:hypothesis_CMAB_twoarms}
H_0\colon \vG^\trans \vzeta_2 = 0 \quad\text{against}\quad H_1\colon \vG^\trans \vzeta_2 > 0,
\end{equation}
for a given vector $\vG \in \rR^q$, while treating $\vb$ and $\vzeta_3, \dots, \vzeta_K$ as nuisance parameters. This hypothesis is commonly used to compare the predictions $\vG^\trans \vb_1$ and $\vG^\trans \vb_2$, for example, in treatment-versus-control comparisons for an individual with covariates $\vG$.

As in Section~\ref{sec:MAB}, we begin by applying our general results from Section~\ref{sec:limitexperiment} to the contextual setting (Section~\ref{subsec:CMABLimit}), and then study the testing problem in \eqref{eqn:hypothesis_CMAB_twoarms} theoretically (Section~\ref{subsec:CMAB_comparetwoarms}) and support the findings with simulations (Section~\ref{subsec:montecarlo_CMAB_twoarms}).

\subsection{CMAB limit experiment} \label{subsec:CMABLimit}
Let $\lawn_{\vb_1,\dots,\vb_K}$ denote the law of $(\vX_1,\dots,\vX_{\n},\A_1,\dots,\A_{\n},\Y_1,\dots,\Y_{\n})$ under the local alternatives $\vbeta_{k,\n}=\vbeta+\vb_k/\sqrt{\n}$, $k\in\SK$. Specializing \eqref{eqn:statistic_general_CS} to the linear model, define
\begin{align*}
\vC_{r_1;k,\ii}
&\defeq
\frac{1}{\sqrt{\n}}\sum_{s=1}^{\ii}
\policy_{k,s}^{-r_1}\indicator_{\{A_s=k\}}
\vX_s(\Y_s-\vX_s^\trans\vbeta), \\
\vS_{r_2;k,\ii}
&\defeq
\frac{1}{\n}\sum_{s=1}^{\ii}
\policy_{k,s}^{-r_2}\indicator_{\{A_s=k\}}
\vX_s\vX_s^\trans,
\end{align*}
for $\ii=1,\dots,\n$ and $k\in\SK$, where
$\policy_{k,s}$ denotes the conditional probability of selecting
Arm-$k$ at round $s$, given the past observations and $\vX_s$.

Write $\vC_{k,\ii}\defeq\vC_{0;k,\ii}$ and $\vS_{k,\ii}\defeq\vS_{0;k,\ii}$ for the unweighted statistics. Under standard Gaussian innovations, the likelihood statistics
satisfy $\CSb_{k,\ii}=\vC_{k,\ii}$ and $\QVb_{k,\ii}=\vS_{k,\ii}$, so the quadratic likelihood expansion in Proposition~\ref{prop:LAQ} is exact:
\[
\log\frac{\dd\lawn_{\vb_1,\dots,\vb_K}}{\dd\lawn_{\mfzero,\dots,\mfzero}}
=
\sum_{k=1}^{K}\left(\vb_k^\trans\vC_{k,\n}-\frac12\vb_k^\trans\vS_{k,\n}\vb_k\right).
\]

In the linear CMAB problem, many feasible sampling policies depend on $\vS_{k,t}$ and the raw reward sums $\vC^\circ_{k,t}\defeq T^{-1/2}\sum_{s\le t}\indicator_{\{A_s=k\}}\vX_s\Y_s$, often through the (local) OLS reward predictions $\vX_{t+1}^\trans\hat{\vb}_{k,t} \defeq \vX_{t+1}^\trans\vS_{k,\ii}^{-1}\vC_{k,\ii}$. Since $\vC^\circ_{k,t}=\vC_{k,t}+\sqrt T\vS_{k,t}\vbeta$, their representation in terms of the centered statistics generally depends on the unknown common parameter $\vbeta$. We therefore use the RTI policies of Section~\ref{subsec:translation_invariant_policies}, whose translation invariance removes this dependence (Corollary~\ref{cor_TI_I}(a)), and write $\pi_{\ii+1}(k\cond\vX_{\ii+1},\filtr_{\ii})=\policy_k^{(\n)}(\vS_{\ii},\vC_{\ii},\vX_{\ii+1})$, where $\vC_{\ii}=(\vC_{1,\ii}^\trans,\dots,\vC_{K,\ii}^\trans)^\trans$ and $\vS_{\ii}=(\vS_{1,\ii}^\trans,\dots,\vS_{K,\ii}^\trans)^\trans$.

Under Assumptions~\ref{assm:policy_feasible}--\ref{assm:sde_solution}, Proposition~\ref{prop:likelihoodweakconvergence} yields joint convergence of the likelihood ratios and the weighted statistics used for inference. Their structural limit follows from Theorem~\ref{thm:structurallimitexperiment_f}. For notational convenience, define
\begin{align*}
\Epolicyb_k(\uu)
&\defeq
\Exp\big[
\policy_k(\vS(\uu),\vC(\uu),\vX)\vX\vX^\trans
\big], \\
\Epolicyb^{\vc}_{r_1;k}(\uu)
&\defeq
\Exp\big[
\policy_k^{1-2r_1}(\vS(\uu),\vC(\uu),\vX)
\vX\vX^\trans
\big], \\
\Epolicyb^{\vx}_{r_1;k}(\uu)
&\defeq
\Exp\big[
\policy_k^{1-r_1}(\vS(\uu),\vC(\uu),\vX)
\vX\vX^\trans
\big], \\
\Epolicyb^{\vs}_{r_2;k}(\uu)
&\defeq
\Exp\big[
\policy_k^{1-r_2}(\vS(\uu),\vC(\uu),\vX)
\vX\vX^\trans
\big],
\end{align*}
where the expectations integrate over the context distribution $\measure_X$. Here, $\vC$ and $\vS$ stack the unweighted processes across arms. Under Gaussianity, $\CSb_k=\vC_k$ and $\QVb_k=\vS_k$.

\begin{ancillary}
\label{ancillary:structurallimitexperiment_CMAB}
The limit experiment $\experiment$ is described as follows.
Under $\law_{\vb_1,\dots,\vb_K}$, we observe
$\vC_k$, $\vS_k$, $\vC_{r_1;k}$, and $\vS_{r_2;k}$,
$k\in\SK$, generated by
\begin{equation}
\label{eqn:SDE_CMAB}
\begin{aligned}
\dd
\begin{pmatrix}
\vC_k(\uu) \\
\vC_{r_1;k}(\uu)
\end{pmatrix}
&=
\begin{pmatrix}
\Epolicyb_k(\uu) \\
\Epolicyb^{\vx}_{r_1;k}(\uu)
\end{pmatrix}
\vb_k\,\dd\uu
+
\begin{pmatrix}
\Epolicyb_k(\uu)
&
\Epolicyb^{\vx}_{r_1;k}(\uu)^\trans
\\
\Epolicyb^{\vx}_{r_1;k}(\uu)
&
\Epolicyb^{\vc}_{r_1;k}(\uu)
\end{pmatrix}^{1/2}
\dd\vB_k(\uu), \\
\dd\vS_k(\uu)
&= \Epolicyb_k(\uu)\,\dd\uu, \\
\dd\vS_{r_2;k}(\uu)
&= \Epolicyb^{\vs}_{r_2;k}(\uu)\,\dd\uu,
\end{aligned}
\end{equation}
for $\uu\in[0,1]$, with all observed processes starting at zero. Here, the $\vB_k$ are independent standard $2q$-dimensional Brownian motions.
\end{ancillary}

\subsection{Tests for comparing the two arms} \label{subsec:CMAB_comparetwoarms}
We specialize to $K=2$ in this subsection. As illustrated in Section~\ref{subsec:MAB_comparetwoarms}, additional arms introduce nuisance effects $\vzeta_3,\dots,\vzeta_K$ that can affect the null distributions of the test statistics even under translation-invariant sampling. For the hypothesis of interest in
\eqref{eqn:hypothesis_CMAB_twoarms}, namely, $H_0:\vG^\trans\vzeta_2=0$ against $H_1:\vG^\trans\vzeta_2>0$, with the remaining parameters treated as nuisance, we study the two-sample Wald statistic and an adaptively weighted pivot. As in Section~\ref{sec:MAB}, we consider the regularized, translation-invariant sampling schemes of Theorem~\ref{thm:well-posedness}. We impose Assumption~\ref{assm:policy_limit}(b) throughout this subsection, which can be ensured by clipping as described in Remark~\ref{remark:clipping}.


We omit the original AW-Wald and IPW-Wald statistics and compare the standard Wald test with the AW pivot, whose null distribution in the limit experiment is standard normal.

\subsubsection*{Two-sample Wald test}
The classical two-sample Wald statistic for the linear model is given by
\begin{align*}
    \stat^{\text{TS-Wald}}_\n(\vG) \defeq \left(\vG^\trans(\vS_{2,\n}^{-1} + \vS_{1,\n}^{-1})\vG\right)^{-1/2}\left(\vG^\trans\hat\vzeta_{\n}\right),
\end{align*}
where $\hat\vzeta_{2,\ii} \defeq \vS_{2,\ii}^{-1}\vC_{2,\ii} - \vS_{1,\ii}^{-1}\vC_{1,\ii}$ estimates $\vzeta_2$ based on data available up to round $\ii = 1,\dots,\n$. By Proposition~\ref{prop:likelihoodweakconvergence} and Le Cam's third lemma \eqref{eqn:statistic_conv_alt}, under $\lawn_{\vb_1,\vb_2}$, we have $\hat\vzeta_{2,\n}\wto\tilde\vzeta_2(1) \defeq\vS_2(1)^{-1}\vC_2(1)-\vS_1(1)^{-1}\vC_1(1)$ and
\begin{align*}
\stat^{\text{TS-Wald}}_{\n}(\vG)
\wto
\stat^{\text{TS-Wald}}(\vG)
\defeq
\left(\vG^\trans
\big(\vS_2(1)^{-1}+\vS_1(1)^{-1}\big)\vG\right)^{-1/2}
\vG^\trans\tilde\vzeta_2(1),
\end{align*}
where the limiting processes have law $\law_{\vb_1,\vb_2}$. Under the RTI sampling schemes of
Section~\ref{subsec:translation_invariant_policies}, the null distribution does not depend on $\vb$ (see the discussion below Corollary~\ref{cor_TI_I}).

The structural representation \eqref{eqn:SDE_CMAB} gives, under $\law_{\vb_1,\vb_2}$,
\begin{align*}
\tilde\vzeta_2(1)
=
\vzeta_2
+\vS_2(1)^{-1}\int_0^1
\sqrt{\Epolicyb_2(\uu)}\,\dd\vB_2^{(0)}(\uu)
-\vS_1(1)^{-1}\int_0^1
\sqrt{\Epolicyb_1(\uu)}\,\dd\vB_1^{(0)}(\uu),
\end{align*}
where $\vB_1^{(0)}$ and $\vB_2^{(0)}$ are independent standard $q$-dimensional Brownian motions. The limiting Wald statistic is generally not normal due to the nonzero correlation between the integrand process $\sqrt{\Epolicyb_k(u)}$ and the integrator Brownian motion $\vB_k^{(0)}(u)$.

\subsubsection*{Two-sample Adaptively-Weighted pivot}
We propose a contextual extension of the two-sample AW pivot in \eqref{eqn:TSAWpivot}. Let $\vSigma_X\defeq\Exp[\vX\vX^\trans]$ and fix $\vG\ne\mfzero$. Taking $r_1=1$ in \eqref{eqn:SDE_CMAB}, the IPW process $\vC_{1;k}$ has drift $\vSigma_X\vb_k\,\dd\uu$. Hence, the scalar contrast $\vG^\trans\vSigma_X^{-1}
(\dd\vC_{1;2}-\dd\vC_{1;1})$ has drift $\vG^\trans\vzeta_2\,\dd\uu$ and instantaneous variance $v_{\vG}(\uu)\,\dd\uu$, where
\begin{align*}
v_{\vG}(\uu) \defeq \Exp\left[ (\vG^\trans\vSigma_X^{-1}\vX)^2 \left\{ \policy_1^{-1}(\vS(\uu),\vC(\uu),\vX) + \policy_2^{-1}(\vS(\uu),\vC(\uu),\vX) \right\} \right].
\end{align*}
With $h_{\vG}(\uu)\defeq v_{\vG}(\uu)^{-1/2}$, define
\begin{align} \label{eqn:TSAWpivot_CMAB}
\stat^{\text{TS-AW-P}}(\vG)
\defeq
\int_0^1 h_{\vG}(\uu)\vG^\trans\vSigma_X^{-1}
\big\{\dd\vC_{1;2}(\uu)-\dd\vC_{1;1}(\uu)\big\}
=
\vG^\trans\vzeta_2\int_0^1 h_{\vG}(\uu)\,\dd\uu
+\B_{\vG}(1),
\end{align}
where $\B_{\vG}$ is a standard Brownian motion under $\law_{\vb_1,\vb_2}$.
Indeed, the martingale part has quadratic variation $\int_0^\uu h_{\vG}(s)^2v_{\vG}(s)\,\dd s=\uu$, so the claim follows from L\'evy's characterization. Consequently, $\stat^{\text{TS-AW-P}}(\vG)\sim\mathcal{N}(0,1)$ exactly under $H_0$. We provide its finite-sample version in~(\ref{S-eqn:finite_CMAB_pivot}) of Appendix~\ref{S-appsec:finite_sample_pivots}.

For $\vG=\Exp[\vX]$, $\vG^\trans\vzeta_2$ is the local average treatment effect of Arm~2 relative to the control, a difference of two policy values of the kind targeted by CADR \citep{bibaut2021post}. CADR estimates such contrasts without a linear outcome model, using doubly robust scores with predictable variance stabilization; \eqref{eqn:TSAWpivot_CMAB} instead exploits the linear model through the weights $\vG^\trans\vSigma_X^{-1}\vX$.
Analyzing CADR in our framework would require the limit of its estimated outcome model, which we leave for future work.

\subsection{Simulation results} \label{subsec:montecarlo_CMAB_twoarms}
We assess the null distributions and power of the two-sample Wald test and the proposed AW pivot in the linear CMAB model. Our primary focus is on asymptotic results simulated via the limit experiment, for which we simulate SDEs (\ref{eqn:SDE_CMAB}) using an Euler scheme with $200$ time steps. As in the non-contextual case, we maintain a significance level of $\alpha = 5\%$ throughout the analysis. All results are based on $20{,}000$ replications.

We consider a linear model, $\Zkt = \beta_{k}^{\text{intercept}} + \beta_{k}^{\text{slope}}\X_\ii + \ekt$, for $k = 1,2$, where $\X_\ii\sim\mathcal{N}(0,1)$ and $\ekt\sim\mathcal{N}(0,1)$ are mutually independent across rounds and arms Thus, $\vX_\ii=(1,\X_\ii)^\trans$ and $\vSigma_X=\idm_2$. We test the intercept and slope parameters using
$\vG=(1,0)^\trans$ and $\vG=(0,1)^\trans$, respectively. For these cases, the two-sample Wald statistic takes the familiar t-ratio form. We therefore refer to it as TS-t and to the proposed pivot as TS-AW pivot.

We use the RTI Thompson, tempered-greedy, and tempered-LinUCB sampling schemes introduced in Section~\ref{subsec:translation_invariant_policies}. The Thompson parameter is $b=1/20$, the tempered-greedy
parameter is $\alpha=1$, and the tempered-LinUCB parameters are $\alpha=\lambda=1$. We clip the sampling probabilities at $\policy_{\min}=0.01$ using \eqref{eqn:uniform_mixing}. For the null-distribution comparison, we also include classical Thompson sampling with the same $b$ and probability floor. In the finite-sample experiments, each arm is pulled twice (since $q = 2$) during the first four rounds, ensuring that the arm-specific design matrices are invertible almost surely.

\subsubsection*{Null distributions}
\begin{figure}[!htb]
\centering
\centering

\includegraphics[width=\linewidth]{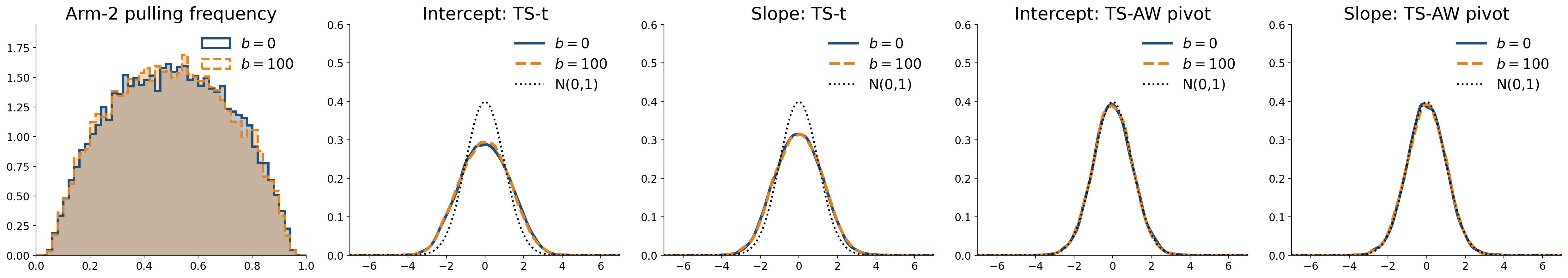}

\includegraphics[width=\linewidth]{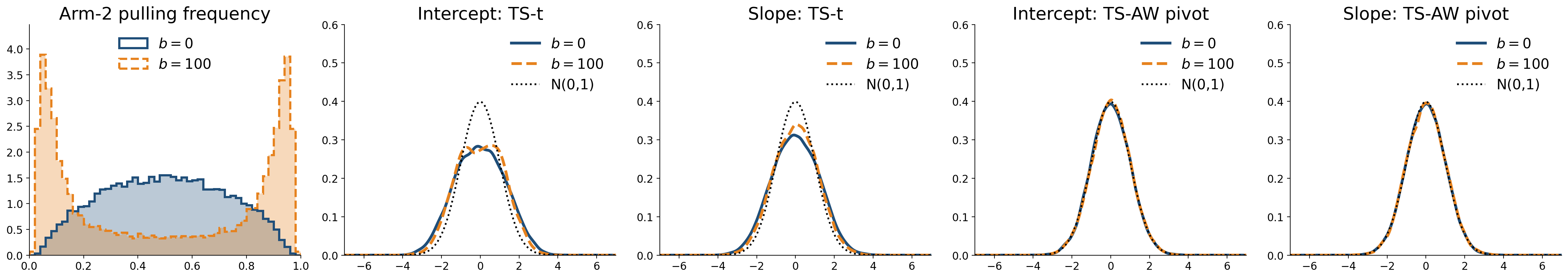}
\caption{{\small Histograms of Arm-2 pulling frequency and density estimates for TS-t (intercept and slope) and TS-AW pivot (intercept and slope), from left to right, in the CMAB limit experiment with $K=2$ under \emph{RTI} (top) and \emph{classical} (bottom) Thompson sampling. Parameter settings: $\vzeta_2=(0,0)^\trans$, with $\vb=(0,0)^\trans$ (solid blue) and $\vb=(100,100)^\trans$ (dashed orange).}}
\label{fig:CMABseq_statz_Thompson_ThompsonInv}
\end{figure}

Figure~\ref{fig:CMABseq_statz_Thompson_ThompsonInv} presents histograms of the Arm-2 pulling frequencies and the density estimates of TS-t and TS-AW pivot for the intercept and slope contrasts under $\vzeta_2=(0,0)^\trans$.
The top row uses RTI Thompson sampling and the bottom row uses classical Thompson sampling, with $\vb=(0,0)^\trans$ and $\vb=(100,100)^\trans$. Similar to the MAB case, under RTI sampling, the distributions remain essentially unchanged as $\vb$ varies, as they are designed to be distribution-free w.r.t $\vb$ under the null. The TS-t distributions nevertheless differ visibly from the standard normal reference, whereas the TS-AW pivot distributions closely match it. Under classical Thompson sampling, increasing the common local parameter $\vb$ shifts the allocation distribution toward the extremes dramatically and, therefore, changes the TS-t distributions. This again highlights the necessity of translation-invariant sampling schemes to ensure valid inference. An exception is the TS-AW pivot, which remains standard normal, regardless of the sampling scheme.

\subsubsection{Power results}
\begin{figure}[!htb]
\centering
\centering
\includegraphics[width=0.831\linewidth]{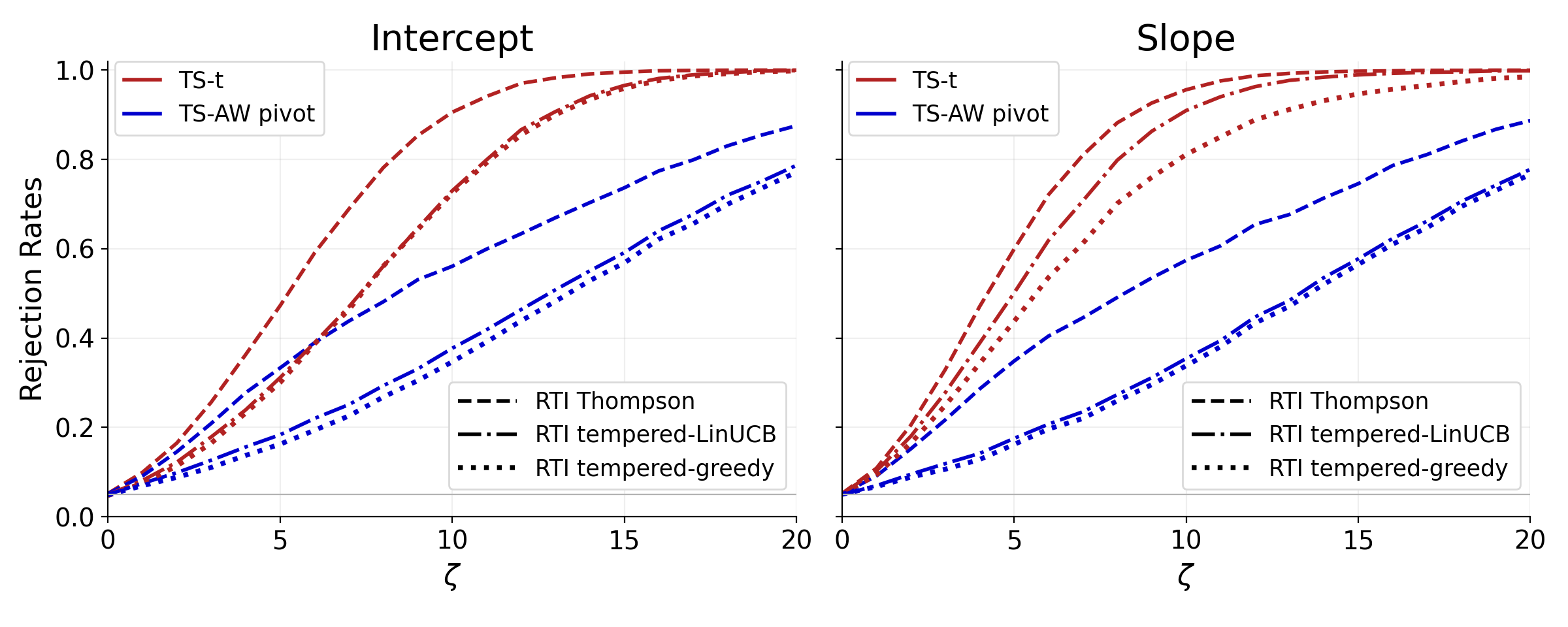}
\caption{{\small CMAB asymptotic powers of TS-t (red) and TS-AW pivot (blue) for the intercept (left) and slope (right) contrasts under RTI Thompson (dashed), tempered-LinUCB (dash-dotted), and tempered-greedy (dotted) sampling. Parameter settings: $\vb=(0,0)^\trans$ and $\vzeta_2=(\zeta,\zeta)^\trans$.}}
\label{fig:CMABlim_powerZalgorithms}
\end{figure}

\begin{figure}[ht]
\centering
\centering
\includegraphics[width=0.831\linewidth]{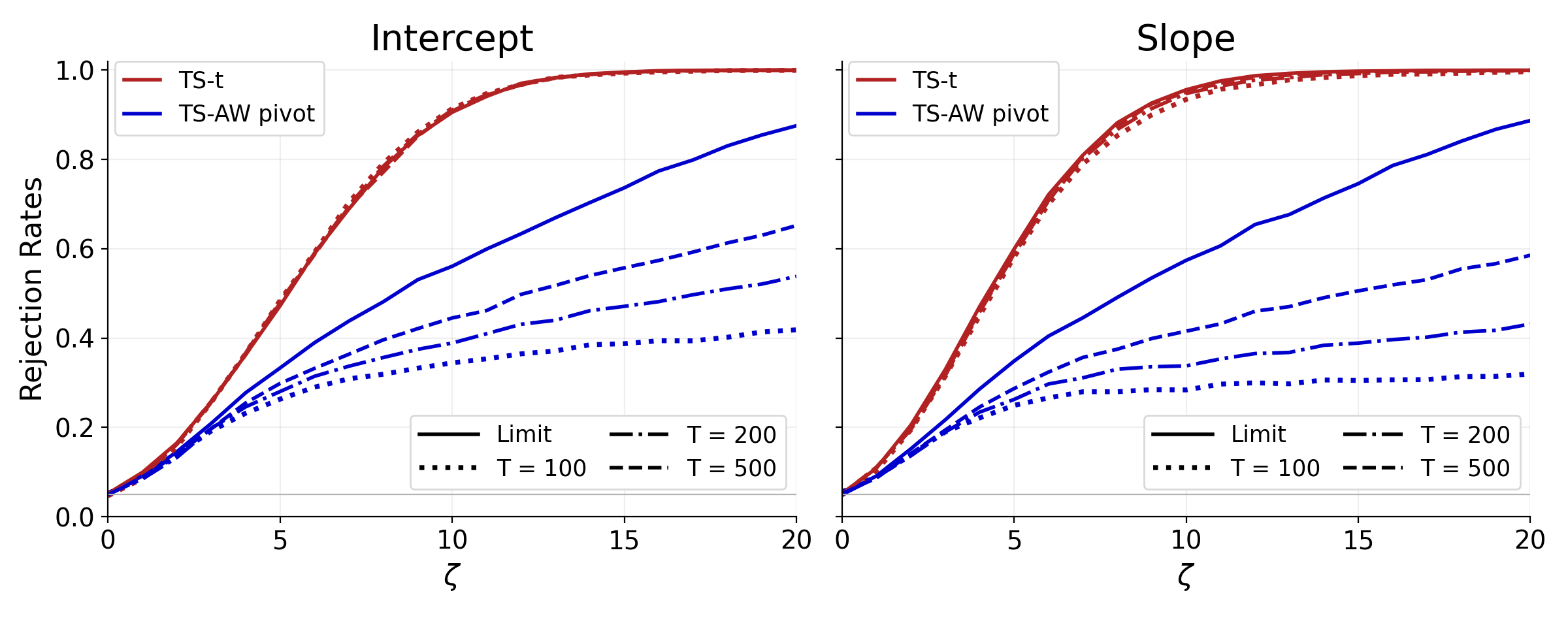}
\caption{{\small CMAB asymptotic (solid) and finite-sample powers for $\n=100$ (dotted), $200$ (dash-dotted), and $500$ (dashed) of TS-t (red) and TS-AW pivot (blue) for the intercept (left) and slope (right) contrasts under RTI Thompson sampling. Parameter settings: $\vb=(0,0)^\trans$ and $\vzeta_2=(\zeta,\zeta)^\trans$.}}
\label{fig:CMABlim_powerZconvergence}
\end{figure}

Turning to the power results, Figure~\ref{fig:CMABlim_powerZalgorithms} compares the asymptotic powers of TS-t and the TS-AW pivot for the intercept (left) and slope (right) contrasts, with $\vb=(0,0)^\trans$, under the three RTI sampling schemes. As in the non-contextual case, TS-t generally has higher asymptotic power than the pivot. Power also depends on the sampling scheme: RTI Thompson sampling, which favors exploration, yields the highest power in these comparisons. Figure~\ref{fig:CMABlim_powerZconvergence} compares asymptotic and finite-sample powers for $\n=100$, $200$, and $500$ under RTI Thompson sampling with the same parameter settings. TS-t has higher power and is already close to its asymptotic limit at $\n=100$. The TS-AW pivot converges more slowly, with substantial gaps at larger effects even at $\n=500$.
These comparisons illustrate the power cost, in this setting, of a standard-normal null limit that is invariant to the sampling policy and nuisance parameters.


\section{Conclusion} \label{sec:conclusion}
We derived the H\'ajek--Le Cam limit experiment for contextual multi-armed bandits under equal-arms asymptotics, gave it a structural representation in SDEs, and established the joint convergence of a general class of statistics. This embeds the diffusion approximations used to describe non-contextual bandits at a fixed parameter value in a framework for inference, and extends them to contextual bandits, non-Gaussian rewards, and weighted statistics. We used the framework to analyze the validity and local power of tests conducted after running a (fixed horizon) bandit experiment.

Our results suggest the following.
(i) Experimenters who intend to
conduct inference should use a translation-invariant policy. The RTI
versions proposed here remove the common reward level as a nuisance
parameter at little cost in regret, and they are available for the
policies most used in practice.
(ii) With two arms and an RTI policy, the two-sample $t$-test is valid over the whole composite null, provided its critical values are taken from the (non-normal) limit experiment, which can be simulated because the policy is known. Among the statistics we considered, it has the highest power.
(iii) With more than two arms, or when the policy is not translation invariant, the AW pivot is exactly standard normal under the null whatever the policy and the nuisance parameters, provided sampling probabilities are bounded away from zero, which clipping guarantees.
(iv) Policies that explore more, such as RTI Thompson sampling, yield more powerful tests.
(v) All weighted statistics, and the simulation of critical values, require the sampling probabilities to be recorded, which is straightforward in
designed experiments.

Several extensions are worth discussing.
(i) The framework is parametric in the conditional reward distribution. The null limiting distributions of the statistics we study, however, depend on the reward distribution only through its variance, because both the policies and the statistics are functions of cumulative rewards and arm-pull frequencies, so their size properties do not rest on the parametric model. The simulations with Student-$t_3$ rewards illustrate this. The parametric structure matters for local power and efficiency: with non-Gaussian rewards, the score-based sufficient statistics of the limit experiment differ from those used by the policies and studied test statistics, which leaves room for more powerful tests.
(ii) We assume i.i.d. contexts with a non-singular second-moment matrix.
Nonstationary or serially dependent contexts could alter the limit experiment.
(iii) We study fixed-horizon inference. Confidence sequences (see, for example, \citet{howard2021time}, \citet{waudby2024time}) provide validity under optional stopping. Combining their time-uniform guarantees with the power comparisons that the limit experiment allows is an interesting open problem.
(iv) Translation invariance is sufficient, not necessary, for the two-sample statistics to be free of the common level, and the pivot shows that validity can alternatively be achieved through the choice of statistic. We view translation invariance as the natural design requirement when the common level is a nuisance parameter, not as a universal requirement for post-bandit inference.

Finally, the limit experiment can be used beyond the analysis of given tests. The Asymptotic Representation Theorem yields power envelopes against which the tests studied here can be bench marked, and the structural SDE representation expresses both regret and power as functionals of the limiting allocation process.
This makes it possible to design sampling schemes that trade off regret against power.

\clearpage
\appendix
\section*{Supplementary appendices}
\addcontentsline{toc}{section}{Supplementary appendices}
\begingroup
\renewcommand{\stat}{\tau}
\renewcommand{\ve}{\boldsymbol{\e}}
\renewcommand{\vm}{\boldsymbol{m}}
\newcommand{\vGamma}{\boldsymbol{\Gamma}}
\newcommand{\jj}{s}
\renewcommand{\Xks}{\X_{\jj}}
\newcommand{\kernel}{K}
\newcommand{\vM}{\boldsymbol{M}}
\renewcommand{\theHequation}{appendix.\Alph{section}.\arabic{equation}}
\captionsetup[table]{skip=10pt}
\renewcommand{\theequation}{\Alph{section}.\arabic{equation}}
\setcounter{equation}{0}

\section{Differentiability in quadratic mean: examples}\label{S-appsec:DQM_examples}

Assumption~\ref{assm:DQM} (DQM) is our key maintained assumption on each arm's reward distribution as a function of the parameter of interest $\theta$. It's a standard assumption in the theory of limiting experiments implying, in a pure i.i.d.\ setting, a Locally Asymptotically Normal (LAN) limit. As explained in the main text, we do not get LAN limits due to the adaptive sampling of observations.

A convenient set of sufficient condition for DQM is given by \citet[Lemma~7.6]{van2000asymptotic}. DQM essentially holds in case the square-root densities $\sqrt{f_k(z\cond\theta)}$ are continuously differentiable in $\theta$ for all $z$ and the Fisher information $\FJ_k(\theta)$ exists and is continuous as a function of $\theta$. In that case, the score is given by $\score_k(z)=\partial_\theta\log f_k(z\cond\theta)$. This lemma covers, among many more examples,
\begin{itemize}
\item \emph{Gaussian location} $\mathcal{N}(\theta,\sigma^2)$ with $\sigma$ known: $\score(z)=(z-\theta)/\sigma^2$, $\FJ=1/\sigma^2<\infty$.
\item \emph{Bernoulli}$(\theta)$ with $\theta\in(0,1)$: $\score(z)=(z-\theta)/[\theta(1-\theta)]$, $\FJ=1/[\theta(1-\theta)]<\infty$.
\item \emph{Poisson}$(\theta)$ with $\theta>0$: $\score(z)=z/\theta-1$, $\FJ=1/\theta<\infty$.
\end{itemize}
Families not satisfying the DQM condition notably are those where the support depends on $\theta$ (e.g.\ $\mathrm{Uniform}(0,\theta)$) or zero-inflated models with a point-mass probability that is non-differentiable in $\theta$.

\section{Proofs concerning limit experiment} \label{S-appsec:proofs}
\begin{proof}[Proof of Proposition~\ref{prop:LAQ}]
Recall that the log-likelihood ratio contribution for arm $k\in\SK$ is, using Assumption~\ref{assm:policy_feasible},
\begin{align*}
\LLRn_{\pintb,k}(\lpintb)
= \sum_{\ii=1}^{\n} \indicator_{\{A_{\ii} = k\}}\log\frac{\fzk(\Y_{\ii}\cond\vX_{\ii},\pintb_{\n})}{\fzk(\Y_{\ii }\cond\vX_{\ii},\pintb)}
= \sum_{\ii=1}^{\n} \indicator_{\{A_{\ii} = k\}}\log\frac{\fzk(\Zkt\cond\vX_{\ii},\pintb_{\n})}{\fzk(\Z_{k,\ii }\cond\vX_{\ii},\pintb)}.
\end{align*}
We follow \citet[Proposition 1]{hallin2015quadratic} to prove the likelihood expansion for $\LLRn_{\pintb,k}(\lpintb)$. To put notions in their language, we let $P_{\n} = \lawn_{\mfzero}$, $\tilde{P}_{\n} = \lawn_{\lpintb}$, define
\begin{align*}
\vS_{\n\ii} = \frac{1}{\sqrt{\n}}\indicator_{\{A_{\ii} = k\}}\scorebk(\Zkt \cond\vX_{\ii}),
\end{align*}
for $\ii = 1,\dots,\n$, and write the individual likelihood ratio of observation $\ii$ as
\begin{align*}
LR_{\n\ii} = 1 + \indicator_{\{A_{\ii} = k\}}\left(\frac{\fzk(\Zkt\cond\vX_{\ii},\pintb_{\n})}{\fzk(\Zkt\cond\vX_{\ii},\pintb)}-1\right).
\end{align*}
By the DQM condition in Assumption~\ref{assm:DQM}, we can decompose
\begin{align*}
\sqrt{LR_{\n\ii}} = 1 + \frac{1}{2}\lpintb^\trans\vS_{\n\ii} + \frac{1}{2}R_{\n\ii},
\end{align*}
where $R_{\n\ii} = \indicator_{\{A_{\ii} = k\}}r_k\big(\Zkt\cond\vX_{\ii},\lpintb/\sqrt{\n}\big)$.

We verify the four conditions in \citet[Proposition 1]{hallin2015quadratic} using, in their notation, the filtration defined by $\filtr_{\n,\ii-1} = \sigma(\filtr_{\ii-1}\cup\sigma(\X_{\ii}))$.

Their \textit{Condition (a)} is trivially met as $\lpintb$ is constant.

For their \textit{Condition (b)}, note
\begin{align*}
\Exp_{P_{\n}}\left[\vS_{\n\ii}\cond\A_{\ii},\X_{\ii},\filtr_{\ii-1}\right]
&= \frac{1}{\sqrt{\n}}\indicator_{\{\A_{\ii} = k\}}\,\Exp_{P_{\n}}\big[\scorebk(\Zkt\cond\vX_{\ii})\cond\A_{\ii},\X_{\ii},\filtr_{\ii-1}\big]  \\
&= \frac{1}{\sqrt{\n}}\indicator_{\{\A_{\ii} = k\}}\,\Exp_{P_{\n}}\big[\scorebk(\Zkt\cond\vX_{\ii})\cond\vX_{\ii}\big] = \mfzero,
\end{align*}
where the second equality is due to the fact that $\Zkt$ is independent of $\A_{\ii}$ conditionally on $\vX_{\ii}$ and $\filtr_{\ii-1}$, and is independent of $\filtr_{\ii-1}$. The third equality is due to the DQM condition (see discussion after Assumption~\ref{assm:DQM}). Thus, their Display (2) follows because
\begin{align*}
  \Exp_{P_{\n}}\left[\vS_{\n\ii}\cond \filtr_{\n,\ii-1}\right]
= \Exp_{P_{\n}}\left[\Exp_{P_{\n}}\left[\vS_{\n \ii}\cond\A_{\ii},\X_{\ii},\filtr_{\ii-1}\right]\cond \filtr_{\n,\ii-1}\right]
= \mfzero.
\end{align*}
Similarly, for their $J_{T}$ in Display (3), we have
\begin{eqnarray*}
\lefteqn{\sum_{\ii=1}^{\n} \Exp_{P_{\n}}\left[ \vS_{\n\ii}\vS_{\n\ii}^\trans \cond \filtr_{\n,\ii-1} \right]}\\
 &=& \frac{1}{\n}\sum_{\ii=1}^{\n}\pi_{\ii}(k\cond\vX_{\ii},\filtr_{\ii-1})
 \Exp_{P_{\n}}\left[ \scorebk(\Zkt\cond\vX_{\ii})\scorebk(\Zkt\cond\vX_{\ii})^\trans \cond \vX_{\ii},\filtr_{\ii-1} \right]\\
 &=& \frac{1}{\n}\sum_{\ii=1}^{\n} \pi_{\ii}(k\cond\vX_{\ii},\filtr_{\ii-1})\FJbpintk(\vX_{\ii}) = \Op(1),
\end{eqnarray*}
where the last equality holds due to the existence of $\Exp[\FJbpintk(\vX)]$ by Assumption~\ref{assm:DQM}.

Turn next to the conditional Lindeberg's condition. For any $\delta > 0$,
\begin{align*}
&~ \sum_{\ii=1}^{\n} \,\Exp_{P_{\n}}\left[\big|\lpintb^\trans\vS_{\n\ii}\big|^2\indicator_{\{|\lpintb^\trans\vS_{\n\ii}| > \delta\}} \,\big|\, \filtr_{\n,\ii-1}\right]  \\
=&~ \frac{1}{\n}\sum_{\ii=1}^{\n} \pi_{\ii}(k\cond\vX_{\ii},\filtr_{\ii-1}) \Exp_{P_{\n}}\left[\big|\lpintb^\trans\scorebk(\Zkt\cond\vX_{\ii})\big|^2\indicator_{\{|\lpintb^\trans\vS_{\n\ii}| > \delta\}} \,\big|\, \vX_{\ii}\right]  \\
\leq&~ \frac{1}{\n}\sum_{\ii=1}^{\n} \,\Exp_{P_{\n}}\left[\big|\lpintb^\trans\scorebk(\Zkt\cond\vX_{\ii})\big|^2\indicator_{\{|\lpintb^\trans\scorebk(\Zkt\cond\vX_{\ii})| > \sqrt{\n}\delta\}} \,\big|\, \vX_{\ii}\right],
\end{align*}
where the inequality is due to the facts that $\pi_{\ii}(k\cond\vX_{\ii},\filtr_{\ii-1})\leq 1$ for all $k$ and $\ii$. The unconditional $\Exp_{P_{\n}}$-expectation of this average is $\Exp_{P_{\n}}\left[\big|\lpintb^\trans\scorebk(\Zkt\cond\vX_{\ii})\big|^2\indicator_{\{|\lpintb^\trans\scorebk(\Zkt\cond\vX_{\ii})| > \sqrt{\n}\delta\}}\right]$ which tends to zero as $\n\to\infty$ by dominated convergence.

For their \textit{Condition (c)}, note
\begin{equation} \label{S-eqn:proof_condition(c)}
\begin{aligned}
 &~ \sum_{\ii=1}^{\n} \Exp_{P_{\n}}\left[ R_{\n\ii}^2 \cond \filtr_{\n,\ii-1} \right]  \\
=&~ \sum_{\ii=1}^{\n} \Exp_{P_{\n}}\left[\indicator_{\{A_{\ii} = k\}}\Exp_{P_{\n}}\left[ r_k^2\big(\Zkt\cond\vX_{\ii},\lpintb/\sqrt{\n}\big) \cond \A_{\ii},\vX_{\ii},\filtr_{\ii-1} \right] \cond \filtr_{\n,\ii-1} \right]  \\
=&~ \sum_{\ii=1}^{\n} \pi_{\ii}(k\cond\vX_{\ii},\filtr_{\ii-1}) \Exp_{P_{\n}}\left[r_k^2\big(\Zkt\cond\vX_{\ii},\lpintb/\sqrt{\n}\big)\cond\vX_{\ii}\right] \\
\leq&~ \sum_{\ii=1}^{\n}\Exp_{P_{\n}}\left[r_k^2\big(\Zkt\cond\vX_{\ii},\lpintb/\sqrt{\n}\big)\cond\vX_{\ii}\right].
\end{aligned}
\end{equation}
Therefore, from Assumption~\ref{assm:DQM}, we have
\begin{align*}
\Exp_{P_{\n}}\sum_{\ii=1}^{\n} \Exp_{P_{\n}}\left[ R_{\n\ii}^2\cond\vX_{\ii},\filtr_{\ii-1} \right]
\leq \n\Exp\left[r_k^2\big(\Z_{k}\cond\vX,\lpintb/\sqrt{\n}\big)\right]
= \n o(1/\n) = o(1).
\end{align*}

Their Display (5) is satisfied as
\begin{align*}
 &~ \sum_{\ii=1}^{\n} (1- \Exp_{P_{\n}}[LR_{\n\ii}\cond\vX_{\ii},\filtr_{\ii-1}]) \\
=&~ \sum_{\ii=1}^{\n} -\Exp_{P_{\n}}\left[\indicator_{\{A_{\ii} = k\}}\left(\frac{\fzk(\Zkt\cond\vX_{\ii},\pintb_{\n})}{\fzk(\Zkt\cond\vX_{\ii},\pintb)}-1\right) \,\Big|\, \vX_{\ii},\filtr_{\ii-1}\right]  \\
=&~ \sum_{\ii=1}^{\n} - \pi_{\ii}(k\cond\vX_{\ii},\filtr_{\ii-1})\Exp_{P_{\n}}\left[\frac{\fzk(\Zkt\cond\vX_{\ii},\pintb_{\n})}{\fzk(\Zkt\cond\vX_{\ii},\pintb)}-1 \,\Big|\, \vX_{\ii}\right] = 0,
\end{align*}
where the second equality follows the same arguments as (\ref{S-eqn:proof_condition(c)}). The last equality is automatic when the densities $\fzk$ are strictly positive.

Finally, their \textit{Condition (d)} is naturally true under our setting.
\end{proof}

\begin{lemma}\label{S-lem:subsystem}
Let $\vW$ be a standard $d$-dimensional Brownian motion and consider
\[
\dd\boldsymbol\xi^{(1)}=\vb_1(\boldsymbol\xi^{(1)})\dd\uu+\boldsymbol\sigma_1(\boldsymbol\xi^{(1)})\dd\vW,\qquad
\dd\boldsymbol\xi^{(2)}=\vb_2(\boldsymbol\xi^{(1)})\dd\uu+\boldsymbol\sigma_2(\boldsymbol\xi^{(1)})\dd\vW,
\]
with $\boldsymbol\xi(0)=\mfzero$, where $\vb_i,\boldsymbol\sigma_i$ are continuous and bounded. Suppose the martingale problem for
$(\vb_1,\boldsymbol\sigma_1\boldsymbol\sigma_1^\trans)$ is well posed (e.g., $\vb_1$ and some
$\tilde{\boldsymbol\sigma}_1$ with $\tilde{\boldsymbol\sigma}_1\tilde{\boldsymbol\sigma}_1^\trans=\boldsymbol\sigma_1\boldsymbol\sigma_1^\trans$
are locally Lipschitz). Then the system has a weak solution, the joint law of $(\boldsymbol\xi,\vW)$
is unique, and the martingale problem for $\boldsymbol\xi$ is well posed.
\end{lemma}
\begin{proof}
Existence: the coefficients are continuous and bounded (Skorokhod; \citealp[Ch.~IV.2]{ikeda2014stochastic}).
Uniqueness: in any weak solution, $\boldsymbol\xi^{(1)}$ solves the martingale problem for
$(\vb_1,\boldsymbol\sigma_1\boldsymbol\sigma_1^\trans)$, whose solution is unique in law. Hence
uniqueness in law holds for the $\boldsymbol\xi^{(1)}$-equation driven by $\vW$, and by
\citet[Thm.~3.1]{cherny2002uniqueness} the joint law of $(\boldsymbol\xi^{(1)},\vW)$ is unique.
As $\boldsymbol\sigma_2(\boldsymbol\xi^{(1)})$ is continuous and adapted,
$\boldsymbol\xi^{(2)}=\int\vb_2(\boldsymbol\xi^{(1)})\dd\uu+\int\boldsymbol\sigma_2(\boldsymbol\xi^{(1)})\dd\vW$
is the limit in probability of Riemann sums along dyadic partitions, hence a fixed measurable
functional of $(\boldsymbol\xi^{(1)},\vW)$, so the law of $(\boldsymbol\xi,\vW)$ is unique. Every
solution of the martingale problem for $\boldsymbol\xi$ is a weak solution of the stacked equation
on an extension \citep[Prop.~5.4.6]{karatzas2012brownian}. Under the example condition, pathwise
uniqueness \citep[Thm.~5.2.5]{karatzas2012brownian} and Yamada--Watanabe give uniqueness in law,
which is equivalent to well-posedness \citep[Prop.~5.4.11]{karatzas2012brownian}.
\end{proof}

\begin{proof}[Proof of Proposition~\ref{prop:weakconvergence_general}]

Our proof is based on \citet[Theorem 2.1]{NelsonFoster1994} which, in turn, is a modified version of  \citet[Theorem 11.2.3]{stroock1997multidimensional}. In particular, we take in their notation: $n = lK$, $h = 1/\n$, $T = 1$, and $\mathit{\Delta} = 1$. The discrete-time triangular array $\Ubparsum_{\ii} = ({\Ubparsum_{1,\ii}}^\trans,{\Ubparsum_{2,\ii}}^\trans)^\trans \defeq ({\parsumfb_{\ii}}^\trans,{\parsumqb_{\ii}}^\trans)^\trans$ is a time-homogeneous Markov process for each $\n$. Note that we allow for $\Ubparsum_{\ii}$ to depend on $\n$, but we omit it from the notation. Define the conditional mean and covariance matrix of $\Ubparsum_{\ii}$ in block form as
\begin{align*}
\vmu^{(\n)}(\vu) = \begin{pmatrix}
\vmu_1^{(\n)}(\vu) \\ \vmu_2^{(\n)}(\vu) \end{pmatrix}, ~~
\vOmega^{(\n)}(\vu) = \begin{pmatrix}
\vOmega^{(\n)}_{11}(\vu) & \vOmega^{(\n)}_{12}(\vu)  \\
\vOmega^{(\n)}_{21}(\vu) & \vOmega^{(\n)}_{22}(\vu)  \\
\end{pmatrix},
\end{align*}
where (recall $h^{-\Delta} = \n$)
\begin{equation}
\begin{aligned}
\vmu_{i}^{(\n)}(\vu) \defeq&~ \n\,\Exp\left[ \Ubparsum_{i,\ii+1} - \Ubparsum_{i,\ii}\cond\Ubparsum_{\ii} = \vu\right], \\
\vOmega_{i j}^{(\n)}(\vu) \defeq&~ \n\cov\left[\Ubparsum_{i,\ii+1} - \Ubparsum_{i,\ii}, \Ubparsum_{j,\ii+1} - \Ubparsum_{j,\ii}\cond\Ubparsum_{\ii} = \vu\right],
\end{aligned}
\end{equation}
for $i,j \in \{1,2\}$ and $\vu\in\rR^{lK}$. And we define their limiting versions as
\begin{equation*}
\vmu(\vu)
\defeq
\begin{pmatrix}
\mfzero_{l_1K}  \\
\big(\Exp\big[\policy_k(\vu,\vX)\funqb_k(\Z_{k},\vX,\vu)\big]\big)_{k=1}^{K}
\end{pmatrix},  \quad
\vOmega(\vu)
\defeq
\begin{pmatrix}
\vOmega_{11}(\vu) & \mfzero_{l_1K\times l_2K}  \\
\mfzero_{l_2K\times l_1K} & \mfzero_{l_2K\times l_2K}  \\
\end{pmatrix},
\end{equation*}
with, for $k,k'\in\SK$,
\begin{align*}
\vOmega_{11}(\vu)[k,k']  &=
\begin{cases}
\Exp\big[\policy_k(\vu,\vX){\funfb_{k}}(\Z_{k},\vX\cond\vu)\funfb_{k}(\Z_{k},\vX\cond\vu)^\trans\big], & \text{if~} k = k', \\
0 & \text{if~} k \neq k'.
\end{cases}
\end{align*}
Then, in what follows, we are going to show, for some $\delta > 0$ and as $\n\to\infty$,
\begin{align}
& \vmu^{(\n)}(\vu) \to \vmu(\vu), \label{S-eqn:condition_b} \\
& \vOmega^{(\n)}(\vu) \to \vOmega(\vu), \label{S-eqn:condition_c} \\
& \n\Exp\left[\left\|\Ubparsum_{\ii+1} - \Ubparsum_{\ii}\right\|^{2+\delta} \cond \Ubparsum_{\ii} = \vu\right] \to 0, \label{S-eqn:condition_d}
\end{align}
uniformly on bounded $\vu$ subsets, i.e., on every set of the form $\{\vu : \|\vu\| \leq c\}$ with $c>0$.
Here, (\ref{S-eqn:condition_b}), (\ref{S-eqn:condition_c}), and (\ref{S-eqn:condition_d}) correspond to Conditions~(b'), (c'), and (d') in \citet[Theorem 2.1]{NelsonFoster1994}, respectively. Their Condition~(a') is immediate.

\noindent
(I) We start with verifying (\ref{S-eqn:condition_b}) by showing the convergence of these expectation block-wise conditioning on $\Ubparsum_{\ii} = \vu$.
\begin{itemize}
\item[1)] For $k\in\SK$, we have
\begin{align*}
\vmu_1^{(\n)}(\vu)[k]
&= \n\,\Exp\left[ \parsumf_{k,\ii+1} - \parsumf_{k,\ii}\cond\Ubparsum_{\ii} = \vu\right] \\
&= \sqrt{\n}\,\Exp\left[ \indicator_{\{A_{\ii+1} = k\}}\funfb_k(\Y_{\ii+1},\vX_{\ii+1} \cond \Ubparsum_{\ii}) \cond \Ubparsum_{\ii} = \vu\right] \\
&= \sqrt{\n}\,\Exp\left[ \Exp\left[\indicator_{\{A_{\ii+1} = k\}}\funfb_k(\Y_{\ii+1},\vX_{\ii+1}\cond\Ubparsum_{\ii})\cond\A_{\ii+1} = k,\vX_{\ii+1},\Ubparsum_{\ii}\right]\cond\Ubparsum_{\ii} = \vu\right] \\
&= \sqrt{\n}\,\Exp\left[ \indicator_{\{A_{\ii+1} = k\}}
\Exp\left[\funfb_k(\Y_{\ii+1},\vX_{\ii+1}\cond\Ubparsum_{\ii})\cond\A_{\ii+1} = k,\vX_{\ii+1},\Ubparsum_{\ii}\right]\cond\Ubparsum_{\ii} = \vu\right] \\
&= \sqrt{\n}\,\Exp\left[ \indicator_{\{A_{\ii+1} = k\}}
\Exp\left[\funfb_k(\Z_k,\vX\cond\vu)\cond\vX\right]\right]
= 0,
\end{align*}
where the last equality holds because of~(\ref{eqn:functioncondition_moment_I}).
\item[2)] For $k\in\SK$, we have
\begin{align*}
\vmu_2^{(\n)}(\vu)[k]
=&~ \n\,\Exp\big[ \parsumq_{k,\ii+1}  - \parsumq_{k,\ii}\cond\Ubparsum_{\ii} = \vu\big] \\
=&~ \Exp\big[ \indicator_{\{A_{\ii+1} = k\}}\funqb_k(\Y_{\ii+1},\vX_{\ii+1},\Ubparsum_{\ii}) \cond \Ubparsum_{\ii} = \vu\big] \\
=&~ \Exp\big[\policy_k^{(\n)}(\vu,\vX)\funqb_k(\Z_{k},\vX,\vu)\big]
\to \Exp\big[\policy_k(\vu,\vX)\funqb_k(\Z_{k},\vX,\vu)\big],
\end{align*}
where the convergence is uniform over bounded $\vu$-sets by Assumption~\ref{assm:policy_limit}(a).
\end{itemize}

\noindent
(II) Turn next to the covariance in (\ref{S-eqn:condition_c}).
\begin{itemize}
\item[1)] For $k\in\SK$, we have
\begin{align*}
  \vOmega_{11}^{(\n)}(\vu)[k,k]
=&~ \Exp\left[\indicator_{\{A_{\ii+1} = k\}}\funfb_k(\Y_{\ii+1},\vX_{\ii+1}\cond\Ubparsum_{\ii})\funfb_k(\Y_{\ii+1},\vX_{\ii+1}\cond\Ubparsum_{\ii})^\trans\cond\Ubparsum_{\ii} = \vu\right]  \\
=&~ \Exp\left[\policy_k^{(\n)}(\vu,\vX)\funfb_{k}(\Z_{k},\vX\cond\vu)\funfb_{k}(\Z_{k},\vX\cond\vu)^\trans\right]  \\
\to&~ \Exp\left[\policy_k(\vu,\vX)\funfb_{k}(\Z_{k},\vX\cond\vu)\funfb_{k}(\Z_{k},\vX\cond\vu)^\trans\right],
\end{align*}
where the convergence follows again from Assumption~\ref{assm:policy_limit}(a). For $k,k'\in\SK$ with $k \neq k'$, we have $\vOmega_{11}^{(\n)}(\vu)[k,k'] = 0$ since one of $\indicator_{\{A_{\ii+1} = k\}}$ and $\indicator_{\{A_{\ii+1} = k'\}}$ has to be zero.
\item[2)] For all $k,k'\in\SK$, we have
\begin{equation}
\|\vOmega_{22}^{(\n)}(\vu)[k,k']\|
 \leq
 \frac{1}{\n}
 \sqrt{\Exp\|\funqb_k(\Z_k,\vX,\vu)\|^2\Exp\|\funqb_{k'}(\Z_{k'},\vX,\vu)\|^2}
 \to 0
\end{equation}
uniformly on bounded $\vu$-sets. Similarly, $\vOmega_{1 2}^{(\n)}(\vu) \to 0$ uniformly on bounded $\vu$-sets due to exploding rates ($\sqrt{\n}$) in the denominator. The same holds for $\vOmega_{2 1}^{(\n)}(\vu)$.
\end{itemize}

\noindent
(III) Finally, we show that the conditional moment condition (\ref{S-eqn:condition_d}) holds because
\begin{align*}
&~ \n\Exp\Big[\left\|\Ubparsum_{\ii+1} - \Ubparsum_{\ii}\right\|^{2+\delta} \cond \Ubparsum_{\ii} = \vu\Big]  \\
=&~ \n\Exp\Big[\Big\{\sum_{k}\big( \n^{-1}\indicator_{\{A_{\ii+1} = k\}}\|\funfb_k(\Y_{\ii+1},\vX_{\ii+1}\cond\Ubparsum_{\ii})\|^2 \\
&~~~~~~~~~~~~~~~~~~~~ + \n^{-2}\indicator_{\{A_{\ii+1} = k\}}\|\funqb_k(\Y_{\ii+1},\vX_{\ii+1},\Ubparsum_{\ii})\|^2\big)\Big\}^{1+\frac{\delta}{2}} \,\big|\, \Ubparsum_{\ii} = \vu\Big] \\
=&~ \n^{-\frac{\delta}{2}}\Exp\Big[\Big\{\sum_{k}\big( \indicator_{\{A_{\ii+1} = k\}}\|\funfb_k(\Y_{\ii+1},\vX_{\ii+1} \cond \Ubparsum_{\ii})\|^2 \\
&~~~~~~~~~~~~~~~~~~~~ + \n^{-1}\indicator_{\{A_{\ii+1} = k\}}\|\funqb_k(\Y_{\ii+1},\vX_{\ii+1},\Ubparsum_{\ii})\|^2\big)\Big\}^{1+\frac{\delta}{2}} \,\big|\, \Ubparsum_{\ii} = \vu\Big] \\
\leq&~ \n^{-\frac{\delta}{2}}\Exp\Big[\Big\{\sum_{k}\big( \|\funfb_k(\Z_{k,\ii+1},\vX_{\ii+1} \cond \Ubparsum_{\ii})\|^2 + \n^{-1}\|\funqb_k(\Z_{k,\ii+1},\vX_{\ii+1},\Ubparsum_{\ii})\|^2  \big)\Big\}^{1+\frac{\delta}{2}} \,\big|\, \Ubparsum_{\ii} = \vu \Big] \\
\leq&~ (2K)^{\frac{\delta}{2}}\n^{-\frac{\delta}{2}}
\sum_{k}\Exp\Big[ \|\funfb_k(\Z_{k,\ii+1},\vX_{\ii+1} \cond \Ubparsum_{\ii})\|^{2+\delta} \\
&~~~~~~~~~~~~~~~~~~~~ + \n^{-1-\frac{\delta}{2}}\|\funqb_k(\Z_{k,\ii+1},\vX_{\ii+1},\Ubparsum_{\ii})\|^{2+\delta} \,\big|\, \Ubparsum_{\ii} = \vu \Big] \to 0,
\end{align*}
where the convergence is uniform over bounded $\vu$-sets due to (\ref{eqn:functioncondition_moment_II}).
\end{proof}

\begin{proof}[Proof of Proposition~\ref{prop:likelihoodweakconvergence}]
The proof of Part (a) follows by an immediate application of Proposition~\ref{prop:weakconvergence_general}, where we stack $\CSb_{k,\ii}$ and $\vC_{r_1;k,\ii}$ in the first component, $\parsumfb_{k}$, and stack $\QVb_{k,\ii}$ and $\vS_{r_2;k,\ii}$ in the second component, $\parsumqb_{k}$. Condition~\eqref{eqn:functioncondition_moment_II} for the stacked functions requires care, as the $\vC_{r_1;k}$ and $\vS_{r_2;k}$ blocks carry the factors $\policy_k^{-r_1(2+\delta)}$ and $\policy_k^{-r_2(2+\delta)}$. For $r_1,r_2\le0$ these are bounded by one and Assumption~\ref{assm:policy_limit}(a) suffices. For $r_1>0$ or $r_2>0$, Assumption~\ref{assm:policy_limit}(b) bounds them by the constants $\policy_{\min}^{-r_1(2+\delta)}$ and $\policy_{\min}^{-r_2(2+\delta)}$, leaving $\Exp\big[\|\vX\|^{2+\delta}|\Z_k-\Exp[\Z_k\cond\vX]|^{2+\delta}\big]$ and $\Exp\|\vX\|^{2(2+\delta)}$, which are finite by Cauchy--Schwarz and Assumption~\ref{assm:DQM} once $\delta$ is taken small enough that $2(2+\delta)$ does not exceed its moment order. This, in turn, implies Part (b). Note that for the $\CSb_{k}$ terms, we have
\begin{align*}
\Exp\left[\policy_k(\vu,\vX)\scorebk(\Z_{k}\cond\vX)\scorebk(\Z_{k}\cond\vX)^\trans\right]
&= \Exp\left[\Exp\left[\policy_k(\vu,\vX)\scorebk(\Z_{k}\cond\vX)\scorebk(\Z_{k}\cond\vX)^\trans\cond\vX\right]\right] \\
&= \Exp\left[\policy_k(\vu,\vX)\Exp\left[\scorebk(\Z_{k}\cond\vX)\scorebk(\Z_{k}\cond\vX)^\trans\cond\vX\right]\right] \\
&= \Exp\left[\policy_k(\vu,\vX)\FJbpintk(\vX)\right] = \Epolicyb^{\vartriangle}_{k}(\vu).
\end{align*}
Part~(c): by Part~(a), $\langle\lpintb^\trans\CSb_k\rangle(\uu)=\int_0^\uu\lpintb^\trans
\Epolicyb^{\vartriangle}_k(\UblimE(s))\lpintb\,\dd s=\lpintb^\trans\QVb_k(\uu)\lpintb$, and
$\langle\CSb_k,\CSb_{k'}\rangle=0$ for $k\neq k'$ as $\vW_k,\vW_{k'}$ are independent.
$\sum_k\lpintb^\trans\QVb_k(1)\lpintb\le\sum_k\lpintb^\trans\Exp[\FJbpintk(\vX)]\lpintb\defeq c_{\lpintb}<\infty$ deterministically as $0\le\policy_k\le1$.
Hence Novikov's condition
$\Exp\exp\big(\tfrac12\sum_k\lpintb^\trans\QVb_k(1)\lpintb\big)\le e^{c_{\lpintb}/2}<\infty$ holds
\citep[Cor.~3.5.13]{karatzas2012brownian}, so both exponentials are martingales and
$\Exp_{\pintb,\mfzero}[\exp\LLRlim_{\pintb}(\lpintb)]=1$.
\end{proof}

Propositions \ref{prop:LAQ}--\ref{prop:likelihoodweakconvergence} lead to the limit experiment for contextual bandit problems as stated in Theorem~\ref{thm:structurallimitexperiment_f}. Its proof is now straightforward.
\begin{proof}[Proof of Theorem~\ref{thm:structurallimitexperiment_f}]
Under $\law_{\pintb,\mfzero}$, the processes $\vW_{k}$ are independent standard Brownian motions. By Proposition~\ref{prop:likelihoodweakconvergence}(a),
\begin{eqnarray*}
\lpintb^\trans\CSb_{k}(1) &=&
    \int_0^1\left(\begin{array}{c}\lpintb\\ \mfzero\end{array}\right)^\trans
    \left[\begin{array}{cc}
    \Epolicyb^{\vartriangle}_{k}(\UblimE(\s)) &
        \Epolicyb^{\vx}_{r_1;k}(\UblimE(\s))^\trans\\
    \Epolicyb^{\vx}_{r_1;k}(\UblimE(\s)) &
        \Epolicyb^{\vc}_{r_1;k}(\UblimE(\s))
    \end{array}\right]^{\frac12}
    \dd\vW_k(\s),\\
\QVb_{k}(1) &=&
    \int_0^1\Epolicyb^{\vartriangle}_{k}(\UblimE(\s))\,\dd\s.
\end{eqnarray*}
Observe that the quadratic variation of $\lpintb^\trans\sum_{k}\CSb_{k}(\cdot)$ is $\int_0^\uu
\lpintb^\trans\sum_k\Epolicyb^{\vartriangle}_k(\UblimE(v))\lpintb\,\dd v$ by the independence across arms of $\vW_k$.
The likelihood ratio $\exp\left(\sum_{k}\big(\lpintb^\trans\CSb_{k}(1)-\tfrac12\lpintb^\trans\QVb_{k}(1)\lpintb\big)\right)$
is the time-$1$ value of the Dol\'eans--Dade exponential of the continuous martingale $\lpintb^\trans\sum_{k}\CSb_{k}(\cdot)$.
This Dol\'eans--Dade exponential is itself a martingale by Proposition~\ref{prop:likelihoodweakconvergence}(c). Thus $\law_{\pintb,\lpintb}$ is a probability measure equivalent to $\law_{\pintb,\mfzero}$. Write
\[
\vM_k(\vu)\defeq
\begin{pmatrix}
\Epolicyb^{\vartriangle}_k(\vu) &
(\Epolicyb^{\vx}_{r_1;k}(\vu))^\trans\\
\Epolicyb^{\vx}_{r_1;k}(\vu) &
\Epolicyb^{\vc}_{r_1;k}(\vu)
\end{pmatrix}.
\]
In what follows, we suppress the common argument $\UblimE(\s)$
of $\vM_k$ and its blocks. Girsanov's theorem \citep[Ch.~3, Thm~5.1]{karatzas2012brownian} gives, under $\law_{\pintb,\lpintb}$, $\dd\vW_k=\vM_k^{1/2}\binom{\lpintb}{\mfzero}\dd\s+\dd\vB_k$, where the $\vB_k$ are independent standard Brownian motions. Hence,
\[
\dd\binom{\CSb_k}{\vC_{r_1;k}}
=
\vM_k^{1/2}\dd\vW_k
=
\vM_k\binom{\lpintb}{\mfzero}\dd\s+\vM_k^{1/2}\dd\vB_k
=
\binom{\Epolicyb^{\vartriangle}_k\lpintb}
      {\Epolicyb^{\vx}_{r_1;k}\lpintb}\dd\s
+\vM_k^{1/2}\dd\vB_k.
\]
Together with the equations for $\QVb_k$ and $\vS_{r_2;k}$, this yields the result.
\end{proof}

\section{Proofs concerning sampling schemes} \label{S-appsec:TI_algorithms}

\begin{proof}[Proof of Corollary~\ref{cor_TI_I}]
For Part~(a), observe, for $k\in\SK$,
\begin{align*}
\vC_{k,\ii} = \vC_{k,\ii}^\circ - \sqrt{\n}\vS_{k,\ii}\vbeta,
\end{align*}
where $\vC_{k,\ii}^\circ \defeq \frac{1}{\sqrt{\n}}\sum_{\jj=1}^{\ii}\indicator_{\{A_{\jj} = k\}}\vX_\jj\Y_{\jj}$. Let $\vC_{\ii}^\circ \defeq (\vC_{1,\ii}^\circ,\dots,\vC_{K,\ii}^\circ)$. By translation-invariance of $\policy_k^{(\n)}$, and setting $\boldsymbol{e} = \sqrt{\n}\vbeta$, we obtain $\policy_k^{(\n)}\big(\vS_{\ii}, \vC_{\ii}, \vx\big) = \policy_k^{(\n)}\big(\vS_{\ii}, \vC_{\ii} + \vS_{\ii}\boldsymbol{e}, \vx\big) = \policy_k^{(\n)}\big(\vS_{\ii}, \vC_{\ii}^\circ, \vx\big)$, which does not depend on $\vbeta$. Part (b) follows by taking pointwise limits in the identity $\policy_k^{(\n)}\big(\vs, \vc, \vx\big) = \policy_k^{(\n)}\big(\vs, \vc + \vs\boldsymbol{e}, \vx\big)$.
\end{proof}

\begin{proof}[Proof of Theorem~\ref{thm:well-posedness}]\label{S-appsubsec:thm-well-posedness-proof}
Throughout the regularization constant $\kappa>0$ is fixed. We prove the stated properties below for RTI tempered-greedy, tempered-LinUCB, and Thompson sampling. They then follow for RTI exploration sampling by smooth reweighting of the Thompson probabilities. We argue for the linear contextual case; the non-contextual case is recovered by setting $q=1$ and $\vX_\ii\equiv1$, so that $\vS_k=\freq_k$, $\vc_k=\reward_k$, $\operatorname{adj}(\vS_k)=1$, and $\det(\vS_k)=\freq_k$, under which the contextual regularized contrast~\eqref{eqn:regularized_contrast} reduces to its scalar form~\eqref{eqn:regularized_contrast_MAB}. We write
\begin{equation*}
\begin{aligned}
N_{k,k'}(\vS,\vc,\vx) &\defeq \det(\vS_{k'})\,\vx^\trans\!\operatorname{adj}(\vS_k)\vc_k-\det(\vS_k)\,\vx^\trans\!\operatorname{adj}(\vS_{k'})\vc_{k'}, \\
D_{k,k'}(\vS) &\defeq \max\!\big(\det(\vS_k)\det(\vS_{k'}),\kappa\big),
\end{aligned}
\end{equation*}
so that $\contrast_{k,k'}=N_{k,k'}/D_{k,k'}$.

For any square matrix $\vA$, the adjugate satisfies $\operatorname{adj}(\vA)\vA=\vA\operatorname{adj}(\vA)=\det(\vA)\idm_q$, an identity valid for \emph{all} $\vA$ including singular ones. Each entry of $\operatorname{adj}(\vS_k)$ is a signed $(q-1)\times(q-1)$ minor and is therefore a polynomial in the entries of $\vS_k$. Consequently, $N_{k,k'}$ is a polynomial in $(\vS,\vc,\vx)$. On the set $\{\det\vS_k\det\vS_{k'}\neq0\}$, we have $N_{k,k'}=\det(\vS_k)\det(\vS_{k'})\,\vx^\trans(\vS_k^{-1}\vc_k-\vS_{k'}^{-1}\vc_{k'})$.

\smallskip\noindent\emph{Proof of translation invariance.}
Under the shift $\vc_k\mapsto\vc_k+\vS_k\boldsymbol{e}$ ($\boldsymbol{e}\in\rR^q$), we have
\begin{align*}
 &~ N_{k,k'}(\vS,\vc+\vS\boldsymbol{e},\vx)-N_{k,k'}(\vS,\vc,\vx) \\
=&~ \det(\vS_{k'})\,\vx^\trans\!\operatorname{adj}(\vS_k)\vS_k\boldsymbol{e}-\det(\vS_k)\,\vx^\trans\!\operatorname{adj}(\vS_{k'})\vS_{k'}\boldsymbol{e}\\
=&~ \big(\det(\vS_{k'})\det(\vS_k)-\det(\vS_k)\det(\vS_{k'})\big)\,\vx^\trans\boldsymbol{e}=0 .
\end{align*}
\begin{sloppypar}
The denominator $D_{k,k'}$ depends on $\vS$ only. Hence $\contrast_{k,k'}$ is invariant. The exploration bonus $\sqrt{\vx^\trans(\vS_k\vee\kappa\idm_q)^{-1}\vx}$ and the Thompson posterior covariance depend on $\vS$ alone, hence are invariant. As the tempered-greedy, tempered-LinUCB, and Thompson choice probabilities are functions of the $\contrast_{k,k'}$'s and of these $\vS$-only terms, they satisfy Definition~\ref{def:translationinvariant}; the same computation applies verbatim to the finite-$\n$ schemes, since $\kappa$ does not depend on $\n$.
\end{sloppypar}

\smallskip\noindent\emph{Proof of Assumption~\ref{assm:policy_feasible}.} The translation invariance established above induces sampling schemes that do not depend on the global parameter $\pintb=\vbeta$.

\smallskip\noindent\emph{Proof of Assumption~\ref{assm:policy_limit}(a).}

\emph{Continuity.} $N_{k,k'}$ is a polynomial and $D_{k,k'}\ge\kappa>0$ is continuous (a maximum of continuous functions), so $\contrast_{k,k'}$ is jointly continuous on the whole state space, including the singular set $\{\det\vS_k=0\}$. The spectral cap $\vS_k\mapsto(\vS_k\vee\kappa\idm_q)^{-1}$ is continuous with operator norm at most $1/\kappa$, and, after removing the redundant $k$-th coordinate, the Gaussian orthant probability is continuous in its mean and (positive-definite) covariance. Composition with the (smooth) softmax and orthant maps yields joint continuity of $\policy_k$. Continuity of $\Epolicyb_k^{\circ}$ and $\Epolicyb_k^{\star}$ follows from dominated convergence.

\emph{Strict positivity.} On any bounded set $\mathcal{K}$ there is $M_N(\mathcal{K})<\infty$ with $|N_{k,k'}|\le M_N(\mathcal{K})$, whence $|\contrast_{k,k'}|\le M_N(\mathcal{K})/\kappa$. Since $\Contrast_k$ is an average of the $\contrast_{k,j}$, we also have $|\Contrast_k|\le M_N/\kappa$. The LinUCB bonus lies in $[0,\lambda M_X/\sqrt{\kappa}]$, where $M_X\defeq\sup_{\mathcal{K}}\|\vx\|$. Lower bounds on the numerators of \eqref{eqn:TG} and \eqref{eqn:LinUCB}, together with upper bounds on each of the $K$ terms in their denominators, yield
\begin{align*}
\policy_k^{\mathrm{TG}}\ge \frac{1}{K\exp(2\alpha M_N/\kappa)},\qquad
\policy_k^{\mathrm{LinUCB}}\ge \frac{1}{K\exp\big(2\alpha M_N/\kappa+\alpha\lambda M_X/\sqrt{\kappa}}\big)>0.
\end{align*}
For RTI Thompson sampling, the orthant probability in \eqref{eqn:TS} is likewise bounded below on $\mathcal{K}$ by a positive constant: its mean is bounded, and its covariance is positive definite uniformly on $\mathcal{K}$, as $\vphi_{k,\kappa}(\vx)\ge\|\vx\|^2/(\lambda_{\max}(\vS_k)\vee\kappa)\ge1/(\lambda_{\max}(\vS_k)\vee\kappa)$ by the intercept in $\vx$. Since $\policy_k(\vu,\vx)>0$ for every $\vx$, $\Exp[\policy_k(\vu,\vX)]>0$.

\emph{Uniform convergence.} The only dependence on $\n$ is through the hyperparameter vector $\eta_\n\defeq(\alpha_\n/\sqrt{\n},\lambda_\n,b_\n^2/\n)\to\eta\defeq(\alpha,\lambda,b^2)$. The map $(\eta,\vu,\vx)\mapsto\policy_k^{\eta}(\vu,\vx)$ is jointly continuous, hence \emph{uniformly} continuous on any compact set (Heine--Cantor); therefore $\sup_{(\vu,\vx)\in\mathcal{K}}|\policy_k^{\eta_\n}(\vu,\vx)-\policy_k^{\eta}(\vu,\vx)|\to0$ as $\eta_\n\to\eta$, for every compact $\mathcal{K}$.
By the lower bound on $\vphi_{k,\kappa}$ established above, the Thompson posterior covariance on the contrast space is uniformly positive definite on bounded state-context sets for $b$ in bounded intervals, including $b=0$. Multiplying by the weights $\funfb_k{\funfb_k}^{\trans}$ (resp.\ $\funqb_k$), of order $\|\vX\|^2$, and integrating against $\measure_X$ -- legitimate by dominated convergence applied to $\sup_{\vu}|\policy_k^{(\n)}(\vu,\vX)-\policy_k(\vu,\vX)|\,\|\vX\|^2$, which tends to $0$ pointwise and is bounded by a multiple of $\|\vX\|^2$ or $\FJbpintk(\vX)$, integrable by the moment conditions imposed -- gives uniform convergence on bounded $\vu$-sets of $\Exp[\policy_k^{(\n)}(\vu,\vX)\funfb_k{\funfb_k}^{\trans}]\to\Epolicyb_k^{\circ}(\vu)$ and $\Exp[\policy_k^{(\n)}(\vu,\vX)\funqb_k]\to\Epolicyb_k^{\star}(\vu)$.

\smallskip\noindent\emph{Proof of Assumption~\ref{assm:sde_solution}.}

We first consider the closed subsystem $(\vC_0,\vS_0)$, with drift blocks $H_k$ and diffusion blocks $H_k^{1/2}$, where $H_k(\vu)\defeq\Exp[\policy_k(\vu,\vX)\vX\vX^\trans]$.

\emph{Local Lipschitz.} On $\{\|(\vS,\vc,\vx)\|\le M\}$ the polynomial $N_{k,k'}$ is Lipschitz, $D_{k,k'}\ge\kappa$ is Lipschitz on bounded sets, so $\contrast_{k,k'}$ is Lipschitz with a constant $L_{\contrast}(\kappa,M)<\infty$; composing with $\exp$ (Lipschitz on the bounded range of $\Contrast_k$) and dividing by the denominator of \eqref{eqn:TG}, which is bounded below by $1$ because $\contrast_{k,j}=-\contrast_{j,k}$ implies $\sum_k\Contrast_k=0$ and hence $\max_k\Contrast_k\ge0$, shows $\policy_k^{\mathrm{TG}}$ is locally Lipschitz, and likewise for $\policy_k^{\mathrm{LinUCB}}$ (the map $\vS_k\mapsto(\vS_k\vee\kappa\idm_q)^{-1}$ is Lipschitz on bounded sets, the spectral cap being Lipschitz in operator norm and the inverse Lipschitz on $\{\|\cdot^{-1}\|\le1/\kappa\}$) and $\policy_k^{\mathrm{TS}}$ (the orthant probability is $C^1$ in its mean and covariance where the covariance is positive definite, which holds uniformly on bounded sets by the same lower bound on $\vphi_{k,\kappa}$). Hence $\policy_k(\vu,\vx)$ is locally Lipschitz in $\vu$, uniformly over $\vx$ in bounded sets; integrating against $\measure_X$ shows $\vu\mapsto H_k(\vu)$ is locally Lipschitz on bounded $\vu$-sets: the per-$\vx$ Lipschitz constant of $\policy_k(\cdot,\vx)$ is $O(\|\vx\|)$, as the state-Lipschitz constants of $\contrast_{k,j}$ and $\sqrt{\vphi_{k,\kappa}}$ are bounded by a constant times $\|\vx\|$ on bounded state sets, hence is $\measure_X$-integrable against the weight $\vx\vx^\trans$, of order $\|\vx\|^2$, since $\Exp\|\vX\|^{3}<\infty$ under Assumption~\ref{assm:DQM}.

\emph{Boundedness $\Rightarrow$ linear growth.}
Since $0<\policy_k\le1$, $\|H_k(\vu)\|_{\mathrm{op}} \le\|\Exp[\vX\vX^\trans]\|_{\mathrm{op}}<\infty$. Thus the subsystem's drift and diffusion are bounded and hence of linear growth.

\emph{Lipschitz diffusion.} For any unit vector $\vv$, $\vv^\trans H_k(\vu)\vv=\Exp[\policy_k(\vu,\vX)(\vv^\trans\vX)^2]>0$, since $\policy_k(\vu,\vx)>0$ for $\measure_X$-almost every $\vx$ and $\prob(\vv^\trans\vX\neq0)>0$ (as $\Exp[\vX\vX^\trans]\succ0$); hence $H_k(\vu)\succ0$ for every $\vu$. By the preceding local Lipschitz property, $\vu\mapsto H_k(\vu)$ is continuous, so $\vu\mapsto\lambda_{\min}(H_k(\vu))$ is continuous and strictly positive; it therefore attains a minimum $c(\mathcal{K})>0$ on each compact $\vu$-set $\mathcal{K}$, giving $H_k(\vu)\succeq c(\mathcal{K})\idm_q$ on $\mathcal{K}$. On the cone $\{\Sigma\succeq c\idm_q\}$ the principal square root is Lipschitz, $\|\Sigma_1^{1/2}-\Sigma_2^{1/2}\|\le\|\Sigma_1-\Sigma_2\|/(2\sqrt{c})$ (from the Sylvester identity $\Sigma_1^{1/2}\vA+\vA\Sigma_2^{1/2}=\Sigma_1-\Sigma_2$ for $\vA=\Sigma_1^{1/2}-\Sigma_2^{1/2}$). Composing with the locally Lipschitz $H_k$ shows $\vu\mapsto H_k(\vu)^{1/2}$ is locally Lipschitz. Strict positivity ensures this square-root argument applies. For the reweighted coefficients $\Exp[\policy_k^{1-2r}(\vu,\vX)\,\vX\vX^\trans]$ underlying the IPW/AW statistics, the exponent $1-2r\ge0$ when $r\le1/2$, so $\policy_k^{1-2r}\le1$ and finiteness follows from $\Exp\|\vX\|^2<\infty$; for $r>1/2$, however, the negative power $\policy_k^{-(2r-1)}$ requires additional integrability control. A uniform positive lower bound on $\policy_k$ suffices. The clipped schemes~\eqref{eqn:uniform_mixing} provide one: $\policy_k\ge\policy_{\min}$ gives $\policy_k^{1-2r}\le\policy_{\min}^{1-2r}$.

\emph{Clipped schemes.}
The affine transformation~\eqref{eqn:uniform_mixing} preserves translation invariance and the convergence, Lipschitz, and boundedness properties established above. Since $\policy_k^{(\n)},\policy_k\ge\policy_{\min}>0$, Assumption~\ref{assm:policy_limit}(b) holds and $H_k(\vu)\succeq\policy_{\min}\Exp[\vX\vX^\trans]\succ0$, so the square-root argument for the closed subsystem remains valid.

\emph{Conclusion.}
For both the original and clipped schemes, the closed $(\vC_0,\vS_0)$ subsystem has locally Lipschitz drift $H_k$ and diffusion $H_k^{1/2}$, both of linear growth. By the existence--uniqueness theorem for It\^o SDEs with locally Lipschitz coefficients of linear growth \citep[Thm~IV.3.1]{ikeda2014stochastic}, the subsystem considered above admits a unique strong solution on $[0,1]$. Equivalently, this follows by applying the global-Lipschitz result of \citet[Ch.~5, Thm~2.5]{karatzas2012brownian} through localization on successive balls, with linear growth ruling out explosion. The local Lipschitz regularity also holds at the initial state $\vu(0)=\mfzero$, so no separate boundary argument at $\uu=0$ is needed.
All remaining coefficients ($\Epolicyb^{\vartriangle}_k$, $\Epolicyb^{\vx}_{r_1;k}$, $\Epolicyb^{\vc}_{r_1;k}$, $\Epolicyb^{\vs}_{r_2;k}$) are bounded continuous functions of $(\vC_0,\vS_0)$, for $r_1,r_2\le0$ and, under clipping, for all $r_1,r_2$. Lemma~\ref{S-lem:subsystem}, with $\boldsymbol\xi^{(1)}=(\vC_0,\vS_0)$ and $\tilde{\boldsymbol\sigma}_1=\operatorname{diag}(H_k^{1/2})$, therefore yields Assumption~\ref{assm:sde_solution} for the full stacked system.
\end{proof}

\section{MAB: additional results} \label{S-appsec:MAB_additional}
\subsection{Hypothesis on a single arm} \label{S-subsec:MAB_evaluateonearm}
This subsection evaluates a single arm against a known baseline. Throughout we assume that the arm distributions have unit variance. For a known $\mu_0\in\rR$, let $\mu_{k,\n} = \mu_0 + \m_k/\sqrt{\n}$, $k\in\SK$, and consider
\begin{equation} \label{S-eqn:hyp_1}
H_0\colon \m_2 = 0 \quad\text{versus}\quad H_1\colon \m_2 > 0,
\end{equation}
with $\m_1,\m_3,\dots,\m_K$ as nuisance parameters.\footnote{Instead of the one-sided alternative $\m_2 > 0$, one could, of course, also consider $\m_2 < 0$ or $\m_2 \neq 0$. Additionally, one could consider a hypothesis in which $\mu_0$ is treated as a nuisance parameter. This, however, would complicate the analysis and is left for future research.} Unlike the comparison of two arms in Section~\ref{subsec:MAB_comparetwoarms}, the baseline is known, so there is no common level to remove: the nuisance parameters are the effects of the other arms, which affect the test through the sampling scheme. As in Section~\ref{subsec:MAB_comparetwoarms}, we assume the sampling schemes are RTI (Section~\ref{subsec:translation_invariant_policies}) and hence translation invariant in the sense of Definition~\ref{def:translationinvariant}. For notational convenience, we write $\policy_{k,\ii}$ for $\policy_k^{(\n)}(\freqb_{\ii-1}, \rewardb_{\ii-1})$ and $\policy_k(\uu)$ for $\policy_k(\freqb(\uu), \rewardb(\uu))$. Setting $\mu=\mu_0$ in~(\ref{eqn:reward_frequency_r1r2}), the representation~(\ref{eqn:SDERepar}) of the limit experiment reads, for $k\in\SK$,
\begin{equation} \label{S-eqn:SDERepar_onearm}
\dd\rewardd_{r;k}(\uu)
    =
\m_k\policy_k^{1-r}(\uu)\dd\uu +
    \policy_k^{1/2-r}(\uu)\dd\B_{\e_k}(\uu),
\end{equation}
with $\m_k$ the local effect of Arm~$k$ relative to $\mu_0$.

We exploit the structural MAB limit experiment~(\ref{eqn:SDE_MAB}) to study the asymptotic distributions of three statistics that are natural to consider for this problem: the classical Student's~t statistic, the Adaptively-Weighted (AW) statistic studied by \citet{hadad2021confidence,zhang2021statistical}, and the Inverse Propensity Weighting (IPW) statistic studied in, e.g., \citet{hadad2021confidence}. They are one-arm counterparts of the two-sample statistics of Section~\ref{subsec:MAB_comparetwoarms}, weighting the Arm-2 reward increments by $\policy_{2,\ii}^{-r}$ with $r=0$, $1/2$, and $1$, respectively. Their limit distributions follow via~(\ref{eqn:statistic_conv_joint_lr})--(\ref{eqn:statistic_conv_alt}) and can be read off from~(\ref{S-eqn:SDERepar_onearm}). As the AW and IPW statistics involve inverse sampling probabilities, we impose Assumption~\ref{assm:policy_limit}(b) throughout.

It turns out that only AW is distribution-free w.r.t the nuisance parameters $\m_1,\m_3,\dots,\m_K$ under $H_0$: in the limit experiment it is exactly standard normal, so standard-normal critical values yield an asymptotically valid test. By contrast, the asymptotic null distributions of the t and IPW statistics depend on these nuisance parameters (which, being parameters at the contiguity rate, cannot be estimated consistently). As such it is unclear, and perhaps even impossible, to obtain asymptotically valid tests from these statistics. Hence, perhaps surprisingly, the t and IPW tests cannot be used---assuming one insists on (asymptotic) validity of tests---for testing hypothesis~(\ref{S-eqn:hyp_1}).

\subsubsection{Test statistics}\label{S-subsubsec:MAB_onearmstatistics}

\subsubsection*{One-arm Student's t-test ($r=0$)}
Using $\mu=\mu_0$ in the definition of $\reward_{2,\n}$ (and exploiting known unit variance), the classical Student's t statistic is given by $\stat_{\n}^{\text{t}} = \reward_{2,\n}/\sqrt{\freq_{2,\n}}$. Result~\ref{ancillary:structurallimitexperiment_MAB} yields that the joint convergence in~(\ref{eqn:statistic_conv_joint_lr}) holds. Specifically, under $\lawn_{0,\dots,0}$, we have
$\left(\stat_\n^{\text{t}}, \log\left(\dd\lawn_{\m_1,\dots,\m_K}/\dd\lawn_{0,\dots,0}\right)\right) \wto \left(\stat^{\text{t}}, \log\left(\dd\law_{\m_1,\dots,\m_K}/\dd\law_{0,\dots,0}\right) \right)$, where $\stat^{\text{t}} = \reward_2(1)/\sqrt{\freq_2(1)}$. Invoking Le Cam's third lemma, see~(\ref{eqn:statistic_conv_alt}), yields $\stat_\n^{\text{t}} \wto \stat^{\text{t}}$ under $\lawn_{\m_1,\dots,\m_K}$ for all $\m_1,\dots,\m_K\in\rR$. By~(\ref{S-eqn:SDERepar_onearm}) with $r=0$, the distribution of $\stat^{\text{t}}$ under $\law_{\m_1,\dots,\m_K}$ thus is
\begin{align*}
\mathcal{L}(\stat^{\text{t}} \cond \law_{\m_1,\dots,\m_K}) = \mathcal{L}\left( \m_2\sqrt{\freq_{2}(1)} + \freq_{2}(1)^{-1/2}{\textstyle\int_0^1}\sqrt{\policy_2(\uu)}\dd\B_{\e_2}(\uu) \cond \law_{\m_1,\dots,\m_K}\right).
\end{align*}

We make two observations. First, the distribution of $\stat^{\text{t}}$ is generally not normal. This is because, in $\int_0^1\sqrt{\policy_2(\uu)}\dd\B_{\e_2}(\uu)$, the integrand process $\sqrt{\policy_2(\uu)} = \sqrt{\policy_2(\freqb(\uu),\rewardb(\uu))}$ depends on the process $\rewardb(\uu)$, which in turn depends on the integrator process $\B_{\e_2}(\uu)$. This is the same reason why the MAB limit experiment is not normal, and has been documented in recent literature, including \citet{deshpande2018accurate}, \citet{zhang2020inference}, and \citet{hadad2021confidence}. Second, under the null hypothesis where $\m_2 = 0$ and $\m_1,\m_3,\dots,\m_K\in\rR$, $\stat^{\text{t}}$ is not distribution-free with respect to $\m_1,\m_3,\dots,\m_K$, as they influence $\rewardb(\uu)$, which in turn affects $\policy_2(\freqb(\uu),\rewardb(\uu))$. Figure~\ref{S-fig:MAB_statM_ThompsonInv} provides Monte Carlo evidence for the case of two arms, showing that the null distribution of $\stat^{\text{t}}$ indeed varies with $\m_1$. To the best of our knowledge, we are the first to highlight this issue, which becomes evident through the structural limit experiment.

\subsubsection*{One-arm Adaptively-Weighted (AW) test ($r=1/2$)}
Adaptively weighted estimation of arm means was developed by \citet{hadad2021confidence}, while \citet{zhang2021statistical} developed a general framework for adaptively weighted M-estimation, including contextual bandits. Using $\mu=\mu_0$ and $r = 1/2$ in the definition of $\rewardd_{r;2,\n}$, we consider the simplest version of the AW test statistic,\footnote{The general version of the AW statistic is given by $\stat^{\text{AW}}_{\n}(\psista) \defeq \widebar\psista_2^{-1/2}\sum_{\ii=1}^{\n}\sqrt{\psista_2(\ii/\n)/\policy_{2,\ii}}\,(\reward_{2,\ii} - \reward_{2,\ii-1})$, where $\psista_2: [0,1] \to (0,1]$ is a deterministic, so-called variance stabilizing policy and $\widebar\psista_2 \defeq \int_0^1\psista_2(\uu)\dd\uu$. The simplest version corresponds to a constant $\psista_2$.}
\begin{align*}
\stat^{\text{AW}}_{\n}
\defeq \rewardd_{1/2;2,\n}
= \frac{1}{\sqrt{\n}}\sum_{\ii=1}^{\n} \frac{\indicator_{\{A_{\ii} = 2\}}}{\sqrt{\policy_{2,\ii}}}(\Y_{\ii} - \mu_0).
\end{align*}
As for the t statistic, Result~\ref{ancillary:structurallimitexperiment_MAB} and Le Cam's third lemma yield $\stat_\n^{\text{AW}} \wto \stat^{\text{AW}} \defeq \rewardd_{1/2;2}(1) = \W_{\e_2}(1)$ under $\lawn_{\m_1,\dots,\m_K}$ for all $\m_1,\dots,\m_K\in\rR$. By~(\ref{S-eqn:SDERepar_onearm}) with $r=1/2$, the diffusion coefficient equals one, so that the distribution of $\stat^{\text{AW}}$ under $\law_{\m_1,\dots,\m_K}$ is
\begin{equation}\label{S-eqn:AW1arm_shift}
\mathcal{L}(\stat^{\text{AW}} \cond \law_{\m_1,\dots,\m_K}) = \mathcal{L}\left( \m_2\int_0^1\sqrt{\policy_2(\uu)}\dd\uu + \B_{\e_2}(1) \cond \law_{\m_1,\dots,\m_K}\right).
\end{equation}

Note that $\B_{\e_2}$ is a standard Brownian motion under every $\law_{\m_1,\dots,\m_K}$. Consequently, under the null hypothesis ($\m_2=0$), $\stat^{\text{AW}} = \B_{\e_2}(1) \sim \mathcal{N}(0,1)$ exactly, whatever the values of $\m_1,\m_3,\dots,\m_K$, making it distribution-free with respect to these nuisance parameters. Standard-normal critical values thus yield an asymptotically valid test. As for the TS-AW pivot in~(\ref{eqn:TSAWpivot}), exact normality stems from normalizing each reward increment by its predictable standard deviation $\sqrt{\policy_{2,\ii}}$. Moreover, (\ref{S-eqn:AW1arm_shift}) shows that the local asymptotic power increases with $\m_2$, as well as with $\int_0^1\sqrt{\policy_2(\uu)}\dd\uu$, a random term that becomes larger when Arm~$2$ is sampled more frequently. This random term leads, in general, to a non-Gaussian distribution of the AW statistic under the alternative ($\m_2\neq 0$).

\subsubsection*{One-arm Inverse Propensity Weighting (IPW) test ($r=1$)}
Using $\mu=\mu_0$ and $r = 1$ in the definition of $\rewardd_{r;2,\n}$, the one-arm IPW statistic is
\begin{align*}
\stat^{\text{IPW}}_\n
\defeq \rewardd_{1;2,\n}
= \frac{1}{\sqrt{\n}}\sum_{\ii=1}^{\n} \frac{\indicator_{\{A_{\ii} = 2\}}}{\policy_{2,\ii}}(\Y_{\ii} - \mu_0),
\end{align*}
which, by Result~\ref{ancillary:structurallimitexperiment_MAB} and Le Cam's third lemma, weakly converges to $\stat^{\text{IPW}} \defeq \rewardd_{1;2}(1)$ under $\lawn_{\m_1,\dots,\m_K}$ for all $\m_1,\dots,\m_K\in\rR$. By~(\ref{S-eqn:SDERepar_onearm}) with $r=1$, the drift coefficient equals $\m_2$, so that the distribution of $\stat^{\text{IPW}}$ under $\law_{\m_1,\dots,\m_K}$ is
\begin{equation}\label{S-eqn:IPW1arm}
\mathcal{L}(\stat^{\text{IPW}} \cond \law_{\m_1,\dots,\m_K}) = \mathcal{L}\left( \m_2 + {\textstyle\int_0^1}\policy_2^{-1/2}(\uu)\dd\B_{\e_2}(\uu) \cond \law_{\m_1,\dots,\m_K}\right).
\end{equation}
Similar to the t statistic, this distribution is generally not normal and, under $H_0$, depends on $\m_1,\m_3,\dots,\m_K$ through $\policy_2$.

\subsubsection*{An upper bound to the local asymptotic power}
By the asymptotic representation theorem, an upper bound to the power of valid tests in the limit experiment bounds the local asymptotic power of asymptotically valid tests. Fix $\mref_1,\dots,\mref_K\in\rR$ with $\mref_2>0$ and consider, in the limit experiment, the auxiliary hypothesis $H_0\colon (\m_1,\dots,\m_K) = (\mref_1,0,\mref_3,\dots,\mref_K)$ versus $H_1\colon (\m_1,\dots,\m_K) = (\mref_1,\dots,\mref_K)$. Its Neyman--Pearson (NP) test statistic is
\begin{align*}
\stat^{\text{NP}}
&\defeq \log\frac{\dd\law_{\mref_1,\mref_2,\mref_3,\dots,\mref_K}}{\dd\law_{\mref_1,0,\mref_3,\dots,\mref_K}}
 = \log\frac{\dd\law_{\mref_1,\mref_2,\mref_3,\dots,\mref_K}}{\dd\law_{0,\dots,0}} -\log\frac{\dd\law_{\mref_1,0,\mref_3,\dots,\mref_K}}{\dd\law_{0,\dots,0}}   \\
&= \LLRlim_2(\mref_2)
= \mref_2\reward_{2}(1) - \frac{1}{2}\mref_2^2\freq_{2}(1).
\end{align*}
Letting $\cv^*_{\alpha}=\cv^*_{\alpha}(\mref_1,\dots,\mref_K)$ denote the $1-\alpha$ quantile of $\stat^{\text{NP}}$ under $\law_{\mref_1,0,\mref_3,\dots,\mref_K}$, the NP test has power
\begin{equation*}
	\power_{\alpha}^*(\mref_1,\dots,\mref_K) = \Exp_{0,\dots,0}\left[\indicator_{\left\{\stat^{\text{NP}} > \cv_{\alpha}^* \right\}}\frac{\dd\law_{\mref_1,\dots,\mref_K}}{\dd\law_{0,\dots,0}}\right]
    = \law_{\mref_1,\dots,\mref_K}
    \left(  \stat^{\text{NP}} > \cv_{\alpha}^* \right).
\end{equation*}
Every valid test of~(\ref{S-eqn:hyp_1}) in the limit experiment has level $\alpha$ for the auxiliary hypothesis, so, the NP test being most powerful, $\power_{\alpha}^*(\mref_1,\dots,\mref_K)$ bounds the local asymptotic power of any asymptotically valid test of~(\ref{S-eqn:hyp_1}) under $\lawn_{\mref_1,\dots,\mref_K}$. Whether this bound is attained by a feasible asymptotically valid test remains open; Figure~\ref{S-fig:MABlim_powerMalgorithms_AW} shows that the one-arm AW test comes close to it.

\subsubsection{Simulation results}\label{S-subsubsec:montecarlo_MAB_onearm}
We use the setting of Section~\ref{subsec:montecarlo_MAB_twoarms} with $K = 2$: the same RTI sampling schemes and tuning parameters, clipping at $\policy_{\min}=0.01$, significance level $\alpha=5\%$, and $20{,}000$ replications. The SDEs~(\ref{eqn:SDE_MAB}) are simulated by an Euler scheme with $200$ grid points. In the sequence of experiments, rewards are Gaussian with unit variance and means $\mu_{k,\n}=\m_k/\sqrt{\n}$ (i.e., $\mu_0 = 0$), and each arm is pulled once in the first $K$ rounds.

\begin{figure}[!htb]
\centering

\includegraphics[width=\linewidth]{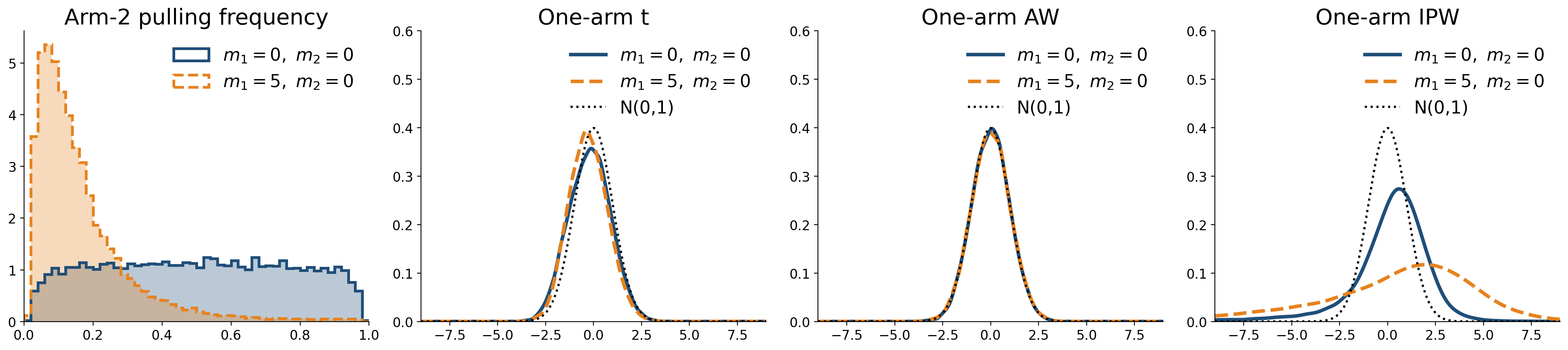}

\includegraphics[width=\linewidth]{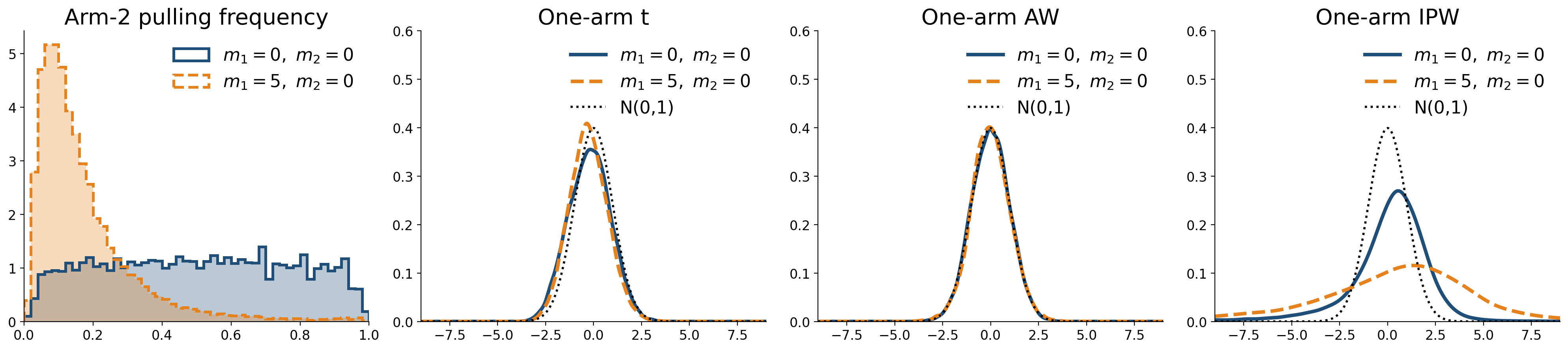}
\caption{{\small Histograms of Arm-2 pulling frequency and density estimates for the one-arm t, AW, and IPW statistics (from left to right) in the MAB limit experiment (top) and in finite-sample MAB experiments with $\n = 200$ (bottom), with $K = 2$ under RTI Thompson sampling. Parameter settings: $\m_1 = \m_2 = 0$ (solid blue) and $\m_1 = 5$, $\m_2 = 0$ (dashed orange).}}
\label{S-fig:MAB_statM_ThompsonInv}
\end{figure}

Figure~\ref{S-fig:MAB_statM_ThompsonInv} compares the null distributions of Arm-2 pulling frequency and the one-arm t, AW, and IPW statistics under RTI Thompson sampling, with $\m_2 = 0$ and $\m_1 = 0, 5$. Increasing $\m_1$ shifts the allocation away from Arm~2 and changes the null distributions of the t and IPW statistics: the density of the t statistic shifts to the left and becomes more peaked, while the IPW statistic is markedly non-normal and spreads out as Arm~2 is sampled less. The density of the AW statistic is indistinguishable from the standard-normal density for both values of $\m_1$, corroborating that it is distribution-free with respect to $\m_1$. The finite-sample distributions with $\n = 200$ (bottom) closely resemble their limiting counterparts (top). Size results for the one-arm AW test are reported in Table~\ref{S-table:sizes} in Appendix~\ref{S-appsec:additionalsimulations}.

\begin{figure}[ht]
\centering
\centering
\includegraphics[width=0.415\linewidth]{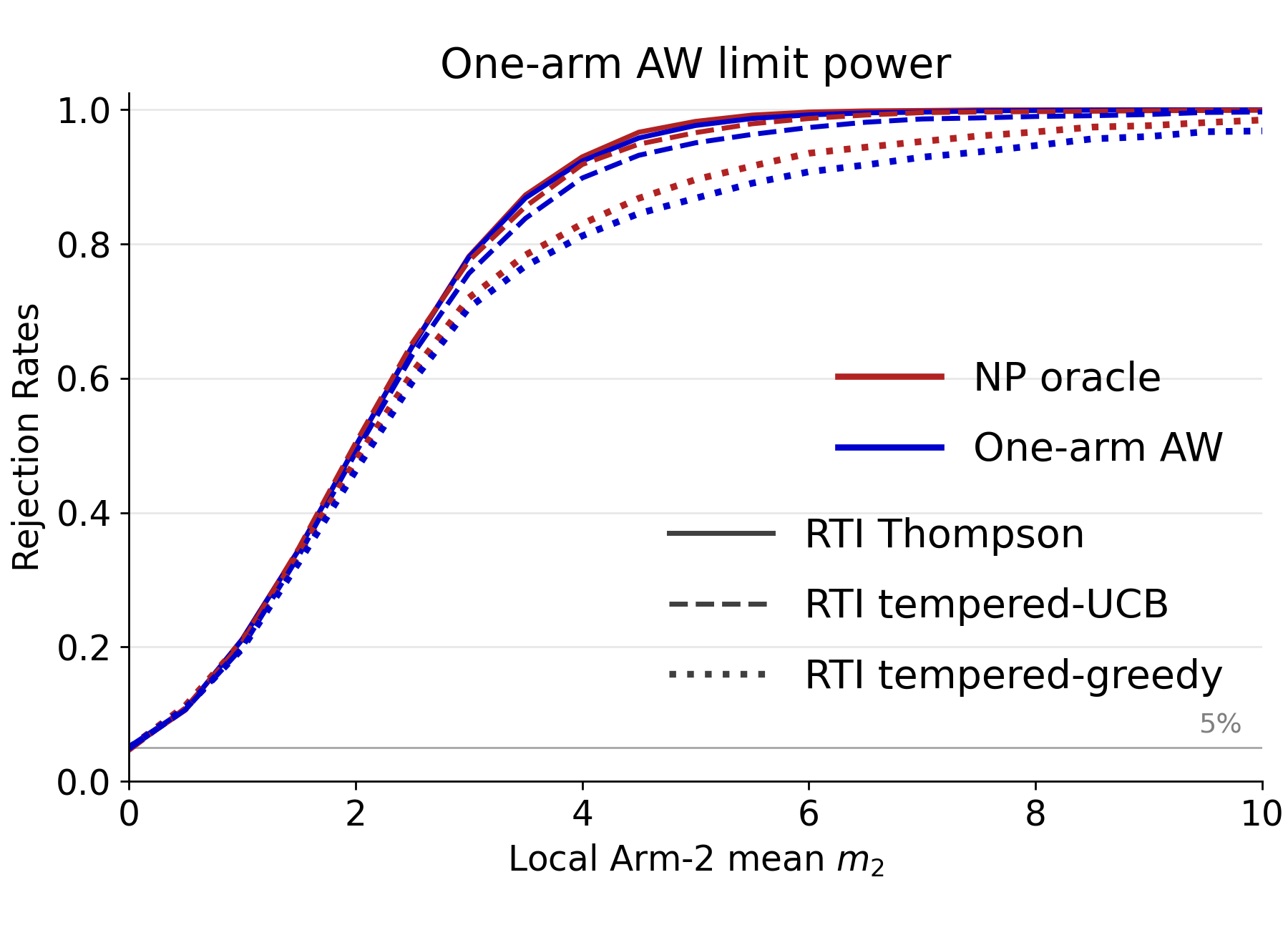}
\includegraphics[width=0.415\linewidth]{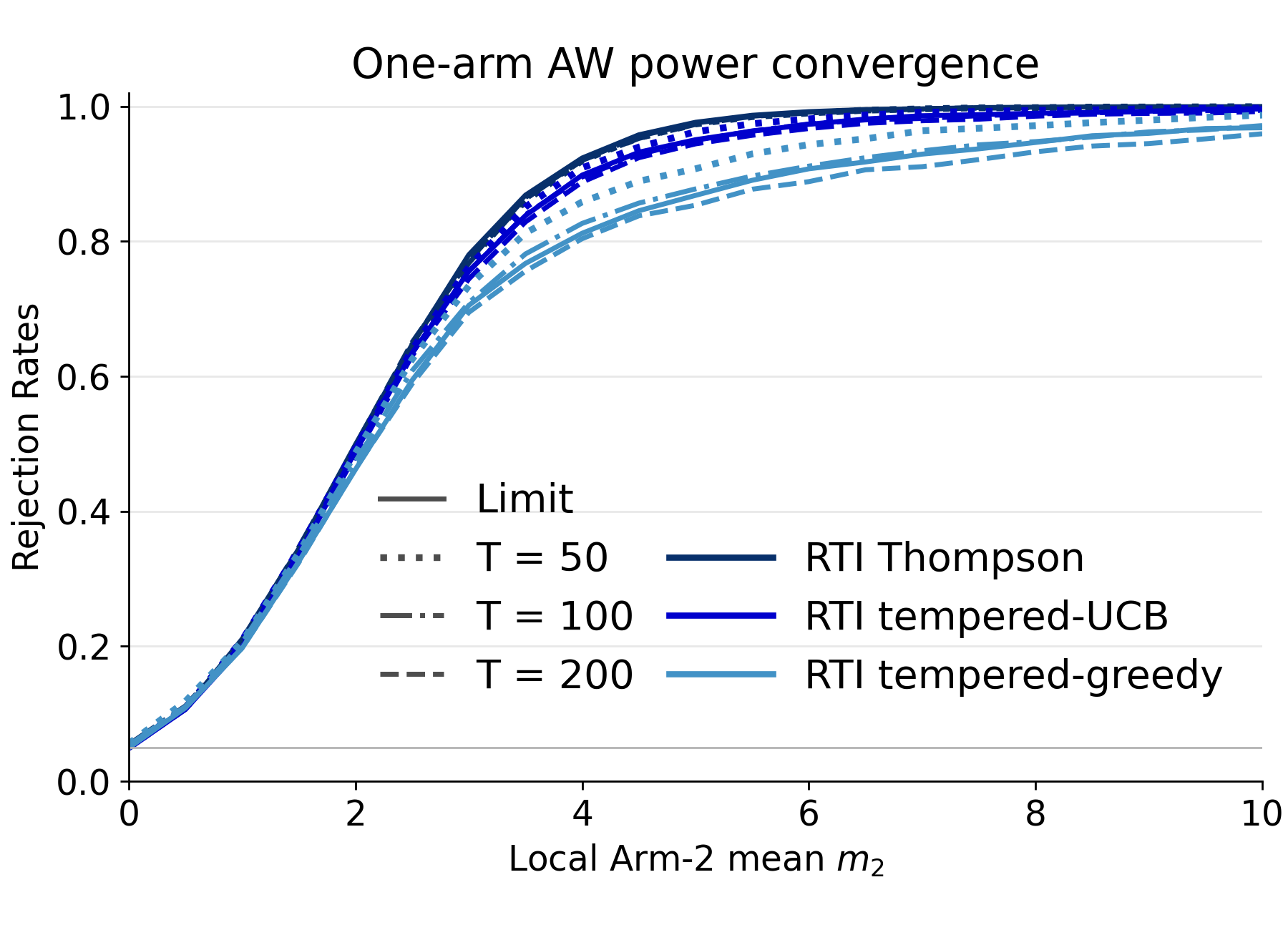}
\caption{{\small Left: asymptotic powers of the NP test (red) and the one-arm AW test (blue) under RTI Thompson sampling (solid), RTI tempered-UCB (dashed), and RTI tempered-greedy (dotted). Right: asymptotic (solid) and finite-sample powers for $\n = 50$ (dotted), $100$ (dash-dotted), and $200$ (dashed) of the one-arm AW test under RTI Thompson sampling (dark blue), RTI tempered-UCB (blue), and RTI tempered-greedy (light blue). In both panels $\m_1 = 0$.}}
\label{S-fig:MABlim_powerMalgorithms_AW}
\end{figure}

Turn next to the power results. Since only the AW test is asymptotically valid, the power analysis considers this test, together with the NP upper bound. The left panel of Figure~\ref{S-fig:MABlim_powerMalgorithms_AW} compares the asymptotic power of the one-arm AW test with the NP bound under the three RTI sampling schemes, with $\m_1 = 0$.\footnote{The AW test is distribution-free with respect to $\m_1$ only under the null; its power still depends on $\m_1$.} As in the two-arm comparison of Section~\ref{subsec:montecarlo_MAB_twoarms}, power is highest under RTI Thompson sampling, followed by RTI tempered-UCB and RTI tempered-greedy: sampling schemes that favor exploration tend to yield higher power.
The AW test nearly attains the NP bound under all three schemes.
The right panel compares finite-sample and asymptotic powers. Under RTI Thompson sampling and RTI tempered-UCB, finite-sample power is within about $0.01$ of its limit for all $\n\geq 50$. Under RTI tempered-greedy the differences are larger and not monotone in $\n$: power exceeds its limit by up to $0.05$ at $\n = 50$, is close to it at $\n = 100$, and falls slightly below it at $\n = 200$. There is no theoretical reason for convergence to be monotone, but this is yet another uncommon feature of the MAB problem.

\subsection{Testing equality of all arms}\label{S-subsec:MAB_omnibus}
Beyond comparing two arms, one may ask whether all arms share a common mean,
\begin{equation}\label{S-eqn:hyp_omnibus}
H_0:\ \m_1=\dots=\m_K \quad\text{versus}\quad  H_1:\ \text{not all }\m_k\text{ equal},
\end{equation}
with the common local mean $\m$ (and, in finite samples, the global mean $\mu$) as nuisance. The TS-$t$ test and the TS-AW pivot of that section both extend to this hypothesis. Both resulting statistics are distribution-free with respect to $\m$ under $H_0$ (and their finite-sample versions with respect to $\mu$), so both yield asymptotically valid tests; unlike the two-arm comparison with $K>2$, $H_0$ leaves no further nuisance parameters. We define the statistics in the limit experiment; the feasible finite-sample tests replace $\reward_k(1)$, $\freq_k(1)$, and $\policy_k(\uu)$ by $\reward_{k,\n}$, $\freq_{k,\n}$, and $\policy_{k,\ii}$, and use the critical values of the limit statistics; for the OM-pivot test, as for the two-sample pivot, each reward is moreover centered at the pooled mean of the previous rounds, a predictable baseline for the unknown common level.

\subsubsection{Test statistics}

\subsubsection*{Omnibus $t$ test (OM-$t$)}
The OM-$t$ statistic is the classical between-arm sum of squares
\[
\stat^{\text{OM-t}} \defeq \sum_{k=1}^K \freq_k(1)\big(\bar\reward_k-\bar\reward\big)^2,\qquad \bar\reward_k\defeq \frac{\reward_k(1)}{\freq_k(1)},\quad \bar\reward\defeq\frac{\textstyle\sum_{k}\reward_k(1)}{\textstyle\sum_{k}\freq_k(1)},
\]
which equals the Wald statistic of the contrasts $\bar\reward_k-\bar\reward_1$, $k=2,\dots,K$, and reduces to $(\stat^{\text{TS-t}})^2$ for $K=2$. Under an RTI scheme its null distribution is free of $\m$ but, like that of $\stat^{\text{TS-t}}$, not $\chi^2_{K-1}$, as the limit experiment is not LAN (Section~\ref{sec:limitexperiment}); critical values are simulated from the limit experiment at $\m=0$.

\subsubsection*{Omnibus pivot (OM-pivot)}
Let $\dd\boldsymbol{X}(\uu)$ collect the instantaneous contrasts $\dd\reward_k(\uu)/\policy_k(\uu)-\dd\reward_1(\uu)/\policy_1(\uu)$, $k=2,\dots,K$. Their conditional covariance is $\boldsymbol{C}(\uu)\dd\uu$, with $\boldsymbol{C}(\uu)\defeq\operatorname{diag}\big(\policy_2^{-1}(\uu),\dots,\policy_K^{-1}(\uu)\big)+\policy_1^{-1}(\uu)\boldsymbol{1}\boldsymbol{1}^\trans$, and
\[
\stat^{\text{OM-P}}\defeq\Big\|\int_0^1\boldsymbol{C}(\uu)^{-1/2}\dd\boldsymbol{X}(\uu)\Big\|^2\sim\chi^2_{K-1}\quad\text{exactly under }H_0
\]
for any predictable sampling scheme, so $\chi^2_{K-1,1-\alpha}$ critical values yield an asymptotically valid test. This follows by the argument for~(\ref{eqn:TSAWpivot}): the drift vanishes under $H_0$, and the martingale part is a $(K-1)$-dimensional standard Brownian motion at time one. For $K=2$, $\stat^{\text{OM-P}}=(\stat^{\text{TS-AW-P}})^2$. The finite-sample version uses a predictable pooled baseline, as the two-sample pivot does.

\subsubsection{Simulation results}\label{S-subsubsec:montecarlo_MAB_omnibus}
\begin{table}[!htb]
\centering
\caption{Omnibus tests of equal arm means under RTI Thompson sampling}
\label{S-tab:omnibus}
\begin{threeparttable}
\footnotesize
\setlength{\tabcolsep}{3pt}
\begin{tabular*}{\textwidth}{@{\extracolsep{\fill}} l *{12}{w{c}{2.4em}} @{}}
\toprule
 & \multicolumn{6}{c}{$K=3$} & \multicolumn{6}{c}{$K=4$} \\
\cmidrule(lr){2-7} \cmidrule(lr){8-13}
 & \multicolumn{2}{c}{$\dlpint=0$} & \multicolumn{2}{c}{$\dlpint=5$} & \multicolumn{2}{c}{$\dlpint=10$} & \multicolumn{2}{c}{$\dlpint=0$} & \multicolumn{2}{c}{$\dlpint=5$} & \multicolumn{2}{c}{$\dlpint=10$} \\
\cmidrule(lr){2-3} \cmidrule(lr){4-5} \cmidrule(lr){6-7} \cmidrule(lr){8-9} \cmidrule(lr){10-11} \cmidrule(lr){12-13}
$\n$ & limit & 200 & limit & 200 & limit & 200 & limit & 200 & limit & 200 & limit & 200 \\
\midrule
OM-$t$   & 4.9 & 4.9 & 40.7 & 39.7 & 92.4 & 92.0 & 5.3 & 4.9 & 41.4 & 40.6 & 96.6 & 96.3 \\
OM-pivot & 4.9 & 5.3 & 29.0 & 28.1 & 67.1 & 60.9 & 5.2 & 5.4 & 29.8 & 28.7 & 77.6 & 71.9 \\
\bottomrule
\end{tabular*}
\begin{tablenotes}[flushleft]
\footnotesize
\item Rejection rates in \% at nominal level $5\%$ in the limit experiment and for $\n=200$, with clipping at $\policy_{\min}=0.01$ and $20{,}000$ replications; Arm~2 is shifted by $\dlpint$ (local mean $\m_2=\dlpint$, other arms equal), so $\dlpint=0$ gives the size. Critical values are $\chi^2_{K-1,0.95}$ for the OM-pivot test and simulated from the limit experiment for the OM-$t$ test; the limit sizes use replications independent of those calibration draws.
\end{tablenotes}
\end{threeparttable}
\end{table}

We use the setting of Section~\ref{S-subsubsec:montecarlo_MAB_onearm} under RTI Thompson sampling, with $K=3,4$ and, under the alternative, Arm~2 shifted by $\dlpint$; Table~\ref{S-tab:omnibus} reports the results. Both tests have size close to $5\%$, in the limit experiment and for $\n=200$: the OM-pivot test with $\chi^2_{K-1}$ critical values, in line with its exact null distribution, and the OM-$t$ test with critical values simulated from the limit experiment. As in the two-arm comparison, the OM-$t$ test is considerably more powerful: at $\dlpint=10$ its asymptotic power is $92\%$ and $97\%$, against $67\%$ and $78\%$ for the OM-pivot test. Its finite-sample power is also closer to the limit: for $\n=200$ it is within about $1$ percentage point, whereas the OM-pivot test loses up to $6$ points, mirroring the slower power convergence of the two-sample pivot (Section~\ref{subsec:montecarlo_MAB_twoarms}).

\section{Additional Simulation Results} \label{S-appsec:additionalsimulations}

\subsection{Finite-sample sizes of the MAB tests}\label{S-appsubsec:sizes}

\begin{table}[ht]
\centering
\caption{Finite-sample sizes of the MAB tests ($K=2$)}
\label{S-table:sizes}
\begin{threeparttable}
\footnotesize
\setlength{\tabcolsep}{3pt}
\begin{tabular}{l *{4}{c} @{\hskip 1.5em} *{4}{c} @{\hskip 1.5em} *{4}{c}}
\toprule
 & \multicolumn{4}{c}{RTI Thompson} & \multicolumn{4}{c}{RTI tempered-greedy} & \multicolumn{4}{c}{RTI tempered-UCB} \\
\cmidrule(lr){2-5} \cmidrule(lr){6-9} \cmidrule(lr){10-13}
$\n$ & 50 & 100 & 200 & 500 & 50 & 100 & 200 & 500 & 50 & 100 & 200 & 500 \\
\midrule
One-arm AW          & 5.29 & 5.22 & 5.13 & 5.07 & 5.59 & 5.36 & 5.33 & 5.04 & 5.52 & 5.17 & 5.05 & 5.15 \\
Two-sample $t$      & 4.87 & 4.86 & 5.17 & 5.06 & 5.22 & 5.22 & 4.91 & 5.15 & 4.57 & 4.64 & 5.08 & 4.91 \\
Two-sample AW       & 4.17 & 4.69 & 4.76 & 4.85 & 4.69 & 5.00 & 4.50 & 5.17 & 4.20 & 4.59 & 4.78 & 4.70 \\
Two-sample AW pivot & 5.04 & 5.17 & 5.18 & 4.93 & 5.59 & 5.57 & 5.05 & 5.16 & 5.09 & 5.29 & 4.91 & 4.78 \\
Two-sample IPW      & 7.11 & 6.46 & 5.83 & 5.22 & 2.66 & 6.29 & 7.69 & 6.78 & 6.46 & 7.22 & 7.30 & 5.99 \\
\bottomrule
\end{tabular}
\begin{tablenotes}[flushleft]
\footnotesize
\item Empirical sizes (in \%) at nominal level $5\%$, with $K=2$ and $\mu_0=0$, clipping at $\policy_{\min}=0.01$, for the two-sample t, AW, AW pivot, and IPW tests comparing the two arms (Section~\ref{subsec:MAB_comparetwoarms}) and the one-arm AW test for evaluating a single arm (Section~\ref{S-subsec:MAB_evaluateonearm}). Critical values are standard normal for the one-arm AW test and the two-sample AW pivot, and simulated from the limit experiment for the other tests.
\end{tablenotes}
\end{threeparttable}
\end{table}

%
%
%
%
%

\subsection{Bernoulli multi-armed bandits}
We consider a $K$-armed Bernoulli MAB problem in which, for each $k\in\SK$, the potential outcomes $\Z_{k,t}$ are i.i.d.\ Bernoulli with success probability $p_k\in(0,1)$. With parameter $\pintb=(p_1,\dots,p_K)' \in (0,1)^K$, The DQM condition in Assumption~\ref{assm:DQM} is satisfied with score function $\scorek = \frac{z-p_k}{p_k(1-p_k)}$ and Fisher information $\FJ_k = \frac{1}{p_k(1-p_k)}$.

Adopting the equal-arms localization $p_{k,T} = p_0 + \frac{m_k}{\sqrt{\n}}$ around a common value $p_0\in(0,1)$, Proposition~\ref{prop:LAQ} and Result~\ref{ancillary:structurallimitexperiment_MAB} apply directly, with $\scorek = \frac{z-p_0}{p_0(1-p_0)}$ and $\FJ_k = \frac{1}{p_0(1-p_0)}$, and with the innovation $\varepsilon_{k,t}=Z_{k,t}-p_{k,T}$ now having limiting variance $p_0(1-p_0)$ rather than unity. The structural representation in Result~\ref{ancillary:structurallimitexperiment_MAB} carries over after a straightforward rescaling of $W_{\varepsilon_k}$ to reflect this variance. Because the score and Fisher information depend on $p_0$, however, translation-invariance of the sampling schemes and distribution-freeness of the test statistics are no longer available. Nonetheless, we show below that the limit experiment still provides accurate approximations under both the null and the local alternative.

In this simulation study, we consider the tempered-UCB algorithm described in Section~\ref{sec:sampling}, with the empirical success rate in place of the RTI contrast, and hyperparameters $\alpha = 1$ and $\lambda = 1$. To account for the bounded support of Bernoulli rewards, we include a brief warm-up phase lasting $20\%$ of the total horizon (i.e., $20$ rounds for each arm when $T = 200$), during which each arm is selected with equal probability (i.e., $\policy_k^{(\n)} = 1/K$ for all $k \in \SK$). In the limit experiment, this corresponds to $\policy_k = 1/K$ over the time interval $[0, 0.2]$ for all $k \in \SK$. Figure~\ref{S-fig:BernMAB_statD_temperedUCB} reports histograms of the arm-pulling frequencies and of the two-sample $t$-statistic under the null $\m_1=\m_2=0$ (left two panels) and under the alternative $\m_1=0$, $\m_2=2$ (right two panels), for the limit experiment and for the sequence experiment with $T=200$ rounds, overlaid in each panel. The sequence and limit distributions are nearly indistinguishable in both regimes, confirming that the limit experiment provides a good approximation.

\begin{figure}[!htb]
\centering
\centering

\includegraphics[width=0.923\linewidth]{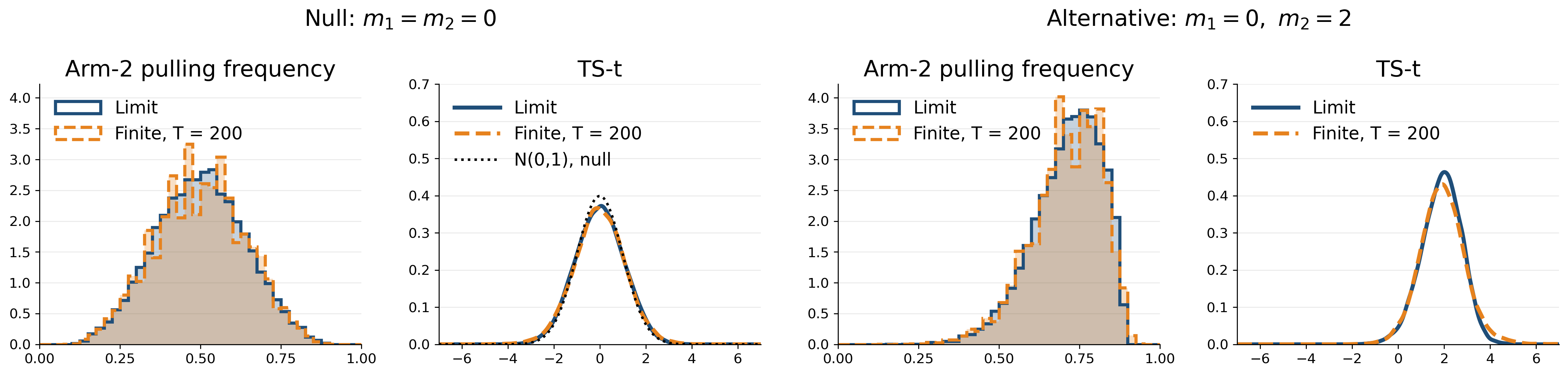}
\caption{{\small Histograms of Arm-2 pulling frequency and density estimates for TS-t in the Bernoulli MAB with $K = 2$ and $p_0 = 0.4$ under tempered-UCB, in the limit experiment (solid blue) and in finite-sample experiments with $\n = 200$ (dashed orange). Parameter settings: $\m_1 = \m_2 = 0$ (left two panels) and $\m_1 = 0$, $\m_2 = 2$ (right two panels).}}
\label{S-fig:BernMAB_statD_temperedUCB}
\end{figure}

\subsection{Exploration sampling}
Figure~\ref{S-fig:MABlim_statD_ExplorationInvK} is the exploration-sampling version of the top row of Figure~\ref{fig:MABlim_statD_Thompson_ThompsonInv}: RTI Thompson sampling is replaced by RTI exploration sampling (Section~\ref{subsec:translation_invariant_policies}), with $K=4$ arms sharing a common local mean. Although exploration sampling is designed to favor exploration, shifting sampling weight away from the arm that currently looks best, the null distribution of the TS-$t$ statistic remains clearly non-normal, while those of the TS-AW statistic and the TS-AW pivot stay close to the standard normal; all are invariant to the common mean.

\begin{figure}[!htb]
\centering

\includegraphics[width=\linewidth]{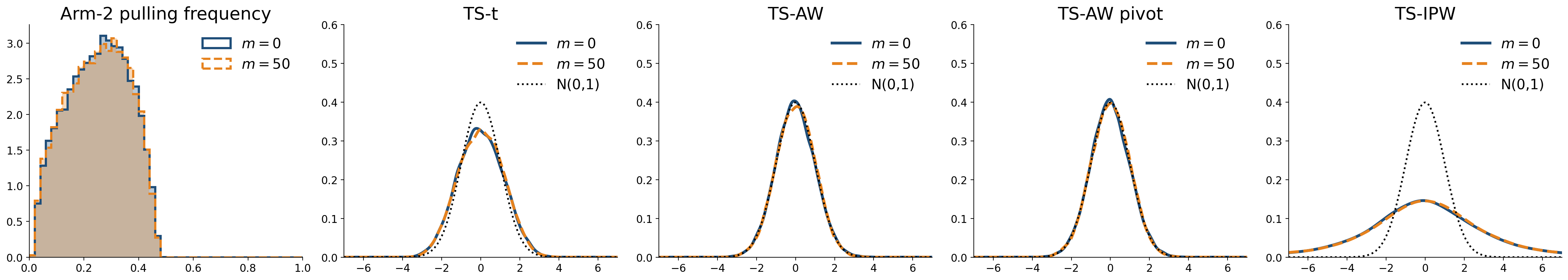}
\caption{{\small Histograms of Arm-2 pulling frequency and density estimates for TS-t, TS-AW, TS-AW pivot, and TS-IPW (from left to right) in the MAB limit experiment with $K = 4$ under RTI exploration sampling. Parameter settings: $\m_1 = \dots = \m_4 = 0$ (solid blue) and $\m_1 = \dots = \m_4 = 50$ (dashed orange). Dotted black curves show the standard-normal density.}}
\label{S-fig:MABlim_statD_ExplorationInvK}
\end{figure}

\subsection{Performance under heavy-tailed rewards}

\begin{figure}[ht]
\centering

\includegraphics[width=\linewidth]{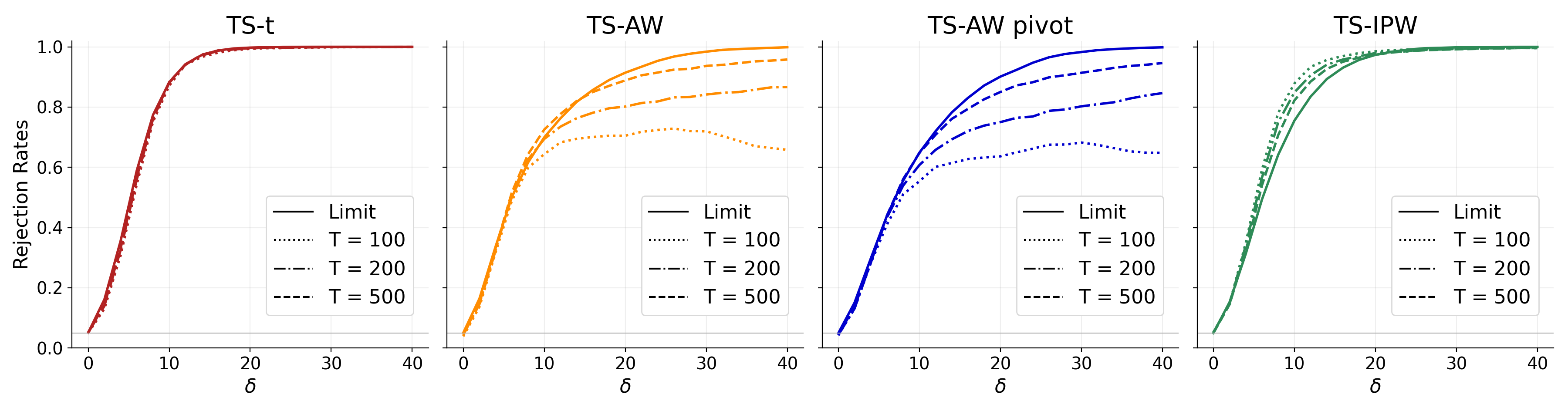}
\caption{{\small MAB asymptotic (solid) and finite-sample powers for $\n = 100$ (dotted), $200$ (dash-dotted), and $500$ (dashed) of TS-t (red), TS-AW (orange), TS-AW pivot (blue), and TS-IPW (green) (from left to right) with $K=2$ under RTI Thompson sampling with Student $t_3$ reward errors.}}

\label{S-fig:MABlim_powerDconvergence_ThompsonInv}
\end{figure}

\subsection{Regrets of RTI sampling schemes}
\label{S-subsec:MAB_regret}

Figure~\ref{S-fig:MABseq_regret_ThompsonInvK} compares the RTI schemes, specified as in Section~\ref{subsec:montecarlo_MAB_twoarms}, with greedy and UCB based on ridge estimates (making them non-translation-invariant), and classical Thompson sampling, all with the same clipping. Specifically, Greedy maximizes the ridge estimate $S_{k,t}/(N_{k,t}+b_\n^2)$; UCB adds the bonus $1/\sqrt{N_{k,t}+b_\n^2}$, where $S_{k,t}$ and $N_{k,t}$ are cumulative rewards and pull counts. Classical Thompson sampling uses a $\mathcal{N}(0,b_\n^{-2})$ prior. All three use $b_\n=\sqrt{\n}/20$.

Excluding warm-up, cumulative regret is stable across common means for the RTI schemes: about $19$, $15$, and $20.5$, respectively. As $\mu$ increases from $0$ to $10$, regret rises from $21$ to $62$ for greedy, from $13$ to $55$ for UCB, and from $20$ to $35$ for classical Thompson sampling. Thus, translation invariance removes dependence on the common reward level, although RTI tempered-UCB incurs a modest regret cost at $\mu=0$.

\begin{figure}[!htb]
\centering
\centering
\includegraphics[width = \textwidth]{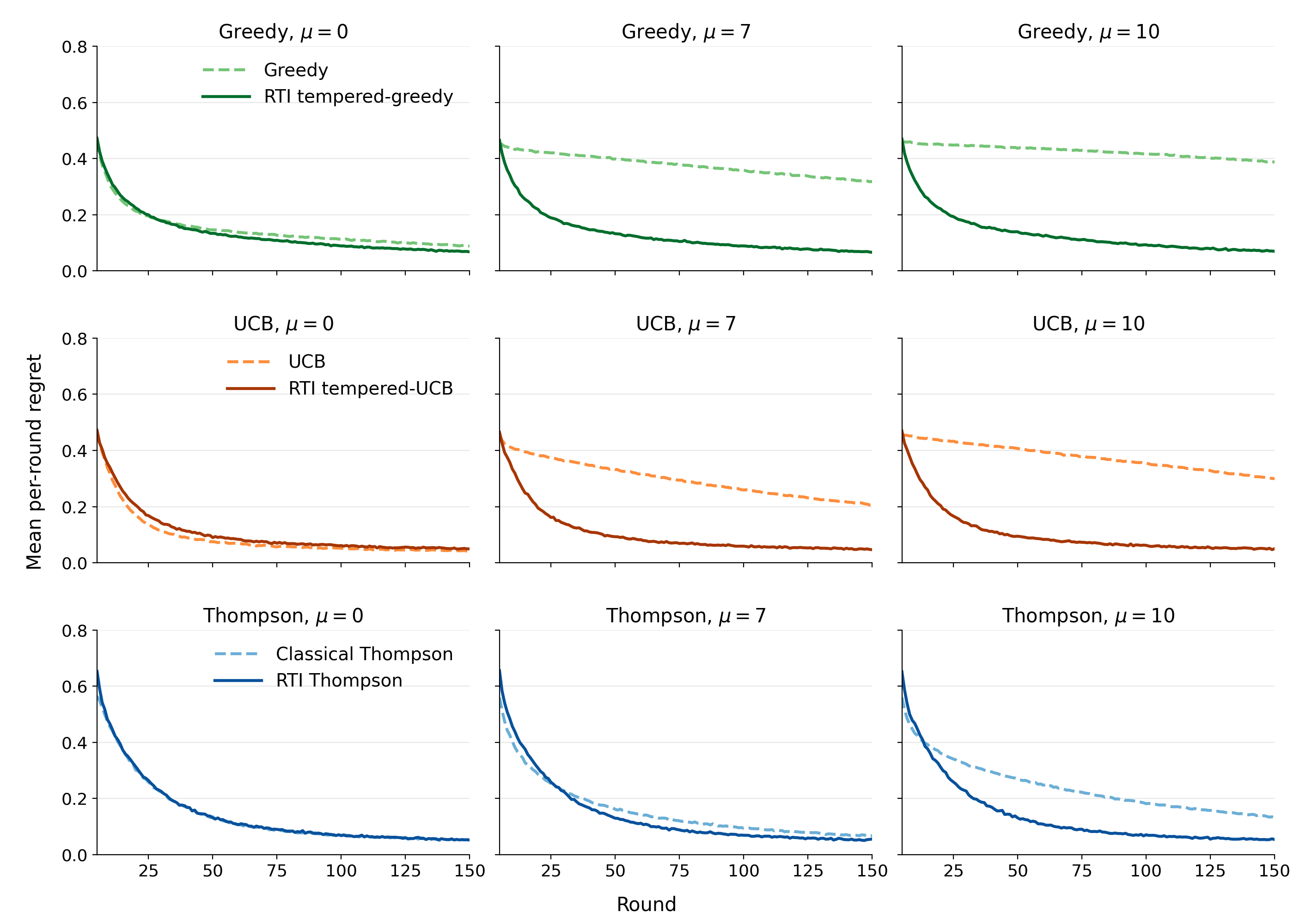}
\caption{\small Mean per-round regret for greedy, UCB, and Thompson sampling (top to bottom): RTI schemes (solid) and comparators (dashed). Four Gaussian arms have unit variance and means $(\mu+1,\mu,\mu,\mu)$, with $\mu=0,7,10$ (left to right). Results use $\n=150$ and $20{,}000$ replications.}
\label{S-fig:MABseq_regret_ThompsonInvK}
\end{figure}

\section{Finite-sample pivots}\label{S-appsec:finite_sample_pivots}
\noindent\emph{MAB pivot.}
With $N=\n-t_0$, $t_0=o(\n)$, and $\hat\mu_{\ii-1}$ the past pooled mean on Arms~1 and~2, the finite-sample pivot is
\begin{equation}\label{S-eqn:finite_MAB_pivot}
\stat_\n^{\text{TS-AW-P}}
=\frac1{\sqrt N}\sum_{\ii=t_0+1}^{\n}
 h_\ii\left(\frac{\indicator_{\{A_\ii=2\}}}{\policy_{2,\ii}}
 -\frac{\indicator_{\{A_\ii=1\}}}{\policy_{1,\ii}}\right)
 (\Y_\ii-\hat\mu_{\ii-1}),
\end{equation}
where $h_\ii=(\policy_{1,\ii}^{-1}+\policy_{2,\ii}^{-1})^{-1/2}$. Under the probability floor, unit error variance, and bounded conditional $2+\delta$ error moments, oracle-centered increments have conditional variance one and baseline estimation contributes $o_P(1)$. The martingale CLT gives a standard-normal null limit, including for non-translation-invariant policies.

\noindent\emph{CMAB pivot.}
Using the past pooled least-squares prediction
$\vX_\ii^\trans\hat\vbeta_{\ii-1}$ gives
\begin{equation}\label{S-eqn:finite_CMAB_pivot}
\stat_\n^{\text{TS-AW-P}}(\vG)
=\frac1{\sqrt N}\sum_{\ii=t_0+1}^{\n}
 h_{\vG,\ii}(\vG^\trans\vSigma_X^{-1}\vX_\ii)
 \left(\frac{\indicator_{\{A_\ii=2\}}}{\policy_{2,\ii}}
 -\frac{\indicator_{\{A_\ii=1\}}}{\policy_{1,\ii}}\right)
 (\Y_\ii-\vX_\ii^\trans\hat\vbeta_{\ii-1}).
\end{equation}
Here $h_{\vG,\ii}\defeq\big(\Exp\big[(\vG^\trans\vSigma_X^{-1}\vX)^2\{p_{1,\ii}(\vX)^{-1}+p_{2,\ii}(\vX)^{-1}\}\big]\big)^{-1/2}$ with the pre-round sampling probabilities $p_{k,\ii}(\vx)\defeq\policy_k^{(\n)}(\vS_{\ii-1},\vC_{\ii-1},\vx)$, the expectation integrating over the context law only. The simulations integrate over the known Gaussian context law by quadrature, using $t_0=K$ for MAB and $t_0=\max\{4,\lceil\sqrt\n\rceil\}$ for CMAB. If the context law is unknown, estimate $\vSigma_X$ and the expectation from past contexts, evaluating probabilities under the current policy; these predictable estimates must be uniformly consistent.

\noindent\emph{Omnibus pivot.}
Use the vector version of~\eqref{S-eqn:finite_MAB_pivot}, with a pooled baseline over all arms and inverse-square-root covariance weights, then take the squared norm.

For non-unit error variances, replace each $1/\policy_{k,\ii}$ in the
variance-normalizing weights by $\hat\sigma_{k,\ii-1}^{2}/\policy_{k,\ii}$, using consistent variance estimates based on past observations and allowing context dependence in CMAB.

\clearpage
\endgroup
\bibliographystyle{asa}
\bibliography{references}
\end{document}